\documentclass[onefignum,onetabnum]{siamonline220329}

\usepackage{hyperref}
\usepackage{xr-hyper}

\ifpdf
\hypersetup{
  pdftitle={Functional Multi-Reference Alignment via Deconvolution},
  pdfauthor={O. Al-Ghattas, A. Little, D. Sanz-Alonso, and M. Sweeney}
}
\fi

\usepackage{enumitem}
\usepackage{booktabs}
\usepackage{array}
\usepackage[normalem]{ulem}

\usepackage{lipsum}
\usepackage{amsfonts}
    \usepackage{mathtools}
    \usepackage{mathrsfs}
    \usepackage{bbm}
\usepackage{graphicx}
\usepackage{caption}
\usepackage{subcaption}
\usepackage{tikz}

\usepackage{epstopdf}
\usepackage[]{algorithm}
\usepackage{algorithmic}
\ifpdf
  \DeclareGraphicsExtensions{.eps,.pdf,.png,.jpg}
\else
  \DeclareGraphicsExtensions{.eps}
\fi

\makeatletter
\newcommand*{\addFileDependency}[1]{% argument=file name and extension
  \typeout{(#1)}
  \@addtofilelist{#1}
  \IfFileExists{#1}{}{\typeout{No file #1.}}
}
\makeatother

\usepackage{enumitem}
 \setlist[enumerate]{leftmargin=.5in}
 \setlist[itemize]{leftmargin=.5in}

\setlist[enumerate,1]%
  {label=\textup{(\alph*)},
   ref=\theassumption\textup{(\alph*)},   % <- makes \label work
   leftmargin=*,       % indent nicely
   }

\newsiamremark{remark}{Remark}
\newsiamremark{hypothesis}{Hypothesis}
\newsiamthm{claim}{Claim}

\headers{Functional MRA via Deconvolution}{O. Al-Ghattas, A. Little, D. Sanz-Alonso, and M. Sweeney}

\title{Functional Multi-Reference Alignment via Deconvolution}
\author{
O. Al-Ghattas\thanks{Broad Institute of MIT and Harvard, Cambridge, MA (\email{oag@broadinstitute.org})} 
  \and  A. Little\thanks{Department of Mathematics, University of Utah, Salt Lake City, UT}
\and  D. Sanz-Alonso\thanks{Department of Statistics, University of Chicago, Chicago, IL}
\and M. Sweeney\footnotemark[2]
}

\usepackage{amsopn}

\providecommand{\mathbbm}{\mathbb} % In case we don't load bbm

\renewcommand{\phi}{\varphi}

\newcommand{\C}{\mathbbm{C}}
\newcommand{\N}{\mathbbm{N}}
\newcommand{\Z}{\mathbbm{Z}}

\definecolor{mygreen}{rgb}{0.13,0.55,0.13}

\usepackage[T1]{fontenc}
\usepackage[utf8]{inputenc}
\usepackage{graphicx}
\usepackage{placeins} 
\usepackage{amssymb}

\newtheorem{assumption}[theorem]{Assumption}

\newcommand{\R}{\mathbb{R}}

\DeclarePairedDelimiterX{\norm}[1]{\lVert}{\rVert}{#1}

\newcommand{\hatomega}{\widehat{\omega}}
\newcommand{\yft}{y^{\text{ft}}}
\newcommand{\etaft}{\eta^{\text{ft}}}
\newcommand{\fft}{f^{\text{ft}}}
\newcommand{\pft}{p^{\text{ft}}}

\newcommand{\Xft}{X^{\text{ft}}}
\newcommand{\mft}{m^{\text{ft}}}
\newcommand{\rft}{r^{\text{ft}}}
\newcommand{\kft}{k^{\text{ft}}}

\newcommand{\hatf}{\widehat{f}}

\newcommand{\hatr}{\widehat{r}}
\newcommand{\hatk}{\widehat{k}}

\newcommand{\hatpsi}{\widehat{\psi}}
\newcommand{\tildePsi}{\widetilde{\Psi}}
\newcommand{\hatPsi}{\widehat{\Psi}}

\newcommand{\hatfft}{\hatf^{\text{ft}}}

\newcommand{\Kft}{K^{\text{ft}}}
\newcommand{\E}{\mathbb{E}}

\newcommand{\inparen}[1]{\left(#1\right)}             %\inparen{x+y}  is (x+y)
\newcommand{\abs}[1]{\ensuremath{\left\lvert #1 \right\rvert}}

\newcommand{\hatxi}{\widehat{\xi}}

\renewcommand{\P}{\mathbb{P}}

\newcommand{\hinv}{h^{-1}}

\newcommand{\checkep}{\check{\epsilon}}
\newcommand{\thresh}{{\mathcal{T}}}

\newcommand{\hatP}{\widehat{P}}

\newcommand{\barf}{\bar{f}}
\newcommand{\barg}{\bar{g}}

\newcommand{\mcE}{\mathcal{E}}
\newcommand{\mcF}{\mathcal{F}}

\newcommand{\mcL}{\mathcal{L}}

\usepackage{multirow}

\usepackage{comment}

\usepackage[scr = esstix, cal = cm, frak=euler]{mathalfa}

\definecolor{mygreen}{rgb}{0.1,0.75,0.2}

\let\epsilon\varepsilon
\begin{document}

\maketitle

\begin{abstract}
This paper studies the multi-reference alignment (MRA) problem of estimating a signal function from shifted, noisy observations. Our functional formulation reveals a new connection between  MRA and deconvolution: the signal can be estimated from second-order statistics via Kotlarski's formula, an important identification result in deconvolution with replicated measurements. To design our MRA algorithms, we extend Kotlarski's formula to general dimension and study the estimation of signals with vanishing Fourier transform, thus also contributing to the deconvolution literature. We validate our deconvolution approach to MRA through both theory and numerical experiments.
\end{abstract}

% REQUIRED

% REQUIRED

\section{Introduction}
\subsection{Aim} This paper studies the functional multi-reference alignment (MRA) problem of estimating a signal function from shifted, noisy observations. We introduce two algorithms rooted in a new connection between MRA and deconvolution: in the Fourier domain, the signal can be estimated from second-order statistics via Kotlarski's formula, a well-known identification result in deconvolution with replicated measurements \cite{meister2007Deconvolution, li1998nonparametric,comte2015density, kurisu2022uniform}.  
In the process of developing our algorithms,
we extend Kotlarski's formula to multidimensional settings and investigate the estimation of signals with vanishing Fourier transform. To validate our deconvolution approach to MRA,
we analyze the estimation error in both Fourier and spatial domains, studying the sample complexity in terms of interpretable model parameters, such as the variance and the lengthscale of the observation noise. As illustrated in Figure \ref{fig:signalsintro}, our algorithms can successfully recover a wide range of signals from highly noisy observations.   

\begin{figure}%[H]
    \includegraphics[width=\textwidth]{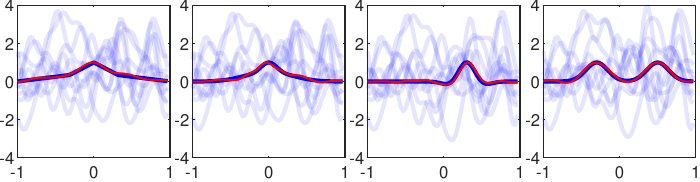}
    \caption{
    Performance of our algorithms in four illustrative examples. True signal (thick blue), representative observations (shadowed blue, drawn from a full sample of size $10^5$), and signal estimates (red).
    See Section \ref{sec:numerics} and Remark \ref{rmk:numerics_signalrec} for implementation details.}\label{fig:signalsintro}
\end{figure}

\subsection{Discrete and functional MRA models}
  MRA models arise in applications such as structural biology, radar, and image processing where an unknown signal is noisily observed after being altered by a group action \cite{diamond1992multiple, park2014assembly, park2011stochastic, sadler1992shift, scheres2005maximum, theobald2012optimal, brown1992survey, foroosh2002extension, zwart2003fast}.
  In particular, MRA arises as a simplified model for cryo-electron microscopy (cryo-EM), where the goal is to reconstruct a molecular structure
  that has been subject to random rotations and translations.
 
In the classical discrete formulation of the MRA problem, the goal is to recover an unknown signal \emph{vector} $f \in \mathbb{R}^J$ from $N$ noisy, randomly shifted observations:
\begin{align}
\label{equ:classic_discrete_MRA}
y_n[j] &= f\big[j - \zeta_n] + \eta_n[j], \qquad 1 \leq j \leq J,\quad 1 \leq n \leq N,
\end{align}
where translations of indices are defined modulo $J$. Each shift $\zeta_n$ is drawn independently from an unknown probability distribution on $\{0, 1, \ldots, J-1\}$, and the noise terms $\eta_1, \ldots, \eta_N$ are independent, centered, isotropic Gaussian vectors with variance $\sigma^2$.
While several discrete MRA models have appeared in the literature, the key unifying feature is that the signal is altered by a group action before being noisily observed. In \eqref{equ:classic_discrete_MRA}, the acting group is the cyclic group $\Z_J,$ which acts on $\R^J$ by circular shifts of the indices. More general group actions are considered for instance in \cite{balanov2025expectation,bandeira2020non}.

Closely related to the discrete MRA problem, and moving towards a functional formulation, one may consider estimating the Fourier coefficients $\fft[j],$ $1 \le j \le J,$ of a \emph{band-limited} signal function $f$ from Fourier domain observations of the form 
\[\yft_n[j] = e^{-i j \zeta_n} \, \fft[j] + \etaft_n[j], \qquad  1 \leq j \leq J, \qquad 1 \leq n \leq N.\]
In this phase-shift model \cite{bandeira2020optimal, dou2024rates}, the signal is translated in physical space prior to observation, which manifests as a multiplicative phase factor in the Fourier domain.

Inspired by the discrete and phase shift models, this paper studies a \emph{functional MRA} model in which the goal is to recover an unknown real-valued signal \emph{function} $f$ on a domain $D \subset \R^d$ from $N$ noisy, randomly shifted observations:   
\begin{align}\label{eq:MRADataGeneratingModelintro}
    y_n(t) = f(t-\zeta_n) + \eta_n(t), \qquad t \in D, \qquad 1 \le n \le N.
\end{align}
Here, each shift $\zeta_n$ determines a spatial translation of the signal and is drawn from an unknown distribution with compact support; the noise terms $\eta_1,\ldots, \eta_N$ are independent, centered Gaussian processes on $D$ with given covariance function. We additionally assume that the signal $f:\R^d\to\R$ is supported on a compact subset of $D$ and that all shifted functions $t \in \R^d \mapsto f(t-\zeta_n)$ are also supported on $D.$ 
This set-up provides a simplified model for applications where interest lies in identifying a compactly supported signal that has been subject to compactly supported random spatial translations; on the other hand, cyclic shifts as in \eqref{equ:classic_discrete_MRA} provide a simplified one-dimensional model for rotations. Both translations and rotations arise in numerous applications in structural biology, radar, and imaging, including cryo-EM. The simple model \eqref{eq:MRADataGeneratingModelintro} retains the key feature of any MRA model: prior to noisy observation, the signal is altered by an unknown element of a group. In \eqref{eq:MRADataGeneratingModelintro}, the acting group is the group of spatial translations in $\R^d$. As in the discrete setting, other group actions can be considered but are beyond the scope of this work.

\subsection{Comparison between discrete and functional models}\label{ssec:comparisons}
We next highlight important differences and similarities between discrete and functional formulations of MRA. 

\subsubsection{Signal modeling}
In many applications that motivate MRA, such as cryo-EM, the latent signal represents a physical object, and is hence more naturally modeled as function rather than a finite-dimensional vector. In such applications, it is important to understand how properties of the noise affect sample complexity independently of the discretization level $J$, 
and, as shown in this work, such understanding is facilitated by a functional formulation. 

The discrete and functional MRA models \eqref{equ:classic_discrete_MRA} and \eqref{eq:MRADataGeneratingModelintro} also differ in the assumptions placed on the support of the signal: no constraints are imposed in the discrete model, whereas the functional model assumes compact support. While more restrictive, the additional support constraint in functional MRA is natural in applications where the signal represents an object subject to random translations in a constrained physical environment.

\subsubsection{Shift modeling}
\label{sec:shifts_modeling}
In the discrete MRA model \eqref{equ:classic_discrete_MRA}, shifts are cyclic in $\{0,\ldots, J-1\}$, while in the functional MRA model \eqref{eq:MRADataGeneratingModelintro} shifts are noncyclic and the underlying domain is $\R^d$; in both cases, shifts specify a group action. In this aspect discrete MRA provides more flexibility for shift modeling, since by restricting the support of the hidden signal and shift distribution to avoid wrap around, noncyclic shifts can always be obtained as a special case of cyclic shifts. However, noncyclic shift modeling is natural for spatial translations in biological applications, as these translations are inherently not cyclic. 

Notice also that discrete MRA assumes that the shifts occur exactly on the discretization grid, i.e. shifts are modeled by a \text{finite} group whose cardinality agrees with the signal length $J$. In contrast, in functional MRA shifts follow a continuous distribution. While discretization is required for practical implementation, in the functional setting it can be carried out after the signals have been altered by the action of a continuous \textit{infinite} group, as described in Subsection \ref{sec:comp_to_spectral}.
This is important, since in many applications motivating MRA, such as cryo-EM, the latent transformations are governed by continuous symmetries, and, as we will show in Subsection \ref{sec:comp_to_spectral}, standard recovery algorithms struggle when the finite group assumption of discrete MRA is violated.

\subsubsection{Noise modeling}
The white noise assumption in~\eqref{equ:classic_discrete_MRA} forces the noise to blow up as $J \to \infty$, which in turn necessitates that the noise level $\sigma \to 0$ as $J\to\infty$ for recovery to remain possible~\cite{romanov2021multi, dou2024rates}. Beyond this issue, white noise is often an unrealistic modeling choice. In many applications, including cryo-EM, the noise is known to exhibit spatial correlations~\cite{hazon2022noise, bejjanki2017noise, huang2009noise}. As resolution increases, the assumption that pixelwise noise is uncorrelated becomes increasingly implausible; consequently, modeling the noise as in \eqref{eq:MRADataGeneratingModelintro} using a Gaussian process with short-range correlations (small lengthscale) may offer a more accurate representation.  

In order to analyze the impact of random dilations, the papers \cite{yin2024bispectrum, hirn2023power,hirn2021wavelet} consider a functional model with white noise; however, theoretical guarantees are derived only for the relevant Fourier invariants, and inversion is done with existing discrete algorithms which lack theoretical guarantees in the functional context.

\subsection{Algorithms and recovery guarantees for MRA}
 Numerous algorithmic strategies have been developed for discrete MRA. When the noise level is moderate, \textit{synchronization}-based methods ---which estimate the shifts and then align the observations--- are often effective~\cite{bandeira2020non, bandeira2023estimation, bandeira2016low, bandeira2014multireference, boumal2016nonconvex, chen2018projected, chen2014near, perry2018message, singer2011angular, zhong2018near}. However, synchronization fails in the high-noise regime, motivating alternative approaches such as expectation-maximization (EM)~\cite{abbe2018multireference, dempster1977maximum, balanov2025expectation}, the method of moments~\cite{hansen1982large, kam1980reconstruction, sharon2019method, abas2022generalized}, and the method of invariants~\cite{bandeira2020non, bendory2017bispectrum, collis1998higher, hirn2023power, hirn2021wavelet}.
If the shift distribution is \textit{aperiodic}, the signal $f$ can be recovered from the first and second empirical moments using a spectral method, with sample complexity $N \gtrsim \sigma^4$~\cite{abbe2018multireference}. In contrast, if the shift distribution is \textit{periodic}  ---e.g., uniform on $\{0, \ldots, J-1\}$  ---higher-order moments are typically required. A common approach in this case is to estimate the bispectrum of $f$ (which requires $N \gtrsim \sigma^6$ samples), and then invert it; see~\cite{perry2019sample} for theoretical guarantees using Jennrich’s algorithm. Recent work~\cite{dou2024rates} provides precise error bounds for bispectrum inversion with explicit dependence on signal length and frequency structure. The high-dimensional setting has also been studied: for instance,~\cite{romanov2021multi} shows that for Gaussian signals, a phase transition in sample complexity occurs at the threshold $J/(\sigma^2 \log J)=2.$ There are important exceptions where second moments are sufficient for recovery, including when the signal is sparse \cite{bendory2024sample, ghosh2023sparse, ghosh2023minimax} or, more generally, when it lies in a semi-algebraic set of sufficiently small dimension (for example a linear subspace) \cite{bendory2025transversality}.

While several algorithms have been proposed to solve discrete MRA from third moments, to the best of our knowledge only the spectral algorithm in \cite{abbe2018multireference} solves discrete MRA from second moments without additional structural assumptions as discussed in \cite{bendory2025transversality, bendory2024sample, ghosh2023sparse, ghosh2023minimax}, thus requiring $N \gtrsim \sigma^4$ samples instead of $N \gtrsim \sigma^6$. However, as we demonstrate in Subsection \ref{sec:comp_to_spectral}, the performance of this spectral algorithm deteriorates as (1) the noise level increases, (2) the signal length increases, or (3) the finite group action assumption discussed in Subsection \ref{sec:shifts_modeling} is violated. Inspired by the functional formulation, we propose an algorithm for recovery from second moments which is stable under all of the above scenarios. In particular, even when data is generated exactly from a discrete MRA model with isotropic Gaussian noise (with support constraints ensuring the equivalence of cyclic and noncyclic shifts), our proposed algorithm significantly improves upon the spectral algorithm in empirical evaluations. 

Our methodology builds on a novel connection between functional MRA and the classical statistical problem of deconvolution from replicated measurements \cite{meister2007Deconvolution, li1998nonparametric,comte2015density, kurisu2022uniform}, wherein one seeks to identify the distribution of a latent variable from repeated, conditionally independent noisy observations ---possibly under unknown noise. This connection enables us to leverage established inferential techniques, most notably Kotlarski’s identity \cite{rao1992identifiability, kotlarski1967characterizing}, which guarantees identifiability of the latent distribution from the joint law of the noisy replicates. This connection additionally enables us to establish recovery rates for functional MRA under natural smoothness assumptions.

\subsection{Outline and main contributions}
This paper is organized as follows:
\begin{itemize}
    \item Section \ref{sec:settingandalgorithm} formalizes our functional MRA model, introduces our main algorithm, and discusses the novel connection between MRA and deconvolution that underpins our work. Furthermore, we extend Kotlarski's formula to arbitrary dimension, which is useful for both MRA (Theorem \ref{thm:kotlarski}) and deconvolution (Corollary \ref{cor:MultivariateKotlarskiDeconvolution}). 
    \item Section \ref{sec:theory} analyzes our main algorithm. Theorem \ref{thm:ErrorBoundDeconvolutionEstimator} provides the first recovery guarantees for functional MRA for general (i.e. non-bandlimited) functions. In addition,  we characterize the sample complexity in terms of the signal-to-noise ratio (Corollary \ref{corr:highNoiseRegime}) and interpretable model parameters, such as the marginal variance and correlation lengthscale of the noise distribution (see Corollaries \ref{corr:largesigma} and \ref{corr:smallLengthscale}). 
    \item Section \ref{sec:vanishing} extends our algorithm and analysis to the case where $\fft$ may vanish (see Theorem \ref{thm:ErrorBoundDeconvolutionEstimatorVanishingFT} and Corollary \ref{corr:highNoiseRegimeVanishing}).  This removes a limiting but nonetheless standard assumption in the MRA literature, i.e. that the hidden signal has frequencies given by a finite-dimensional vector with non-vanishing entries.
    \item  Section \ref{sec:numerics} shows the performance of our algorithms in several computed examples  and also compares with the spectral algorithm proposed in \cite{abbe2018multireference}. 
\end{itemize}

\section{MRA: functional setting and deconvolution algorithm}\label{sec:settingandalgorithm}

\subsection{Problem setting}\label{ssec:problemsetting}
Let $D := [-1,1]^d$ and let $f : D \to \R$ be a square-integrable function supported on $\bigl[-\frac{1}{2}, \frac{1}{2}\bigr]^d$. Our goal is to estimate the signal $f$ from $N$ independent observations $\{y_n\}_{n=1}^N$  of the form
\begin{align}\label{eq:MRADataGeneratingModel}
    y_n(t) = f(t-\zeta_n) + \eta_n(t), \qquad t \in D , \qquad n=1, \ldots, N,
\end{align}
 with the convention that $f(x) =0$ for any $x \notin D.$  Here, $\{\zeta_n\}_{n=1}^N$ are independent copies of a centered random vector $\zeta$ with density $p_{\zeta}$ supported on $\bigl[-\frac{1}{2}, \frac{1}{2}\bigr]^d,$ and $\{\eta_n\}_{n=1}^N$ are independent copies of a centered square-integrable Gaussian process $\eta$ on $D$ with known covariance function $k_\eta.$ Thus, each observation $y_n$ represents a shifted, noisy measurement of the signal $f.$ We assume that the shifts $\{\zeta_n\}_{n=1}^N$ and measurement errors $\{\eta_n\}_{n=1}^N$ are all independent. We will let $y(t) = f(t-\zeta) + \eta(t)$ denote a generic observation.
Note that we do not assume periodicity of the signal $f,$ nor cyclic shifts; due to the compact support of the signal and shifts, all observations are supported in the domain $D.$\footnote{In our numerical examples, we use a slightly different convention and take $D =[-2,2]^d$ and the signal and shifts to be supported on $[-1,1]^d.$ See Figure \ref{fig:signalsintro} for illustrative signals and observations in dimension $d =1$.} 

\subsection{Identifiability via second-order statistics in the Fourier domain}
We will work in the Fourier domain, where the observations take the form
\begin{align}\label{eq:MRADataGeneratingModelFourier} 
    \yft_n(\omega) &= e^{-i \omega \cdot \zeta_n} \fft(\omega)+\etaft_n(\omega), \qquad \omega \in \R^d, \qquad n=1,\dots, N.
\end{align} 
Here, $f^\text{ft}(\omega) := \int_{\R^d} f(t) e^{-it \cdot \omega} \, dt$ and $\etaft_n(\omega):= \int_D \eta_n(t) e^{-it\cdot \omega} \, dt$ denote the Fourier transforms of $f$ and $\eta_n$. As before, we denote by $\yft$ and $\etaft$ a generic observation and measurement error in Fourier domain. 

Our deconvolution approach to MRA is based on the fact that the autocorrelation function $r_{\yft}(\omega_1,\omega_2) : = \E \bigl[ \yft(\omega_1) \overline{\yft(\omega_2) }\bigr] $ of the noncentered process $\yft$
satisfies
\begin{align*}\label{eq:identity}
    r_{\yft}(\omega_1,\omega_2)
    &= \fft(\omega_1) \overline{\fft(\omega_2)} \pft_{\zeta}(\omega_1-\omega_2)
    +
    k_{\etaft} (\omega_1, \omega_2)\\
    &=: \Psi(\omega_1,\omega_2) + k_{\etaft} (\omega_1, \omega_2), \qquad \omega_1, \omega_2 \in \R^d,
\end{align*}
where $k_{\etaft}$ is the covariance function of $\etaft$. Note that  since $f, p_\zeta$ have compact support, it follows that $\fft,\pft_{\zeta} \in C^\infty;$ consequently, we also have that $\Psi \in C^\infty.$  
The following identity shows that $\fft$ is completely determined by a nonlinear transformation of $\Psi.$ For $u\in \R^d, v \in \C^d$ we abuse notation slightly and write $u \cdot v$ to mean $u \cdot \Re(v) + i u \cdot \Im(v),$ where $\Re(v), \Im(v)$ are the real and imaginary parts of $v,$ respectively.

\begin{theorem}[Identification] 
\label{thm:kotlarski} 
    Let $p_{\zeta}\in L^2(D)$ be the probability density function of a centered random variable $\zeta$ with Fourier transform $\pft_{\zeta}$
    and let $f \in L^2(D)$ have Fourier transform $\fft$ satisfying $\fft(0)=1.$ 
    Define 
    \begin{equation}\label{eq:Kotlarski_second_mom}
        \Psi(\omega_1,\omega_2) 
        := \fft(\omega_1) \overline{\fft(\omega_2)}\pft_{\zeta}(\omega_1-\omega_2), 
        \qquad \omega_1, \omega_2 \in \R^d,
    \end{equation}
    and assume that $\Psi$ is nowhere vanishing on its diagonal (i.e. $\Psi(\omega, \omega) \neq 0$ for any $\omega$). Then, 
    \begin{align}\label{eq:formulaft}
    \fft(\omega) = 
    \exp \left ( 
    \int_0^1 \frac{\nabla_{1} \Psi(\alpha \omega, \alpha \omega)}{\Psi(\alpha \omega,\alpha \omega)} \cdot \omega ~ d \alpha 
    \right ),
    \end{align}
    where $\nabla_1$ indicates the gradient with respect to the first argument.
\end{theorem}

\begin{proof}
     The assumption that $\Psi$ is nowhere vanishing on its diagonal implies that $\fft$ is nowhere vanishing. Note that since $f \in L^2(D)$ and $D \subset \R^d$ is compact, we have that $\fft \in C^\infty(\R^d).$ Further, $\fft(0)=1$ by assumption. Therefore, Lemma~\ref{lem:technicalKotlarski} implies
    \begin{align}\label{eq:phi1lemma}
        \fft(\omega)= 
        \exp \inparen{
        \int_0^1  \frac{\nabla \fft(\alpha \omega)}{\fft(\alpha \omega)}  \cdot \omega \, d\alpha
        }, \quad \omega \in \R^d.
    \end{align}
    Noting that $\pft_\zeta \in C^\infty(\R^d),$ direct calculation shows that 
 \begin{align*}
     \frac{\nabla \fft (\alpha \omega)}{\fft (\alpha \omega)}
     =
     \frac{\nabla_1 \Psi(\alpha \omega, \alpha \omega)}{\Psi(\alpha \omega, \alpha \omega)}
     -
     \frac{\nabla \pft_{\zeta}(0)}{\pft_{\zeta}(0)}.
 \end{align*}
 Finally,  
 observe that $\nabla \pft_\zeta(0) = 0$ since $\zeta$ is centered by assumption.
\end{proof}
\begin{remark}\label{remarkFT1}
    The setting where $\fft$ is allowed to vanish (allowing $\Psi(\omega,\omega)$ to vanish as well) will be considered in  Section \ref{sec:vanishing}.
    We assume $\fft(0)=1$ only for ease of presentation: If $\fft(0)=c \neq 0,$ the constant $c$ can be estimated by $N^{-1} \sum_{n=1}^N \yft_n(0)$, and the data can then be rescaled accordingly. 
\end{remark}

\subsection{Main algorithm}\label{ssec:mainalgorithm}
We introduce a three-step approach to estimate the signal $f$:
\begin{enumerate}
    \item Estimate $\Psi(\omega,\omega),$ for instance by   $\hatPsi(\omega, \omega)
        := 
        \frac{1}{N} \sum_{n=1}^N \yft_n(\omega) \overline{\yft_n(\omega)}
        -
        \kft_\eta (\omega, \omega);$
    \item Estimate $\fft$ by plugging in the estimate of $\Psi$ into \eqref{eq:formulaft};
    \item Deconvolve the estimate of $\fft$ to obtain an estimate of $f$.
\end{enumerate}

There is significant flexibility in the implementation of each step. Algorithm \ref{alg:mainalgorithm} outlines the method that we will analyze in Section \ref{sec:theory} and test numerically in Section \ref{sec:numerics}. First, we estimate $\Psi$ using the sample moment estimator $\hatPsi,$ and also introduce a regularized estimator $\tildePsi.$   Second, we estimate $\fft$ by plugging in $\hatPsi$ and $\tildePsi$ in the numerator and denominator in \eqref{eq:formulaft}, respectively. Third, as is standard in the deconvolution literature, we estimate $f$ by regularizing the inverse Fourier transform of $\hatfft$ with a user-chosen kernel function.

\begin{algorithm}[htp]
\caption{\label{alg:mainalgorithm}MRA via Deconvolution}
\begin{algorithmic}[1]
     \STATE {\bf Input:} Data $\{y_n\}_{n=1}^N;$ Covariance function $k_{\eta};$ Kernel $K;$ Bandwidth parameter $h>0.$ 
    \STATE {\bf Estimate $\Psi$:} Set \vspace{-0.45cm}
    \begin{align*}
        \hatPsi(\omega, \omega)
        := 
        \frac{1}{N} \sum_{n=1}^N \yft_n(\omega) \overline{\yft_n(\omega)}
        -
        \kft_\eta (\omega, \omega), \qquad 
        \tildePsi(\omega,\omega):=  \frac{\hatPsi(\omega,\omega)}{1 \land \sqrt{N}|\hatPsi(\omega,\omega)|}. \vspace{-0.45cm}
    \end{align*}
    \STATE  \vspace{-0.45cm} {\bf Estimate $\fft$:} Set  \vspace{-0.45cm}
    \begin{align*}
       \hatfft(\omega)
        :=
        \exp \left ( 
            \int_0^1 \frac{\nabla_1 \hatPsi(\alpha \omega, \alpha \omega)}{\tildePsi(\alpha \omega,\alpha \omega)} \cdot \omega ~d\alpha
        \right) .
        \end{align*}
        \STATE  \vspace{-0.45cm}{\bf Deconvolve:} Set  \vspace{-0.45cm}
      \begin{align*}
        \hatf(t) := \frac{1}{(2\pi)^d} \int_{\R^d} e^{i \omega \cdot  t} \hatfft(\omega) \Kft(h \omega) \, d\omega.
        \end{align*}\vspace{-0.45cm}
    \STATE {\bf Output}: Approximation $\hatf$ to the hidden signal $f.$
\end{algorithmic}
\end{algorithm}

\subsection{Deconvolution perspective and Kotlarski's identity}\label{ssec:perspective}
Theorem~\ref{thm:kotlarski} reveals a fundamental connection between functional MRA and deconvolution. In this subsection, we elaborate on this connection by framing the theorem as a multivariate generalization of Kotlarski’s lemma~\cite{kotlarski1967characterizing, rao1992identifiability}. We begin by reviewing the \emph{replicated measurements framework}~\cite[Section 2.6.3]{meister2007Deconvolution}, a generalization of classical density deconvolution widely studied in econometrics.

In standard density deconvolution, the goal is to recover the density of a real-valued latent variable $X_0 \in \R$ from i.i.d. observations of $Z = X_0 + X_1$, where $X_0$ and $X_1$ are independent. Without knowledge of the distribution of $X_1$, there is lack of identifiability. This issue disappears in the presence of \emph{replicated measurements}, where we observe $N$ independent copies $\{(Z_{n,1}, Z_{n,2})\}_{n=1}^N$ of the random vector $(Z_1,Z_2)$ defined by
\begin{align*}
Z_{1} := X_{0} + X_{1}, \quad Z_{2} := X_{0} + X_{2},
\end{align*}
with $X_1$ and $X_2$ independent, identically distributed, and independent of $X_0$. The joint characteristic function of $(Z_1, Z_2)$ is then
\begin{align}\label{eq:psifunction}
\psi(\omega_1, \omega_2) 
= \E [\exp(i\omega_1 Z_1+ i\omega_2 Z_2)]
= \phi_1(\omega_1)\phi_1(\omega_2)\phi_0(\omega_1 + \omega_2),
\end{align}
where $\phi_0$ and $\phi_1$ denote the characteristic functions of $X_0$ and $X_1$, respectively. The key insight that allows us to connect MRA and deconvolution is that the function $\Psi$ defined in~\eqref{eq:Kotlarski_second_mom} has the same structure (up to a sign flip) as the function $\psi$ in ~\eqref{eq:psifunction}, with $\fft$ playing the role of $\phi_1$ and $\pft_\zeta$ playing the role of $\phi_0.$ 

Kotlarski’s identity~\cite[Lemma 1]{rao1992identifiability} provides a powerful identification result in the deconvolution setting (leveraged for instance in~\cite{li1998nonparametric, comte2015density,kurisu2022uniform}): if $\psi$ is nowhere vanishing, then the joint distribution of $(Z_1, Z_2)$ identifies the distributions of $X_0, X_1, X_2$ up to location shifts. Moreover, an explicit formula (see \cite[Remarks 2.1.11]{rao1992identifiability}) for the characteristic function of $X_0$ is given by
\begin{align}\label{eq:kotlarski1dA}
\phi_0(\omega) 
= 
\exp\left(i \omega \, \E[X_1] + \int_0^\omega \frac{\partial_1 \psi(0,\omega')}{\psi(0,\omega')} \, d\omega'\right), \quad \omega \in \R,
\end{align}
with $\phi_1$ recovered via
\begin{align}\label{eq:kotlarski1dB}
\phi_1(\omega) = \frac{\psi(\omega, 0)}{\phi_0(\omega)}. 
\end{align}
In settings where the target of estimation is $\phi_1$, \cite{evdokimov2010identification} proposes to use the direct expression
\begin{align}\label{eq:kotlarski1dEvdokimov}
\phi_1(\omega) 
= 
\exp\left(i \omega \, \E[X_0] + \int_0^\omega \frac{\partial_1 \psi(\omega',-\omega')}{\psi(\omega,-\omega')} \, d\omega'\right), \quad \omega \in \R.
\end{align}
It is natural to assume that $X_1,X_2$ are centered as they are most often interpreted as random noise. We can then estimate the density of $X_0$ (resp. $X_1$) as follows:
\begin{enumerate}
    \item Estimate the joint characteristic function $\psi,$ e.g.~by the sample version $\hatpsi(\omega_1,\omega_2) := N^{-1}\sum_{n=1}^N \exp( i \omega_1 Z_{n,1}+i \omega_2 Z_{n,2} ).$
    \item Estimate $\phi_0$ (resp. $\phi_1$) by plugging in $\hatpsi$ into \eqref{eq:kotlarski1dA} (resp. \eqref{eq:kotlarski1dB} or \eqref{eq:kotlarski1dEvdokimov}).
    \item Deconvolve the estimate of $\phi_0$ (resp. $\phi_1$) to  estimate the density of $X_0$ (resp. $X_1$).
\end{enumerate}
Notice that this three-step procedure is completely analogous to our deconvolution approach to MRA in Subsection \ref{ssec:mainalgorithm}.
 
Deconvolution via Kotlarki's identity in the multivariate setting is, however, less well developed. 
\cite[Theorem 3.1.1]{rao1992identifiability} extends  Kotlarski's identification result to random elements in a separable Hilbert space, but without providing an explicit estimator. In the context of errors-in-variables regression,~\cite{schennach2004estimation} derives several Kotlarski-type results in the univariate case and notes ---without proof--- that multivariate extensions may be obtained via path integrals~\cite[Footnote 11, p. 49]{schennach2004estimation}.

Theorem~\ref{thm:kotlarski} and its proof address this gap, delivering a constructive multivariate generalization of Kotlarski’s lemma. A slight adaptation of the proof yields the following corollary:

\begin{corollary}[Multivariate Kotlarski]
\label{cor:MultivariateKotlarskiDeconvolution}
Let $X_0, X_1, X_2$ be independent integrable random vectors in $\R^d$, and define
\begin{align*}
Z_1 := X_0 + X_1, \quad Z_2 := X_0 + X_2.
\end{align*}
Let $\phi_k(\omega) := \E[e^{i \omega \cdot X_k}]$ for $k = 0, 1, 2$, and define
\begin{align*}
\psi(\omega_1, \omega_2) := \E[e^{i \omega_1 \cdot Z_1 + i \omega_2 \cdot Z_2}] 
= \phi_0(\omega_1 + \omega_2)\phi_1(\omega_1)\phi_2(\omega_2), \quad \omega_1, \omega_2 \in \R^d.
\end{align*}
If $\psi$ is nowhere vanishing, then
\begin{align}\label{eq:multviarateKotlarski1}
\phi_0(\omega) 
= 
\exp\left( -i \omega \cdot \E[X_1] + \int_0^1 \frac{\nabla_1 \psi(0, \alpha \omega)}{\psi(0, \alpha \omega)} \cdot \omega \, d\alpha \right), \quad \omega \in \R^d.
\end{align}
\end{corollary}

In the univariate case ($d = 1$), the change of variable $\omega' = \alpha \omega$ in the integral component of \eqref{eq:multviarateKotlarski1} recovers the classical identity in \eqref{eq:kotlarski1dA}. An identical argument can also be used to establish the multivariate version of \eqref{eq:kotlarski1dEvdokimov}, \begin{align}\label{eq:kotlarskiGenDwithCF}
    \phi_1(\omega) 
    &= 
    \exp\left(-i \omega \cdot \E X_0 +\int_0^1 \frac{ \nabla_1 \psi(\alpha \omega,-\alpha \omega) }{\psi(\alpha \omega,-\alpha \omega)} \cdot \omega \, d\alpha\right)\, ,\qquad \omega \in \R^d.
    \end{align}

\section{Theoretical guarantees}\label{sec:theory}
This section analyzes Algorithm \ref{alg:mainalgorithm}. We show a general error bound in Subsection \ref{sec:analysisgeneral} and discuss connections with the MRA literature in Subsection \ref{subsec:snr}. We utilize the notation $\lesssim (\asymp)$ to indicate inequality (resp. equality) up to constants, independent of the parameters of interest.

\subsection{Error bound for Algorithm \ref{alg:mainalgorithm}}\label{sec:analysisgeneral}
In this subsection, we derive a high-probability bound on the deviation $\|\hatf - f\|_\infty.$ To prove our main result, Theorem \ref{thm:ErrorBoundDeconvolutionEstimator}, we analyze the error in the estimation of $\Psi$ and $\nabla_1\Psi$ in Lemma \ref{lem:psiestimates}, and the signal recovery in Fourier domain in Theorem \ref{thm:FTDeviation}. These intermediate results are proved in the Supplementary Material \ref{sec:appendixtheory}. 

Throughout, we will work with the following standing assumption:
\begin{assumption}[Standing assumption]\label{assumption:psiestimates} 
     Let $ H^s(D)$ be a Sobolev space of order $ s > d/2$ on the domain $D = [-1,1]^d.$ 
    \begin{enumerate}[label=(\roman*)]
        \item The signal $f$ satisfies $f \in H^s(D)$ 
        and is supported on $[-1/2,1/2]^d$ with $\fft(0)=1$.
        \item The noise $\eta$ is a centered Gaussian process taking values in  $H^s(D).$ 
        \item The shift $\zeta$ is  a centered random variable supported on $[-1/2,1/2]^d.$
    \end{enumerate} 
\end{assumption}
As noted in Remark \ref{remarkFT1}, the assumption that $\fft(0)=1$ is only made for ease of presentation. 

\begin{lemma}[Concentration of $\hatPsi$]\label{lem:psiestimates}
Let 
    \begin{align*}
        \hatPsi(\omega,\omega) 
        := \frac{1}{N} \sum_{n=1}^N \yft_n(\omega) \overline{\yft_n(\omega)} - k_{\etaft}(\omega, \omega) \quad \text{and} \quad   \widehat{\nabla_1\Psi}(\omega,\omega):=\nabla_1 \hatPsi(\omega,\omega).
    \end{align*}
     There exist positive universal constants $c_1,c_2$ such that, for any $\tau \ge 1,$ it holds with probability at least $1-c_1 e^{-c_2 \tau}$ that
    \begin{align*}
        \sup_{\omega \in \R^d} 
        |\hatPsi(\omega, \omega) -\Psi(\omega, \omega) | 
        &\lesssim \rho_N \quad \text{and} \quad 
        \sup_{\omega \in \R^d}
        \|\nabla_1\hatPsi(\omega, \omega) 
        - \nabla_1 \Psi(\omega, \omega) \|_2
        \lesssim \rho_N,
    \end{align*}
    where, with $r_{y}(t,t') := \E \bigl[ y(t) y(t') \bigr] $ denoting the autocorrelation function of $y$, we have
    \begin{align}
        \label{equ:def_of_rhoN}
        \rho_N :=
        \|r_{y}\|_\infty
        \left (
        \sqrt{\frac{\tau}{N}}
        \lor 
        \frac{\tau}{N}
        \lor 
        \frac{\Gamma(y)}{\sqrt{N}}
        \lor 
        \frac{\Gamma(y)^2}{N}
        \right),
        \qquad \Gamma(y) := \frac{\E[\sup_{t\in D} y(t)]}{\|r_{y}\|_\infty^{1/2}}.
    \end{align}
\end{lemma}

Next, we analyze the recovery of the signal in Fourier domain.  
The following assumption generalizes to a multivariate setting the commonly used \emph{ordinary smoothness} assumption in the deconvolution literature; see for example \cite{Delaigle2008repeated,fan1991optimal,li1998nonparametric, comte2015density,kurisu2022uniform}. We make explicit the dependence on $\|f\|_\infty,$ as this quantity will appear in our notion of signal-to-noise ratio in Subsection \ref{subsec:snr}.

\begin{assumption}[Smoothness]\label{assumption:smoothness}
There exists a constant $\beta>0$ and positive universal constants $C\ge c$ such that, for all $\omega\in\R^d,$
\[
c\|f\|_\infty (1+\|\omega\|_2)^{-\beta}
\le
|\fft(\omega)|
\le
C\|f\|_\infty (1+\|\omega\|_2)^{-\beta}.
\]
\end{assumption}

In Section~\ref{sec:vanishing}, we will show that 
the non-vanishing assumption for $\fft$ can be relaxed, but doing so requires a more intricate procedure and theoretical analysis. We remark further that, in principle, all our results extend to alternative formulations of \textit{super smoothness} considered in the aforementioned literature; however, we do not investigate this extension in this work.

\begin{theorem}[Recovery in Fourier domain]\label{thm:FTDeviation}
Under Assumption~\ref{assumption:smoothness} with smoothness parameter $\beta$, let $h\in(0,1)$ satisfy $\rho_N \le c_0 \|f\|_\infty^2 h^{3\beta+1}$ for a sufficiently small universal constant $c_0>0$. Then, there exist positive universal constants $c_1,c_2$ such that for any $\tau\ge 1$,
with probability at least $1-c_1e^{-c_2\tau}$,
\[
\sup_{\omega\in[-h^{-1},h^{-1}]^d}
\bigl|\hatfft(\omega)-\fft(\omega)\bigr|
\lesssim
\frac{\rho_N h^{-2\beta-1}}{\|f\|_\infty}.
\]
If in addition $\sup_{\omega \in [-h^{-1},h^{-1}]^d} \frac{\|\nabla \fft(\omega) \|_2}{|\fft(\omega)|}\lesssim 1$ and $\rho_N \le c_0 \|f\|_\infty^2 h^{2\beta+1},$ then the same conclusion holds with the improved bound
\[
\sup_{\omega\in[-h^{-1},h^{-1}]^d}
\bigl|\hatfft(\omega)-\fft(\omega)\bigr|
\lesssim
\frac{\rho_N h^{-\beta-1}}{\|f\|_\infty}.
\]
\end{theorem}

Finally, to analyze the recovery error in the spatial domain, we place the following assumption on the kernel:
\begin{assumption}[Deconvolution kernel]\label{assumption:kernel}
For a positive even integer $p$, the kernel $K:\R^d \to \R$ satisfies:
\begin{enumerate}[label=(\roman*)]
    \item  $  \int_{\R^d} K(t) \, dt = 1.$
    \item $ \int_{\R^d} t^\nu K(t) \, dt = 0 $ for all multi-indices $\nu \in \N^p$ with $1 \le |\nu| \le p-1.$
    \item $t \mapsto t^\nu K(t) \in L^1(\R^d)$ for all multi-indices $\nu \in \N^p$ with $ |\nu| = p.$
    \item $  \int_{\R^d} \|t\|_2^p\, |K(t)| \, dt \le C_K $ where $C_K$ if a finite constant allowed to depend on $K, d$, and $p.$ 
    \item $ \Kft(\omega) = 0 $ for all $ \|\omega\|_\infty > 1 .$  
\end{enumerate}
\end{assumption}
We give an example of a family of kernels satisfying Assumption \ref{assumption:kernel} in the Supplementary Material \ref{sec:appendixtheory}. We are now ready to state our main result concerning recovery in space. Our proof technique borrows from the deconvolution literature, see e.g. \cite[Lemma A.3]{kurisu2022uniform} and \cite{meister2007Deconvolution} for related one-dimensional analyses, which achieve similar rates for deconvolution problems.

\begin{theorem}[Recovery in spatial domain]\label{thm:ErrorBoundDeconvolutionEstimator}
    Assume the conditions of Theorem~\ref{thm:FTDeviation} hold. In addition, suppose that $f$ satisfies Assumption~\ref{assumption:smoothness} with smoothness parameter $\beta > d$ and let $K$ be a kernel satisfying Assumption \ref{assumption:kernel} with $p> \beta-d$ and bandwidth parameter $h$. Then, there exist positive universal constants $c_1,c_2$ such that, for all $\tau \ge 1,$ it holds with probability at least $1 - c_1 e^{-c_2 \tau}$ that the estimator $\hatf(t)$ in Algorithm \ref{alg:mainalgorithm} satisfies 
    \begin{align*}
        \|\hatf - f\|_\infty 
        \lesssim
        \rho_N \|f\|_\infty^{-1} h^{-2\beta - d - 1} + \|f\|_\infty h^{\beta-d},
    \end{align*}
    where $\rho_N$ is defined in \eqref{equ:def_of_rhoN}. 
    Choosing $h= h^* \asymp \inparen{ \frac{\rho_N}{\|f\|^2_\infty}}^{\frac{1}{3\beta + 1}}$ thus yields an estimator $\hatf$ with relative error satisfying
    \begin{align*}
        \frac{\|\hatf - f\|_\infty}{\|f\|_\infty} 
        \lesssim 
        \inparen{\frac{\rho_N}{\|f\|_\infty^2}}^{\frac{\beta - d}{3\beta + 1}}. 
    \end{align*}
 \end{theorem}

\begin{proof}
Theorem~\ref{thm:FTDeviation} shows that the  event $\mathsf{E}_{h,\tau}$ on which $\sup_{\omega \in [-h^{-1}, h^{-1}]^d} |\hatfft(\omega)-\fft(\omega) | \lesssim \rho_N \|f\|_\infty^{-1} h^{-2\beta-1}$ satisfies $\P(\mathsf{E}_{h, \tau}) \ge 1 - c_1e^{-c_2 \tau}.$ Now, note that 
\begin{align*}
    \|\hatf - f\|_\infty 
    &=
    \sup_{t \in D}
    \left | 
    \frac{1}{(2\pi)^d} \int_{\R^d} e^{i\omega \cdot t} \hatfft(\omega) \Kft(h\omega) \,d\omega
    -
    \frac{1}{(2\pi)^d} \int_{\R^d} e^{i\omega \cdot t} \fft(\omega)  \,d\omega
    \right|\\
    &\hspace{-0.5cm}\le 
    \frac{1}{(2\pi)^d} \int_{\R^d}  |\Kft(h\omega)| |\hatfft(\omega)  - \fft(\omega)| 
    \,d\omega
    + 
    \frac{1}{(2\pi)^d} \int_{\R^d} |\fft(\omega)| |\Kft(h\omega)-1|  \, d\omega\\
    & \hspace{-0.5cm}=: \textbf{(I)} + \textbf{(II)}.
\end{align*}
Conditional on $\mathsf{E}_{h,\tau},$ we have
\begin{align*}
    \textbf{(I)} 
    &=
    \frac{1}{(2\pi)^d} \int_{[-h^{-1}, h^{-1}]^d}  |\Kft(h\omega)| |\hatfft(\omega)  - \fft(\omega)| d\omega
    \le \frac{C_{\Kft}}{(2\pi)^d}  \int_{[-h^{-1}, h^{-1}]^d}  |\hatfft(\omega)  - \fft(\omega)| d\omega\\
    &\le \frac{C_{\Kft}}{(2\pi)^d}  \times 2h^{-d} \sup_{\omega \in [-h^{-1},h^{-1}]^d}  |\hatfft(\omega)  - \fft(\omega)| 
    \lesssim \rho_N \|f\|_\infty^{-1} h^{-2\beta-d-1},
\end{align*}
where the first equality holds since
 $\Kft$ is supported on $[-1,1]^d$ by Assumption~\ref{assumption:kernel} (v),
 for the first inequality we use that  $\sup_{\omega\in \R^d}\Kft(\omega)\le C_{\Kft}$ for some positive finite constant $C_{\Kft}$ since $K$ has compact support.
 
 Next, for $\textbf{(II)}$, we have by Assumption~\ref{assumption:smoothness} that
\begin{align*}
    \textbf{(II)} 
    &\lesssim \frac{\|f\|_\infty}{(2\pi)^d} \int (1+\| \omega\|_2)^{-\beta} |\Kft(h\omega)-1| d\omega\\
    &=\frac{\|f\|_\infty}{(2\pi)^d} 
    \inparen{ 
    \int_{[-h^{-1}, h^{-1}]^d}
    +
    \int_{\R^d \setminus [-h^{-1}, h^{-1}]^d}
    }
    (1+\| \omega\|_2)^{-\beta} |\Kft(h\omega)-1| d\omega 
    \\
    &\le \frac{\|f\|_\infty}{(2\pi)^d} 
    \int_{[-h^{-1}, h^{-1}]^d} (1+\| \omega\|_2)^{-\beta} |\Kft(h\omega)-1| d\omega + 
    \frac{\|f\|_\infty}{(2\pi)^d} 
    \int_{ \|\omega\|_2 > h^{-1} 
    }(1+\| \omega\|_2)^{-\beta} d\omega\\ 
    &=: \textbf{(A)} + \textbf{(B)}.
\end{align*}
For the integral over the complement, since $\omega \in \R^d \setminus [-h^{-1}, h^{-1}]^d$, we have $\|h\omega\|_\infty = h\|\omega \|_\infty > h h^{-1}=1,$ and so by Assumption~\ref{assumption:kernel} (v) $|\Kft(h\omega) -1|=|0-1|=1.$ The second inequality follows since $\R^d \setminus [-h^{-1}, h^{-1}]^d \subset \{\omega: \|\omega\|_2 > h^{-1} \}.$ To control $\textbf{(A)},$ we note that the $p$-th order remainder version of Taylor's expansion of $\Kft(t)$ about $t=0$ yields
\begin{align*}
    \Kft(h\omega) = 1 + \sum_{|\nu|=p}\frac{D^\nu \Kft(\xi)}{\nu !}  (h\omega)^\nu,
\end{align*}
where $\xi$ is a point on the line segment connecting the origin and $h\omega,$ and we have used that all partial derivatives up to order $p-1$ vanish by Assumption~\ref{assumption:kernel} (ii). Note that for any $\omega$, $D^\nu \Kft(\omega) = (-i)^{|\nu|} \int_{\R^d} t^\nu K(t) e^{-it \cdot \omega} dt$, which by Assumption~\ref{assumption:kernel} (iii) exists and is finite, and further is continuous in $\omega.$ Therefore, $D^\nu \Kft$ must be bounded on the compact set $[-1,1]^d.$ Further, for $|\nu|=p$, 
\begin{align*}
    |(h\omega)^\nu| 
    \le h^p \prod_{j=1}^d|\omega_j|^{\nu_j}
    \le h^p \prod_{j=1}^d\|\omega\|_2^{\nu_j}
    = h^p\|\omega\|_2^p.
\end{align*}
Therefore 
\begin{align*}
    \textbf{(A)} \lesssim  
    \|f\|_\infty h^p \int_{[-h^{-1}, h^{-1}]^d}  
    \|\omega\|_2^{p} (1+ \|\omega\|_2)^{-\beta} 
    d\omega.
\end{align*}
Using that $[-h^{-1}, h^{-1}]^d \subset B_d(\sqrt{d}h^{-1})$ where $B_d(R)$ is the $d$-dimensional Euclidean ball of radius $R>0,$ and switching to polar coordinates, we obtain
\begin{align*}
    \textbf{(A)} \lesssim  
    \|f\|_\infty h^p \int_{0}^{\sqrt{d}h^{-1}}  
    r^{p+d-1} (1+r)^{-\beta} dr.
\end{align*}
For $h\in(0,1)$, the integral satisfies
\begin{align*}
    \int_0^{\sqrt d\,h^{-1}}
    r^{p+d-1}(1+r)^{-\beta}\,dr
    \lesssim
    \begin{cases}
    h^{-(p+d-\beta)}, & \beta < p+d,\\
    \log(1/h), & \beta = p+d,\\
    1, & \beta > p+d.
    \end{cases}
\end{align*}
Since $p>\beta-d$ by assumption, we obtain $\textbf{(A)}
    \lesssim
    \|f\|_\infty h^p h^{-(p+d-\beta)}
    =
    \|f\|_\infty h^{\beta-d}.
    $

For $\textbf{(B)},$ provided that $\beta >d$, by switching to polar coordinates we have
\begin{align*}
    \textbf{(B)} 
    &\lesssim
    \|f\|_\infty 
    \int_{\|\omega\|_2 > h^{-1}} (1 + \|\omega\|_2)^{-\beta} d\omega
    \lesssim
    \|f\|_\infty 
    \int_{h^{-1}}^\infty r^{d-1 } (1+r)^{-\beta}dr
    \\
    &\lesssim
    \|f\|_\infty 
    \int_{h^{-1}}^\infty r^{d-\beta-1 } dr
    \lesssim \|f\|_\infty h^{\beta-d},
\end{align*}
where we used that since $h\in(0,1),$ $r \ge h^{-1} \ge 1$ on the domain of integration, and therefore $(1+r)^{-\beta}\lesssim r^{-\beta}.$

It follows that $\textbf{(II)} \lesssim\|f\|_\infty h^{\beta - d}.$ Combining the estimates for \(\mathbf{(I)}\) and \(\mathbf{(II)}\) proves the claim. 
\end{proof}

\subsection{Signal-to-noise ratio and model parameters}\label{subsec:snr}
    In much of the existing literature on MRA, results are formulated in terms of a \textit{signal-to-noise ratio}; see, for example, \cite{perry2019sample, abbe2018multireference, ghosh2023sparse}. While, as described in Subsection \ref{ssec:comparisons},  the functional MRA framework considered in this work differs significantly from these prior settings, we introduce a notion of signal-to-noise ratio $(\mathsf{snr})$ that is tailored to the functional version of the MRA problem addressed here. Namely, we define  
    \begin{align}
    \label{equ:snr}
        \mathsf{snr} := \frac{\|f\|_\infty^2 }{\|k_{\eta}\|_\infty^{1/2} \E[\sup_{t\in D} \eta(t)]}.
    \end{align}
  We now illustrate our results in the challenging \textit{high-noise} regime, where $ \|k_\eta\|_\infty^{1/2} \ge \|f\|_\infty $. Throughout this subsection, we adopt the framework of Theorem~\ref{thm:ErrorBoundDeconvolutionEstimator} and assume we are in the high-noise setting. Proofs for all results stated in this subsection are deferred to the Supplementary Material~\ref{app:highNoiseRegime}.

    \begin{corollary}[High-noise regime] \label{corr:highNoiseRegime} 
    Under the assumptions of Theorem \ref{thm:ErrorBoundDeconvolutionEstimator}, suppose that $\|k_\eta\|^{1/2}_\infty \ge \|f\|_\infty$ and $N \gtrsim \Gamma(y)^2 \lor \tau$; then, it holds with probability at least $1-c_1e^{-c_2 \tau}$ that
    \begin{align*}
        \frac{\|\hatf - f\|_\infty}{\|f\|_\infty}
        \lesssim 
        \inparen{
        \|k_\eta\|_{\infty}
        \sqrt{\frac{\tau}{N}}
        \lor 
        \frac{\mathsf{snr}^{-1}}{\sqrt{N}}
        }^{\frac{\beta - d}{3\beta + 1}}.
    \end{align*}
    \end{corollary}
    \begin{remark}\label{rmk:snr}
        Corollary~\ref{corr:highNoiseRegime} implies that, in the high-noise regime, it suffices to take $N \gtrsim \mathsf{snr}^{-2}$ samples to ensure recovery of the signal function in relative supremum norm. Note further that since $\beta > d$, then as $\beta \to \infty$ (increasing smoothness of the underlying signal) the bound becomes $(\mathsf{snr}^{-1}/\sqrt{N})^{1/3}.$ Finally, we remark that in the high-noise regime, it is necessarily true that $\mathsf{snr} < 1$ since $\E[\sup_{t \in D} \eta(t)] \ge \|f\|_\infty$ as shown in the proof of Corollary~\ref{corr:highNoiseRegime}. 
    \end{remark}
    If the covariance function $k_\eta$ of the noise process $\eta$ belongs to a parametric family, then, for a fixed signal $f,$ we may study the scaling of $\mathsf{snr}$ with the parameters of $k_\eta.$ For example,  suppose that $\eta$ has squared-exponential covariance function
    \begin{align}\label{eq:SEkernel}
        k_{\eta}^{\text{SE}}(t,t') = \sigma^2 \exp \inparen{-\frac{\|t-t'\|_2^2}{2 \lambda^2}}, \qquad t, t' \in D,
    \end{align}
    which satisfies Assumption \ref{assumption:psiestimates} with arbitrary Sobolev order $s.$  Then, the noise becomes stronger as the scale parameter $\sigma>0$ increases, and more spatially uncorrelated as the lengthscale parameter $\lambda>0$ decreases; if all other quantities are fixed, the high-noise regime holds provided that $\sigma$ is sufficiently large, or that $\lambda$ is sufficiently small.  The following two corollaries quantify the individual effect of large $\sigma$ and small $\lambda$ in the sample complexity.  In Corollary \ref{corr:largesigma} (resp. Corollary \ref{corr:smallLengthscale}) we treat $\|f\|_\infty$ and $\lambda$ (resp. $\|f\|_\infty$ and $\sigma$) as constants. Their proofs follow from Corollary \ref{corr:highNoiseRegime}, using the scalings of $\mathsf{snr}$ with $\sigma$ and $\lambda:$ namely, $\mathsf{snr}^{-1} \asymp \sigma^{2}$ and  $\mathsf{snr}^{-1} \asymp \sqrt{\log \lambda^{-d}}$ for sufficiently small $\lambda.$

        \begin{corollary}[High-noise, large-$\sigma$ regime]\label{corr:largesigma}
      For all $\sigma$ sufficiently large, the following holds.
        For any $\tau \ge 1$, if $N \ge \tau$, it holds with probability at least $1-c_1e^{-c_2 \tau}$ that
        \begin{align*}
        \frac{\|\hatf - f\|_\infty}{\|f\|_\infty}
        \lesssim 
        \inparen{
        \sigma^2 \sqrt{\frac{\tau}{N}}
        }^{\frac{\beta - d}{3\beta + 1}}.
    \end{align*}
    \end{corollary}

        \begin{corollary}[High-noise, small-$\lambda$ regime]\label{corr:smallLengthscale}
        For all $\lambda$ sufficiently small, the following holds.
        If $N \gtrsim \log \lambda^{-d},$ it holds with probability at least $1-\lambda^{d}$ that
        \begin{align*}
        \frac{\|\hatf - f\|_\infty}{\|f\|_\infty}
        \lesssim 
        \inparen{ \frac{ \log \lambda^{-d}}{N}}^{\frac{\beta - d}{6\beta + 2}}.
    \end{align*}
    \end{corollary}
    
    \begin{remark}
        Corollary~\ref{corr:largesigma} establishes that it is sufficient to choose $N\gtrsim \sigma^4$ 
        to guarantee accurate recovery of the underlying signal;
         this sample complexity matches that of discrete MRA with an \textit{aperiodic} shift distribution (note $p_\zeta$ is aperiodic due to the assumption of compact support). 
         However,  it is important to note that our analysis pertains to the continuous variant of the problem and employs a distinct error metric, namely the relative supremum norm error. From a modeling viewpoint, the  \textit{small-lengthscale} regime is of particular interest and not captured by previous finite-dimensional analyses of MRA: in the limit as $\lambda \to 0,$ $k_{\eta}^{\text{SE}}$ formally approximates a white noise model with no spatial correlations.
        Relevant to Corollary~\ref{corr:smallLengthscale}, the papers \cite{al2024optimal,al2025covariance} show that the lengthscale fundamentally governs the difficulty of estimating the covariance operator of a Gaussian process. In a similar spirit, Corollary~\ref{corr:smallLengthscale} shows that accurate recovery can be achieved with $N \gtrsim \log \lambda^{-d}$ samples. Similar conclusions hold for a more general family of isotropic, positive, differentiable, and strictly decreasing covariance functions that decay to zero at infinity, including Matérn covariance functions ---see \cite[Assumption 2.7]{al2025covariance}.
    \end{remark}

\section{Extension for vanishing Fourier transform}\label{sec:vanishing}
This section extends the procedure in Section \ref{sec:settingandalgorithm} and the theory in Section \ref{sec:theory} to the case where $\fft$ may vanish. 

\subsection{Algorithm: vanishing Fourier transform}
Motivated by Algorithm \ref{alg:mainalgorithm}, here we propose a four-step procedure with the following outline:
(a) Estimate $\Psi.$
    (b) Estimate the zeros of $\fft.$
    (c) Estimate $\fft$ by plugging the estimates obtained in (a) and (b) into a generalized Kotlarski's formula.
    (d) Deconvolve the estimate of $\fft$ to obtain an estimate of $f$.

We focus on the one-dimensional setting ($d=1$) and leverage a generalization of Kotlarski's formula derived in \cite{evdokimov2012some}. 
Specifically, let $h > 0$ be a bandwidth parameter and let $ k(\hinv)$ denote the number of zeros of $\fft$ on the interval $[0,\hinv]$, which we denote by $\xi_1, \ldots, \xi_{k(\hinv)}$.  Supposing that $k(\hinv) < \infty$ and that the zeros are simple (i.e. of multiplicity one), \cite{evdokimov2012some} shows that, for any fixed $\omega\in [0,\hinv]$,
\begin{align}\label{equ:GeneralizedKotlarski}
	\fft(\omega) &= \lim_{\epsilon \rightarrow 0}\, (-1)^{k(\omega)} \prod_{i=1}^{k(\omega)} \exp\left(\int_{\xi_{i-1}+\epsilon}^{\xi_i-\epsilon}\frac{\partial_{1} \Psi(\xi,\xi)}{\Psi(\xi,\xi)}d\xi \right)\exp\left(\int_{\xi_{k(\omega)}+\epsilon}^{\omega}\frac{\partial_{1} \Psi(\xi,\xi)}{\Psi(\xi,\xi)}d\xi \right) \nonumber \\
	&= \lim_{\epsilon \rightarrow 0}\,(-1)^{k(\omega)}\exp\left(\int_{\omega_\epsilon}\frac{\partial_{1} \Psi(\xi,\xi)}{\Psi(\xi,\xi)} d\xi \right),
\end{align}
where  $k(\omega)$ denotes the number of zeros on $[0,\omega],$ $\xi_0 = 0$ for convenience, and
\begin{align}
\label{equ:omega_tilde_eps}
 \omega_\epsilon &= [0,\omega] \setminus \bigcup_{i=1}^{k(\hinv)} (\xi_i-\epsilon, \xi_i + \epsilon) \, .  
\end{align}

Fixing a small $\epsilon$ in the right-hand side of Equation \ref{equ:GeneralizedKotlarski} would in principle yield a good approximation to $\fft(\omega)$ for $\omega \in [0, \hinv]$. However, two key challenges arise:

\begin{enumerate}
    \item The number $k(\omega)$ of zeros and their locations $\xi_i$ are unknown; and
    \item The values of the integrand $\frac{\partial_{1} \Psi(\xi,\xi)}{\Psi(\xi,\xi)}$ are not directly observable.
\end{enumerate}

\vspace{0.15cm} To address these challenges, we will estimate the zeros and the integrand from data. The idea is to  define, for a small fixed $\epsilon$ and $\omega \in [0, \hinv],$  the estimator
\begin{align}
\label{equ:GeneralizedKotlarskiEst}
\hatfft_\epsilon(\omega) &= (-1)^{\hatk(\omega)}\exp\left(\int_{\hatomega_\epsilon}\frac{\partial_{1} \hatPsi(\xi,\xi)}{\tildePsi(\xi,\xi)} d\xi \right),
\end{align}
where $\hatk(\omega)$ and $\hatomega_\epsilon$ are estimators of $k(\omega), \omega_\epsilon$, respectively, and the functions $\hatPsi$ and $\tildePsi$ are as defined in Algorithm~\ref{alg:mainalgorithm}.
We then extend the definition of $\hatfft_\epsilon(\omega)$ to $\omega \in [-\hinv,0]$ using conjugate symmetry, and set $\hatfft_\epsilon(\omega) = 0$ for $\omega \notin [-\hinv, \hinv]$.
Given the estimator $\hatfft_\epsilon$, we estimate the underlying signal via deconvolution as in Algorithm~\ref{alg:mainalgorithm}. We now turn to the problem of estimating the zeros of $\fft$, which are needed to compute both $\hatk(\omega)$ and $\hatomega_\epsilon.$ To this end, we let $Z_g$ denote the set of zeros of a function $g$, and $Z_g^{\omega} = Z_g \cap [0, \omega]$. Additionally, we let $P(\omega) = \Psi(\omega, \omega)=|\fft(\omega)|^2$ denote the power spectrum of the hidden signal, and let $\hatP(\omega) = \Re\hatPsi(\omega, \omega)$ be the real part of its empirical approximation as defined in Algorithm~\ref{alg:mainalgorithm}.
To construct \eqref{equ:GeneralizedKotlarskiEst}, we need to estimate $Z_{\fft}^{\hinv}$, the zeros of $\fft$ on the interval $[0, h^{-1}]$, or, equivalently, $Z_{P}^{\hinv}$. We exploit the observation that when the zeros of $\fft$ are isolated and simple, the corresponding zeros of $P'$ exhibit linear crossings, in contrast to the quadratic behavior of the zeros of $P$. This motivates a practical approach for estimating $Z_{\fft}^{\hinv}$: we identify the zeros of $\hatP'$ and retain only those for which the value of $\hatP$ is sufficiently small. Specifically, for a threshold $\thresh > 0$, we define the approximate zero set as
\begin{align}
\label{equ:approx_zeros_via_thresholding}
 \mathcal{Z}_\thresh^{\hinv} &:= \Bigl\{ \hatxi \in Z_{\hatP'} \, : \, \hatP(\hatxi) < \thresh \ , \hatxi \in [0, \hinv] \Bigr\} 
\end{align}
and let $\hatxi_1, \ldots, \hatxi_{\hatk(\hinv)}$ 
denote the elements of $\mathcal{Z}_\thresh^{\hinv}$ arranged in increasing order, with $\\\hatk(h^{-1}) := |\mathcal{Z}_\thresh^{\hinv}|$.
We then use the estimated zeros to define $\hatk(\omega)$ and $\hatomega_\epsilon$.  
The full procedure is summarized in Algorithm \ref{alg:zerosalgorithm}.

\begin{algorithm}[htp]
\caption{\label{alg:zerosalgorithm}MRA via Deconvolution (Vanishing Case, $d=1$)}
\begin{algorithmic}[1]
     \STATE {\bf Input:} Data $\{y_n\}_{n=1}^N;$ Covariance function $k_{\eta};$ Kernel $K$; Bandwidth parameter $h>0$; Windowing parameter $\epsilon$; Threshold $\thresh$.
    \STATE {\bf Estimate $\Psi$:} Set \vspace{-0.45cm}
    \begin{align*}
        \hatPsi(\omega, \omega)
        := 
        \frac{1}{N} \sum_{n=1}^N \yft_n(\omega) \overline{\yft_n(\omega)}
        -
        \kft_\eta (\omega, \omega), \qquad 
        \tildePsi(\omega,\omega):=  \frac{\hatPsi(\omega,\omega)}{1 \land \sqrt{N}|\hatPsi(\omega,\omega)|}.
    \end{align*}
    \STATE \vspace{-0.45cm} {\bf Estimate zeros:} Set $\hatP(\omega):= \Re \hatPsi(\omega, \omega)$,
    \begin{align*}
    \mathcal{Z}_\thresh^{h^{-1}} &:= \Bigl\{ \hatxi \in Z_{\hatP'} \, : \, \hatP(\hatxi) < \thresh \ , \hatxi \in [0, h^{-1}] \Bigr\} \, ,
    \qquad 
    \hatk(h^{-1}) := |\mathcal{Z}_\thresh^{\hinv}|,
    \end{align*}
    and let $\hatxi_1, \ldots, \hatxi_{\hatk(\hinv)}$ be the increasingly ordered elements of $\mathcal{Z}_\thresh^{h^{-1}}.$ 
    \STATE  {\bf Estimate $\fft$:} For $\omega \in [0,h^{-1}]$, set \vspace{-0.35cm}
  \begin{align*}
        \hatfft_\epsilon(\omega) &= (-1)^{\hatk(\omega)}\exp\left(\int_{\hatomega_\epsilon}\frac{\partial_{1} \hatPsi(\xi,\xi)}{\tildePsi(\xi,\xi)} d\xi \right), \, \quad
        \hatomega_\epsilon := [0,\omega] 
        \setminus \bigcup_{i=1}^{\hatk(\hinv)}(\hatxi_i-\epsilon,\hatxi_i+\epsilon) \, .
    \end{align*}
     For $\omega \in [-h^{-1},0]$, extend via conjugate symmetry. For $\omega \notin [-h^{-1},h^{-1}],$ set $\hatfft_\epsilon(\omega) = 0.$ 
        \STATE {\bf Deconvolve:} Set \vspace{-0.45cm}
      \begin{align*}
        \hatf(t) := \frac{1}{2\pi} \int_{\R} e^{i \omega  t} \hatfft_\epsilon(\omega) \Kft(h \omega) \, d\omega.
        \end{align*} \vspace{-0.45cm}
    \STATE {\bf Output}:  Approximation $\hatf$ to the hidden signal $f.$
\end{algorithmic}
\end{algorithm}

\subsection{Theoretical guarantees: vanishing Fourier transform}
 Here, we analyze Algorithm \ref{alg:zerosalgorithm}  under our standing Assumption \ref{assumption:psiestimates}  and the following  relaxed version of Assumption~\ref{assumption:smoothness}.  For notational brevity we will henceforth omit the superscript on the zero sets, with the assumption that all zero sets of functions are intersected with $[0,\hinv]$.

\begin{assumption}[Vanishing model]\label{ass:vanishing_model}
\begin{enumerate} 
\item\label{assumption:FT_decay} 
  (Decay with oscillation) 
  There exist constants $\beta\geq 1, C > 0$ and a potentially oscillating function $o_f$ such that, for all $\omega\in\R,$
	\begin{align*}
		\fft(\omega) &= o_f(\omega)\cdot (1+|\omega|)^{-\beta}\, \quad \text{with}\quad |o^{(\ell)}_f(\omega)| \le  C\|f\|_\infty \quad \text{for} \quad \ell =0, 1, 2.
	\end{align*}
    
  \item\label{assumption:zeros}(Isolated, simple zeros) 	
    The zeros of $\fft$ are isolated and simple; in particular, $\xi_1 \geq 1$, $|\xi_i - \xi_j|\geq 1$ for $\xi_i, \xi_j\in Z_{\fft}$ with $\xi_i \ne\xi_j$, and the zeros of $\fft$ and $(\fft)'$ are disjoint.
  \item\label{assumption:nice_oscillations}(Regularity of oscillations) 
  Let $|\omega-Z_{\fft}|$ be the distance from $\omega$ to the closest point in $Z_{\fft}$. Then, there exist constants $L$, $B$, and $\epsilon_{\max}$ such that, for all $\checkep \in [0, \epsilon_{\max}],$
	\begin{align*}
		|\omega - Z_{\fft} | \geq \checkep \quad &\implies \quad |o_f(\omega)| \geq L \|f\|_\infty \checkep\, ,\\
		|\omega - Z_{\fft} | \leq \checkep \quad &\implies \quad |o_f'(\omega)| \geq B \|f\|_\infty\, .
	\end{align*}
    Furthermore, the estimator $\hatfft$ in \eqref{equ:GeneralizedKotlarskiEst} is constructed with $\epsilon \leq\epsilon_{\max}$.
   \item\label{assumption:thresholding}(Thresholding) We fix an interval $[-\hinv, \hinv]$ satisfying $h^{2\beta}\leq \epsilon_{\max} \vee 1$, and assume the threshold $\thresh$ defining $\mathcal{Z}_\thresh^{\hinv}$ satisfies $\rho_N h^{-2\beta} \lesssim \thresh \lesssim \|f\|_\infty^2h^{6\beta}$ for $\rho_N$ defined in \eqref{equ:def_of_rhoN}. In addition, the endpoint $h^{-1}$ is not too close to a zero, i.e. $\xi_{k(\hinv)}+h^{-\beta}\|f\|_\infty^{-1}\sqrt{\thresh} \leq \hinv$. 
\end{enumerate}
\end{assumption}

Assumption \ref{assumption:FT_decay} mirrors Assumption \ref{assumption:smoothness}, i.e. $\fft$ decays like $(1+|\omega|)^{-\beta}$, but is more flexible, allowing $\fft$ to vanish.
Assumption \ref{assumption:zeros} ensures that we can identify individual zeros and estimate them in a stable manner; since the zeros are simple, small perturbations in our estimation do not result in duplicate estimated zeros. We note that in higher dimensions Assumption \ref{assumption:zeros} would generalize to $ Z_{\fft} $ being isolated and simple along any ray out of the origin. Assumption \ref{assumption:nice_oscillations} ensures that $ o_f $ and $ o_f' $ are sufficiently regular, i.e. away from the zeros the oscillatory component cannot be too small and close to the zeros the derivative of the oscillatory component cannot be too small. These assumptions ensure that the zeros are identifiable. Namely, $ |o_f| $ being sufficiently large away from $ Z_{\fft} $ ensures that we do not incorrectly identify false zeros in estimation, while $ |o_f'| $ being sufficiently large in a small neighborhood of $ Z_{\fft} $ ensures that $ \fft $ crosses zero at an acute angle as opposed to flattening out, making the estimation of zeros more stable.

Under the above assumptions, we are able to begin our analysis. 
Three key technical lemmas, proved in the Supplementary Material \ref{app:vanishing}, are needed to control the error of the estimator defined in \eqref{equ:GeneralizedKotlarskiEst}: Lemma \ref{zeros_location_control} bounds the perturbation of the zeros, Lemma \ref{lem:error_from_zeros} bounds the impact of using approximate zeros, and Lemma~\ref{lem:error_from_integrand} bounds the impact of using an approximation of the integrand. To isolate the impact of each of these challenges (i.e. approximation of zeros and approximation of the integrand), it is convenient to define the following intermediate estimator:
\begin{align}
	\label{equ:GeneralizedKotlarskiIntermediate}
	\fft_\epsilon(\omega) &= (-1)^{\hatk(\omega)}\exp\left(\int_{\hatomega_\epsilon}\frac{\partial_{1} \Psi(\xi,\xi)}{\Psi(\xi,\xi)} d\xi \right) \, 
\end{align}
for $\hatomega_\epsilon$ defined in \eqref{equ:omega_tilde_eps}. Lemma \ref{lem:error_from_zeros} guarantees that $\fft \approx \fft_\epsilon$ while Lemma~\ref{lem:error_from_integrand} guarantees that $\fft_\epsilon\approx\hatfft_\epsilon$. 

We begin by showing that, with control over the approximation of the power spectrum and its derivatives ---which holds with high probability by a slight modification of Lemma \ref{lem:psiestimates}--- we can estimate the correct number of zeros of $ \fft $ and accurately approximate their location. 

\begin{lemma}[Perturbation of zeros]\label{zeros_location_control}
	   Under Assumption \ref{ass:vanishing_model}, there exist positive universal constants $c_1,c_2$ such that, for any $\tau \ge 1$, it holds with probability at least $1-c_1 e^{-c_2 \tau}$ that there is a one-to-one correspondence between the true zeros $Z_{\fft}$ and the empirical approximations $\mathcal{Z}_\thresh$ defined in \eqref{equ:approx_zeros_via_thresholding}, and the error of approximating the zeros is bounded by:
	\begin{align*}
		\max_{i=1,\ldots, k(\hinv)} \left| \hatxi_i- \xi_i \right|&:= \delta \lesssim \frac{\rho_N h^{-2\beta}}{\|f\|_\infty^2} \, 
	\end{align*}
    for $\rho_N$ defined in \eqref{equ:def_of_rhoN}.
\end{lemma}

With control over the approximation of $Z_{\fft}$, we quantify the impact of using approximate zeros in the generalized Kotlarski formula.

\begin{lemma}[Error from approximate zeros]\label{lem:error_from_zeros}
	Let $\delta$ bound the perturbation of the zeros and assume $\delta \leq \frac{\epsilon}{4}$. Under Assumption \ref{ass:vanishing_model}, for all $0\leq\omega\leq \hinv$, we have for $\epsilon + \delta\epsilon^{-1} \lesssim h$:
	\begin{align}
    \label{equ:error_intermediate}
		\left|\fft(\omega)-\fft_\epsilon(\omega) \right| &\lesssim \hinv \|f\|_\infty \left(\delta \epsilon^{-1}+\epsilon\right) \, ,
	\end{align}
where $\fft_\epsilon$ is the intermediate estimator  defined in \eqref{equ:GeneralizedKotlarskiIntermediate}. 
\end{lemma}

Finally, we quantify the impact of using an approximate integrand in the generalized Kotlarski formula.

\begin{lemma}[Error from approximate integrand]\label{lem:error_from_integrand}
Suppose Assumption \ref{ass:vanishing_model} holds and that $\rho_N \lesssim \epsilon^3 \|f\|^2_\infty h^{3\beta+1}$
for $\rho_N$ defined in \eqref{equ:def_of_rhoN}.  Then, there exist positive universal constants $c_1,c_2$ such that, for any $\tau \ge 1$, it holds with probability at least $1-c_1 e^{-c_2 \tau}$ that for all $0\leq\omega\leq \hinv$:
	\begin{align*}
    \left |\fft_\epsilon(\omega)-\hatfft_\epsilon(\omega) \right| &\lesssim  \frac{\rho_N h^{-2\beta-1} }{ \epsilon^3 \|f\|_\infty}\, ,
	\end{align*}
where $\fft_\epsilon$ is the intermediate estimator  defined in \eqref{equ:GeneralizedKotlarskiIntermediate}. 
\end{lemma}

With all of these pieces assembled, we now extend Theorem \ref{thm:FTDeviation}. Notice that we presented the preceding technical lemmas for $\omega\in [0, \hinv],$ but for real-valued signals the statements can be trivially extended to $ [-\hinv,0] $ via conjugate symmetry.
We state the final result, Theorem \ref{thm:finalzerosest}, for $\omega\in [-\hinv,\hinv]  $.
\begin{theorem}[Recovery in Fourier domain, vanishing case]\label{thm:finalzerosest}
	Suppose Assumption \ref{ass:vanishing_model} holds and that $\epsilon \lesssim h$, $\rho_N \lesssim \epsilon^3 \|f\|^2_\infty h^{3\beta+1}$ for $\rho_N$ defined in \eqref{equ:def_of_rhoN}.
    Then, there exist positive universal constants $c_1,c_2$ such that, for any $\tau \ge 1$, it holds with probability at least $1-c_1 e^{-c_2 \tau}$ 
	\begin{align*}
		\sup_{\omega\in[-h^{-1},h^{-1}]} \left| \fft(\omega) - \hatfft(\omega)\right| &\lesssim \frac{\rho_N h^{-2\beta-1} }{ \epsilon^3 \|f\|_\infty} +\hinv \|f\|_\infty\epsilon .
	\end{align*}
\end{theorem}
\begin{proof}
    By the triangle inequality, we have 
    \begin{align*}
        \left| \fft(\omega) - \hatfft(\omega)\right| &\leq  \left| \fft(\omega) -\fft_\epsilon(\omega)\right|+\left|\fft_\epsilon(\omega)- \hatfft(\omega)\right| .
    \end{align*}
    By Lemma \ref{zeros_location_control}, it holds with probability at least $1-c_1 e^{-c_2 \tau}$ that
    \begin{align*}
        \max_{i=1,\ldots, k(\hinv)} \left|\xi_i - \hatxi_i \right|= \delta \lesssim \frac{\rho_N h^{-2\beta}}{\|f\|_\infty^2} \, . 
    \end{align*}
    Thus, by Lemma \ref{lem:error_from_zeros} (note $\delta\epsilon^{-1}\lesssim h$ holds with high probability under the assumption on $\rho_N$), 
    \[ \left| \fft(\omega) -\fft_\epsilon(\omega)\right| \lesssim \hinv \|f\|_\infty \left(\delta \epsilon^{-1}+\epsilon\right) 
    \lesssim \frac{\rho_N h^{-2\beta-1}}{\epsilon \|f\|_\infty}+\hinv \|f\|_\infty\epsilon\]
    and by Lemma \ref{lem:error_from_integrand} we have
    \[ \left|\fft_\epsilon(\omega)- \hatfft(\omega)\right| \lesssim \frac{\rho_N h^{-2\beta-1} }{ \epsilon^3 \|f\|_\infty}.\]
    Observing that $\frac{\rho_N h^{-2\beta-1} }{ \epsilon^3 \|f\|_\infty}$ dominates $\frac{\rho_N h^{-2\beta-1}}{\epsilon \|f\|_\infty}$ and that the above bounds hold for all $\omega\in[-\hinv,\hinv]$ proves the theorem. 
\end{proof}

\begin{remark}
The $\epsilon$ minimizing the upper bound in Theorem \ref{thm:finalzerosest} is $\epsilon^* \asymp \left(\frac{\rho_N h^{-2\beta}}{\|f\|_\infty^2}\right)^{\frac{1}{4}}$, which yields an error of
	\begin{align*}
		\sup_{\omega\in[-\hinv,\hinv]} \left| \fft(\omega) - \hatfft(\omega)\right| &\lesssim \rho_N^{\frac{1}{4}}h^{-\frac{\beta}{2}-1}\|f\|_\infty^{\frac{1}{2}}.
	\end{align*}
\end{remark}

With our estimator of $\fft$ in hand, we now define an estimator of the hidden signal $f$ in space via deconvolution (see last step of Algorithm \ref{alg:zerosalgorithm}) and bound its error in the $L^{\infty}$ norm. The proof of the next result is identical to that of Theorem \ref{thm:ErrorBoundDeconvolutionEstimator} and is omitted for brevity.

\begin{theorem}[Recovery in spatial domain, vanishing case]\label{thm:ErrorBoundDeconvolutionEstimatorVanishingFT}
    Suppose Assumption \ref{ass:vanishing_model} holds and that $\epsilon \lesssim h$, $\rho_N \lesssim \|f\|^2_\infty h^{3\beta+4}$ for $\rho_N$ defined in \eqref{equ:def_of_rhoN}. Let $\epsilon = \epsilon^* \asymp \left(\frac{\rho_N h^{-2\beta}}{\|f\|_\infty^2}\right)^{\frac{1}{4}}$. Let $K$ be a kernel satisfying Assumption \ref{assumption:kernel} with $p> \beta-1$. Then, there exist positive universal constants $c_1,c_2$ such that, for all $\tau \ge 1,$ it holds with probability at least $1 - c_1 e^{-c_2 \tau}$ that the estimator $\hatf(t)$ in Algorithm \ref{alg:zerosalgorithm} satisfies 
    \begin{align*}
         \|\hatf - f \|_\infty 
        \lesssim
        \rho_N^{\frac{1}{4}} \|f\|_\infty^{\frac{1}{2}} h^{-\frac{\beta}{2} - 2} + \|f\|_\infty h^{\beta-1} \, .
    \end{align*}
    Choosing $h = h^* \asymp \left(\frac{\rho_N}{\|f\|_\infty^2}\right)^{\frac{1}{6\beta+4}}$ thus yields an estimator $\hatf$ with relative error satisfying 
    \begin{align*}
        \frac{\|\hatf - f\|_\infty}{\|f\|_\infty}
        \lesssim
        \left(\frac{\rho_N}{\|f\|_\infty^2}\right)^{\frac{\beta-1}{6\beta+4}} \, .
    \end{align*}
 \end{theorem}

As in Section \ref{sec:theory}, we can deduce the following corollary in the high-noise regime. 

\begin{corollary}[High-noise regime] \label{corr:highNoiseRegimeVanishing} 
        Under the assumptions of Theorem \ref{thm:ErrorBoundDeconvolutionEstimatorVanishingFT}, suppose that $\|k_\eta\|^{1/2}_\infty \ge \|f\|_\infty$ and $N \gtrsim \Gamma(y)^2 \lor \tau$; then, it holds with probability at least $1-c_1e^{-c_2 \tau}$ that
       \begin{align*}
        \frac{\|\hatf - f\|_\infty}{\|f\|_\infty}
        \lesssim 
        \inparen{
        \|k_\eta\|_{\infty}
        \sqrt{\frac{\tau}{N}}
        \lor 
        \frac{\mathsf{snr}^{-1}}{\sqrt{N}}
        }^{\frac{\beta - 1}{6\beta + 4}}.
    \end{align*}     
\end{corollary}
Note that analogs of Corollaries \ref{corr:largesigma} and \ref{corr:smallLengthscale} also hold with a smaller exponent: namely, $\frac{\beta-1}{3\beta+1}$ becomes $\frac{\beta-1}{6\beta+4}$ in Corollary \ref{corr:largesigma} and $\frac{\beta-1}{6\beta+2}$ becomes $\frac{\beta-1}{12\beta+8}$ in Corollary \ref{corr:smallLengthscale}.

\section{Numerical results}\label{sec:numerics}
In this section, we study the numerical performance of Algorithms \ref{alg:mainalgorithm} and \ref{alg:zerosalgorithm}. We apply Algorithm \ref{alg:mainalgorithm} to estimate three signals with non-vanishing Fourier transform and Algorithm \ref{alg:zerosalgorithm} to estimate a signal with vanishing Fourier transform. In addition to considering signals satisfying Assumptions \ref{assumption:smoothness} and \ref{ass:vanishing_model}, we also consider super smooth signals that exhibit at least exponentially fast decay in frequency. Matlab code to reproduce all numerical experiments is publicly available at \url{https://github.com/msween11/fMRA/}.

\subsection{Setup}
\label{sec:numerics_setup}

\begin{table}[b!]
\setlength{\tabcolsep}{10pt} 
\centering
\begin{tabular}{|c|c|c||c|c|c|}
\hline
\multicolumn{3}{|c||}{\textbf{Ordinary Smooth}} & \multicolumn{3}{c|}{\textbf{Super Smooth}} \\ \hline
\textbf{$f_1$} & \multicolumn{2}{c||}{
  $\renewcommand{\arraystretch}{1.5}\begin{array}{@{}l@{}}
  g_1(x) = \mathbf{1}_{[-.5,.5]}^{\circledast2}(x) \\
  f_1(x) = g_1(x) + g_1(\pi x)
  \end{array}$
} & \textbf{$f_3$} & \multicolumn{2}{c|}{
  $\renewcommand{\arraystretch}{1.5}\begin{array}{@{}l@{}}
  g_3(x) = e^{-20x^2}\cos(8x) \\
  f_3(x) = g_3(x-0.3)
  \end{array}$
} \\ \hline
\textbf{$f_2$} & \multicolumn{2}{c||}{
  $\renewcommand{\arraystretch}{1.5}\begin{array}{@{}l@{}}
  g_2(x) = \mathbf{1}_{[-.25,.25]}^{\circledast 4}(x) \\
  f_2(x) = g_2(x) + g_2(\pi x)
  \end{array}$
} & \textbf{$f_4$} & \multicolumn{2}{c|}{
  $\renewcommand{\arraystretch}{1.5}\begin{array}{@{}l@{}}
  \mu_1 = \mathcal{N}(-.3, .02), \mu_2 = \mathcal{N}(.5, .02) \\
  f_4(x) = \text{pdf of } \mu_1,\mu_2 \text{ mixture}
  \end{array}$
} \\ \hline
\end{tabular}
   \\
\caption{ 
Signals in our experiments. $\mathbf{1}^{\circledast \mathcal{D}}_{[a,b]}$ denotes $\mathcal{D}$-fold convolution of the indicator  $\mathbf{1}_{[a,b]}.$}\label{tab:functions}
\end{table}

 All hidden signals are defined on the domain $[-2,2]$ but are supported on the smaller interval $[-1,1]$. We assume that the shift distribution is also supported on $[-1,1]$, so that every translated and noise-corrupted observation remains supported within $[-2,2]$.  
  We sample in space at a rate of $ 2^{-5} $, unless otherwise stated. Our samples are then zero-padded for increased fidelity in frequency, see Remark \ref{rmk:numerics_padding}.  We choose the infinite order deconvolution kernel $K(x) = \text{sinc}(x)$. Additionally, a regularization constant $r$ was introduced to improve algorithm performance, see Remark \ref{rmk:numerics_regconst}.

 In dimension $d$, suppose each coordinate is discretized to a grid of size $J$, so that the total number of grid points is $J^d$. Computing the Fourier transforms of all $N$ samples then costs $O\left(N\cdot J^d \log(J^d)\right)$. Constructing the matrix $\widehat{\Psi}$ requires $O\left(N\cdot J^{2d}\right)$ operations, since the second-moment matrix is of size $J^d \times J^d$. Consequently, the overall computational complexity of Algorithm~\ref{alg:mainalgorithm} is $O\left(N\cdot J^{2d}\right)$, with the dominant cost arising in Step~2. Notice that in most applications that motivate this work, the signal represents an object in physical space $\R^d$ with $d = 1,2,3.$

As in Subsection~\ref{subsec:snr}, the noise was chosen to have a squared-exponential covariance function with lengthscale parameter $ \lambda $ and scale parameter $ \sigma$ (see \eqref{eq:SEkernel}).
We work with the four signals shown in Figure \ref{fig:signalsintro} and summarized in Table \ref{tab:functions}. The first two, the linear spline $f_1$ and cubic spline $f_2$, are chosen because they explicitly satisfy Assumption \ref{assumption:smoothness}. These signals, by their symmetry, have real-valued Fourier transforms and are ordinary smooth with $\beta = 2$ for $f_1$ and $\beta = 4$ for $f_2$. They are compactly supported on $ [-1,1] $ in space without truncation. The third signal, $f_3$, is a Gabor pulse shifted to the right so that it has a non-trivial complex component in frequency, unlike the first two signals. Additionally, the signal is super smooth. Finally, we let $f_4$ be the pdf of a mixture model with even weighting of two Gaussian random variables. This signal has a finite number of zeros in frequency, satisfying Assumptions \ref{assumption:FT_decay}, \ref{assumption:zeros}, and \ref{assumption:nice_oscillations}, but $\fft_4$ decays squared-exponentially fast.
All four signals are subsequently $ \ell^\infty $ normalized so that $ \max_{[-1,1]}|f_j(x)| = 1,\, j=1,2,3,4 $. Thus, the signals lead to comparable signal-to-noise ratios in the context of Corollary \ref{corr:smallLengthscale}. See Figure \ref{fig:signals} and Remark \ref{rmk:numerics_signalrec} in the Supplementary Material \ref{appendix:numerics} for zoomed-in plots of the four signals, their Fourier transforms, their recoveries, and a discussion of the recovery process. We consider two shift distributions: $\zeta_1$, the uniform on $ [-1,1] $ and $\zeta_2$, for which we take $f_1 $, but normalized to integrate to 1.  Note that neither $\zeta_1$ nor $\zeta_2$ are periodic, as our shifted signals must stay supported in their domain of $[-2,2]$.  We will show that numerically, as in the theory, Algorithms \ref{alg:mainalgorithm} and \ref{alg:zerosalgorithm} perform similarly for both shift distributions.  All simulations in this section and in the Supplementary Material \ref{appendix:numerics} were conducted 20 times and averaged. 

\setlength{\belowcaptionskip}{-15pt }
\begin{figure}
    \centering
    \includegraphics[width=1.0\linewidth]{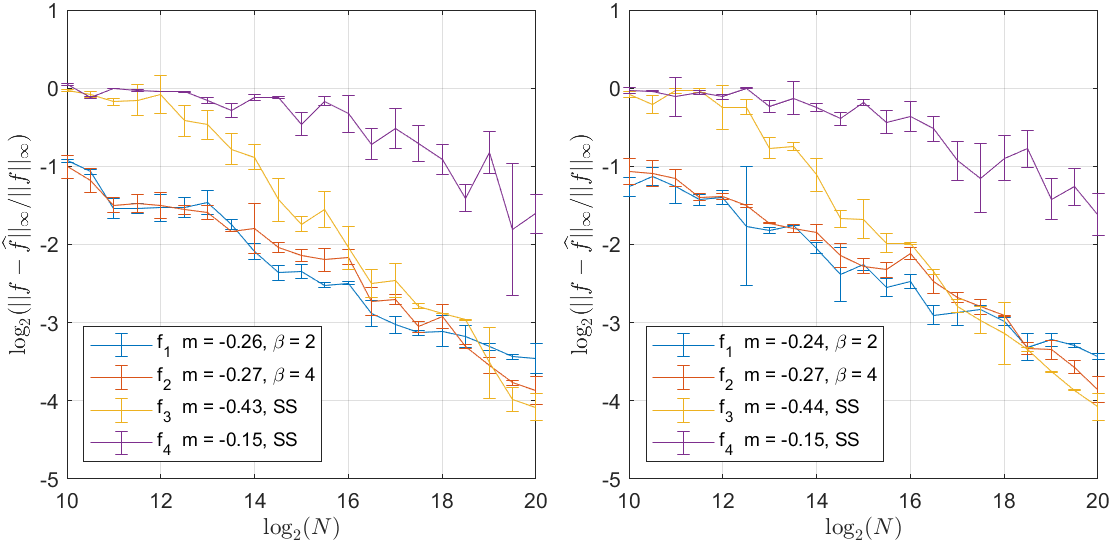}
    \caption{Error decay with varying sample size for fixed $\sigma =1$, $\lambda = 0.1$; slopes $m$ for least-squares fit shown in legend. Left: shifts from $\zeta_1$ (uniform). Right: shifts from $\zeta_2$.}
    \label{fig:nerr_sigma1}
\end{figure}

\subsection{Error response to model parameters}
 All errors are relative $ \ell^\infty $ errors in space. Figure \ref{fig:nerr_sigma1} shows relative $\ell^\infty$ error over an evenly $\log_2$ spaced range of sample sizes from $2^{10} = 1024$ to $2^{20} = 1,048,576$. The noise parameters $ \sigma = 1, \lambda = .1 $ were held constant for all experiments. Thus, all simulations were performed in the high-noise regime of Corollaries \ref{corr:highNoiseRegime} and \ref{corr:highNoiseRegimeVanishing}. The left plot in Figure \ref{fig:nerr_sigma1} shows recovery from samples shifted by $\zeta_1$ while the right plot shows recovery from samples shifted by $\zeta_2$. All experiments are deconvolved with a bandwidth parameter $h$ chosen optimally at each sample size to minimize relative $\ell^\infty$ spatial recovery error; note that choosing $h$ optimally allows us to investigate the sharpness of the $\beta$-dependent convergence rate in Theorem \ref{thm:ErrorBoundDeconvolutionEstimator}. However, the situation is largely the same when $h$ is fixed, as shown in Remark \ref{rmk:numerics_hfix}. 

As the smoothness parameter $\beta$ increases from $\beta =2$ for the linear spline to $\beta = 4$ for the cubic spline, and, essentially, $\beta= \infty$ for the super smooth Gabor wavelet, the average slope for the error decay increases. This demonstrates the dependence on $\beta$ in Theorem \ref{thm:ErrorBoundDeconvolutionEstimator}: as $\beta$ increases, the rate of error decay increases. This pattern is more noticeable at a lower noise level, see Figure \ref{fig:nerr_sigma.5} in the Supplementary Material \ref{appendix:numerics} and the discussion in Remark \ref{rmk:numerics_sig.5} for the same simulation at a noise intensity of $\sigma = 0.5$. Again, although our theory does not cover the super smooth case, the recovery is, as expected, better than with finite $\beta$. Note also that the average slope for the error decay in the case of $f_4$, the Gaussian mixture model, is worse than that of the Gabor wavelet. As expected per Theorem \ref{thm:ErrorBoundDeconvolutionEstimatorVanishingFT}, the error decay rate is slower than in Theorem \ref{thm:ErrorBoundDeconvolutionEstimator} for an otherwise comparable signal. Additionally, one can note that our empirical decay rates are better than the theoretical ones obtained in Theorems \ref{thm:ErrorBoundDeconvolutionEstimator} and \ref{thm:ErrorBoundDeconvolutionEstimatorVanishingFT}, suggesting that in these computed examples our theoretical rates ---which our results inherit from the deconvolution literature--- may be overly pessimistic.

\setlength{\belowcaptionskip}{-1pt }
\begin{figure}[tb]
    \centering
    \begin{subfigure}[b]{0.47\textwidth}
        \includegraphics[width=\linewidth]{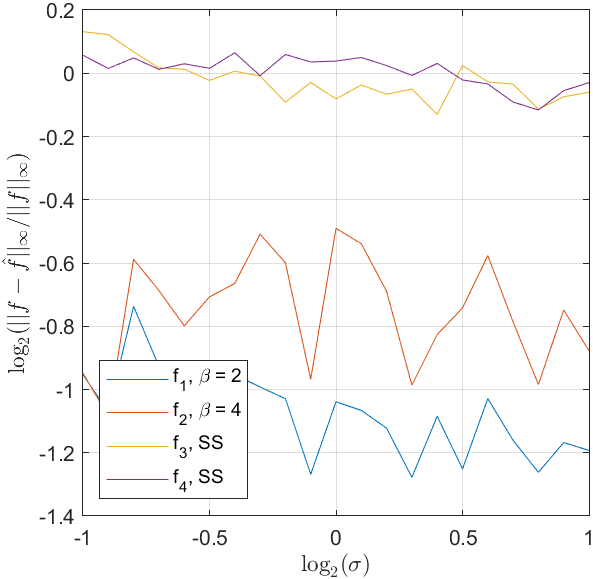}
        \caption{Varying $\sigma$, with $N = 200\cdot\sigma^4$. }
        \label{fig:sigma_nchange}
    \end{subfigure}    
    \hfill
    \begin{subfigure}[b]{0.47\textwidth}
        \includegraphics[width=\linewidth]{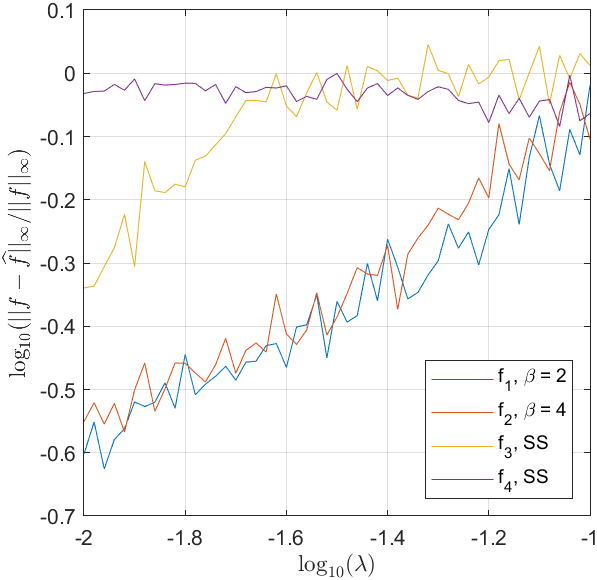}
        \caption{Varying $\lambda$, with $N = 60\log(1/\lambda)$.}
        \label{fig:lambda}
    \end{subfigure}
    \caption{Recovery error as a function of noise parameters.}
\end{figure}
In Figure \ref{fig:sigma_nchange} we investigated the dependence of the recovery error on the noise parameter $\sigma$. We fixed $\lambda = 0.1$ and varied $ \sigma $ from $2^{-1} = 0.5$ to $2^1=2$ evenly in $\log_2$ space and let $N = 200\cdot\sigma^4$. Note $\sigma = 2$ is well into the high noise regime discussed in Corollary \ref{corr:highNoiseRegime} and thus captures the behavior stated in Corollary \ref{corr:largesigma}. The bandwidth $h$ was optimized dynamically at each $\sigma$ choice to emphasize the dependence on $\sigma$. Observe that as $\sigma$ increases, with our choice of $N$, the reported error stays bounded and relatively constant, numerically verifying the stated relationship between $N$ and $\sigma$ in the large $\sigma$ regime discussed in Corollary \ref{corr:largesigma}. Moreover, this plot hints at the relationship between $\beta$ and $\sigma$ in the same corollary. Notably, Corollary \ref{corr:largesigma} states that, after taking a logarithm, the constant preceding the $\sigma^2/\sqrt{N}$ term scales like $(\beta-d)/(3\beta-1)$, so as $\beta$ increases, the error on average should increase as well. Indeed, this is exactly what we observe in our experiment. More on the interplay between $\beta$ and $\sigma$ is discussed in Remark \ref{rmk:numerics_sigma_nfix}. 

In Figure \ref{fig:lambda}, we investigated Corollary \ref{corr:smallLengthscale} by fixing $\sigma  = 1$ and varying $\lambda$ evenly in $\log_{10}$ space from $10^{-2} = 0.01$ to $10^{-1} = 0.1$. For this experiment we sampled in space at a rate of $2^{-7}$ instead of $2^{-5},$ so that the smallest lengthscale $\lambda$ was larger than our sampling rate in space. The sample size at each $\lambda$-step was taken to be $60\cdot\log(1/\lambda)$, and rounded up. Deconvolution was done with a dynamically optimized $h$. For all four signals, as $\lambda$ decreases, the error stays bounded and does not blow up, demonstrating the sample complexity's dependence on $\lambda$. Of note is that for $f_4$, the Gaussian mixture model, although the relative error stays bounded below $10^0 = 1$, it does not notably decrease as $\lambda$ decreases, like for the other signals. This is unsurprising, for as discussed in Corollary \ref{corr:highNoiseRegimeVanishing}, the dependence on $\beta$ is worse for Algorithm \ref{alg:zerosalgorithm}, and more samples are expected to be needed for a comparable recovery.

\subsection{Comparisons to existing algorithms}
\label{sec:comp_to_spectral}
\setlength{\belowcaptionskip}{-1pt }
\begin{figure}[tb]\label{fig:spec_comp}
    \centering
    \begin{subfigure}[b]{0.47\textwidth}
        \includegraphics[width=\linewidth]{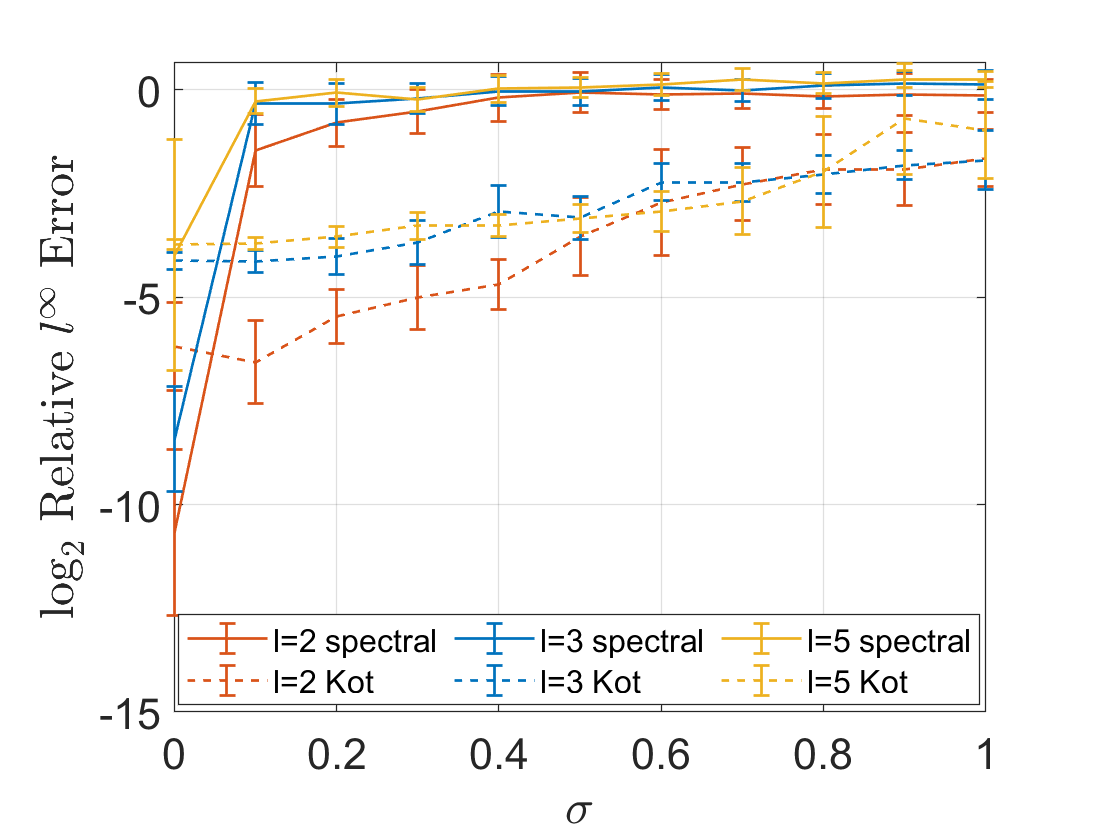}
        \caption{Isotropic noise, shifts on grid.}
        \label{fig:theirs_spec_comp}
    \end{subfigure}    
    \hfill
    \begin{subfigure}[b]{0.47\textwidth}
        \includegraphics[width=\linewidth]{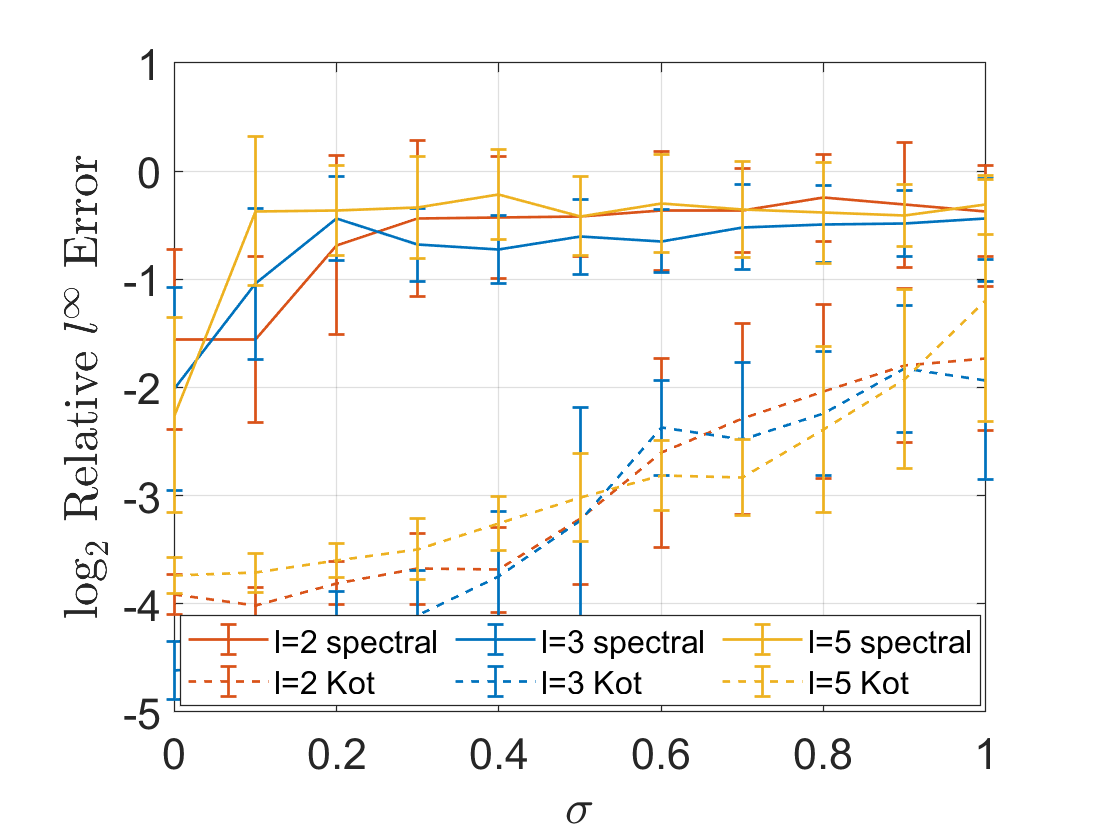}
        \caption{Squared exponential noise, continuous shifts.}
        \label{fig:ours_spec_comp}
    \end{subfigure}
    \caption{ Comparing recovery of $f_1$ from the spectral algorithm and Algorithm \ref{alg:mainalgorithm} across different sampling rates $2^{-l}$ and noise intensities. $N=10,000$ fixed. Note the differing y-axes.}
\end{figure}

In this subsection, we compare the performance of the spectral algorithm of \cite{abbe2018multireference} with that of Algorithm \ref{alg:mainalgorithm} under both the discrete MRA model \eqref{equ:classic_discrete_MRA} and our functional setting \eqref{ssec:problemsetting}. For Algorithm \ref{alg:mainalgorithm}, all experiments were conducted with $r=.1$. Both methods rely on second moment information and admit provable recovery guarantees, see Corollary IV.4 in \cite{abbe2018multireference} for the spectral algorithm. The spectral algorithm depends critically on the discrete Fourier transform of the signal staying bounded away from zero, as well as the assumption of circular shifts. Consequently, its performance can degrade  when the Fourier transform exhibits significant decay or vanishes, as occurs for compactly supported functions. This effect becomes more pronounced when the sampling rate $2^{-l}$ is taken to be fine. It is worth noting that the spectral algorithm naturally accommodates periodic shift distributions, whereas our functional setup—by requiring that shifted copies of the signal remain supported in $[-2,2]$—implicitly assumes an aperiodic shift distribution.

\subsubsection{Discrete MRA data model}
We first describe the data generation corresponding to the discrete MRA framework of \cite{abbe2018multireference}. For both algorithms, the shift distribution was supported on $[-1,1]\cap \{\mathbb{Z}\cdot 2^{-l}\}$, that is the discretized signal was shifted by a random integer number of grid points. As recommended in \cite{abbe2018multireference}, the samples used by the spectral algorithm were additionally circularly shifted according to a randomly drawn probability vector, ensuring that the shift distribution has distinct entries almost surely. See Algorithm \ref{alg:samplingdiscrete} for a detailed description of the discrete sampling procedure. The noise was chosen to be isotropic Gaussian with variance $\sigma^2$. 

Under this construction, the spectral algorithm was implemented exactly in the setting of \cite{abbe2018multireference}, while Algorithm \ref{alg:mainalgorithm} was run in the $\lambda \to 0$ regime with a discretized uniform shift distribution. We evaluate both algorithms at three sampling rates $2^{-l}$, with $l=2,3,5$, fixing the sample size at $N=10,000$ and varying the noise level $\sigma$ over the interval $[0,1]$. 

Figure \ref{fig:theirs_spec_comp} reports the relative $\ell^\infty$ error in the spatial domain for recovery of the signal $f_1.$ In the noiseless case ($\sigma = 0$), the spectral algorithm achieves slightly smaller error. However, its performance deteriorates rapdily as $\sigma$ increases, and is outperformed by Algorithm \ref{alg:mainalgorithm} at all positive noise levels considered. This suggests that Algorithm \ref{alg:mainalgorithm} provides improved robustness even in the regime of isotropic noise and discretely supported shift distributions. Moreover, the performance of the spectral algorithm degrades as the sampling rate $2^{-l}$ decreases (equivalently, as the signal length increases).

\subsubsection{Functional MRA data model}

We next consider the functional MRA setting introduced in Subsection~\ref{ssec:problemsetting}. Figure \ref{fig:ours_spec_comp} reports an experiment analogous to that of \ref{fig:theirs_spec_comp}, but with shifts drawn from a continuous uniform distribution on $[-1,1]$ and noise sampled from a centered Gaussian process with squared exponential covariance, with $\lambda = .1$ fixed. The spectral algorithm was modified to remain unbiased by incorporating the appropriate covariance matrix. 

This experiment matches the setting described in Subsection \ref{sec:numerics_setup}. In particular, the discrete MRA assumption that shifts occur exactly on the discretization grid is no longer satisfied. For each sample $y_n$, the signal and noise are evaluated on the interval $[-1+\zeta_n,1+\zeta_n]$ at resolution $2^{-l}$, so the observations do not arise from circular shifts of a fixed grid. See Algorithm \ref{alg:samplingfunctional} for further details on the functional sampling procedure. 

In this setting, the spectral algorithm performs worse than Algorithm \ref{alg:mainalgorithm} across all sampling rates considered, including the noiseless case. In Section \ref{ssec:suppspectral}, we repeat the experiments of Figures \ref{fig:theirs_spec_comp} and \ref{fig:ours_spec_comp} for the signal $f_4$, using Algorithm \ref{alg:zerosalgorithm} for its recovery, and obtain results consistent with those for $f_1$. Although the spectral algorithm is theoretically not designed to handle signals whose Fourier transform vanishes, when the sampling rate is sufficiently coarse the zeros may lie off the discretization grid, allowing the spectral method to recover $f_4$ in practice.

\section{Conclusions}\label{sec:conclusions}        
This paper has introduced a novel deconvolution framework for functional MRA, establishing a  connection to Kotlarski’s identity and extending it to multidimensional and vanishing-frequency settings. By formulating MRA in function space, our approach avoids issues arising in the discrete model as resolution increases, and enables recovery guarantees for non-bandlimited signals. Our theory provides error bounds in both Fourier and spatial domains, including in high-noise regimes, and accommodates signals with vanishing Fourier transforms ---a significant advance over prior work. Numerical experiments confirm the practical effectiveness of our estimators, demonstrating accurate estimation for a variety of signals  and improvement over existing second moment based recovery algorithms.

\subsection*{Acknowledgments}
OAG and DSA were funded by NSF DMS CAREER 2237628 and DOE DE SC0022232.  OAG was funded by the Eric and Wendy Schmidt Center at the Broad Institute of MIT and Harvard.  
AL thanks NSF DMS 2309570 and NSF DMS CAREER 2441153. MS was partly funded by NSF RTG DMS 2136198.

\bibliographystyle{siam}
\bibliography{ref.bib}

\begin{appendix}
\section{Supplementary Materials for Section~\ref{sec:settingandalgorithm}}
This section contains the proof of an auxiliary result, Lemma~\ref{lem:technicalKotlarski}, used to establish Theorem~\ref{thm:kotlarski}.

\begin{lemma}[Path-wise exponential representation]\label{lem:technicalKotlarski}
    Suppose that $\xi:\R^d \to \C$ is a continuously differentiable function and 
    let $\Re \xi, \Im \xi : \R^d \to \R$ denote the real and imaginary parts of $\xi$, respectively, so that $\xi = \Re \xi+ i \Im \xi.$ Then for any $\omega_1, \omega_2 \in \R^d$ for which $\xi$ is non-vanishing on the straight-line path from $\omega_1$ to $\omega_2,$ which we denote by $\gamma: [0,1] \to \C$,  we have
    \begin{align*}
        \frac{\xi(\omega_2)}{\xi(\omega_1)} = \exp \inparen{\int_0^1 \frac{\nabla \xi(\gamma(\alpha) )}{\xi(\gamma(\alpha))} \cdot (\omega_2-\omega_1) d\alpha},
    \end{align*}
    where $\nabla \xi(\omega) \cdot \omega := \nabla \Re \xi(\omega) \cdot \omega +i\nabla \Im \xi(\omega)\cdot \omega$. 
\end{lemma}
\begin{proof}
    Fix $\omega_1, \omega_2 \in \R^d$ for which the stated assumption holds and let $\gamma: [0,1] \to \R^d$ be the straight line path from $\omega_1$ to $\omega_2$ (i.e. $\gamma(\alpha)= (1-\alpha) \omega_1 + \alpha \omega_2$). Let $\mcL(\alpha) := \log_c \xi( \gamma(\alpha))$, where $\log_c$ denotes the continuously chosen branch of the complex logarithm along the path $\alpha \mapsto \xi(\gamma(\alpha)),$ with initial condition $\mcL(0) = \log \xi(\omega_1).$ Note that this is well-defined since $\xi\circ \gamma : [0,1] \to \C$ is continuously differentiable and non-vanishing. Further, $\mcL$ is smooth on $[0,1].$ We then have 
    \begin{align}\label{eq:complexChainRule}
        \frac{d}{d\alpha} \mcL(\alpha) = 
        \frac{\nabla \xi( \gamma(\alpha) )}{\xi(\gamma(\alpha))} \cdot (\omega_2 - \omega_1).
    \end{align}
    The above can be justified as follows. Write
    \begin{align*}
        \log_c \xi  = \log |\xi| + i \arg (\xi),
    \end{align*}
    where $|\xi |$ is the magnitude and $\arg(\xi )$ is the continuously chosen argument of $\xi$ along the curve. Below, all functions are evaluated at $\gamma(\alpha)$. Standard calculation then yields
    \begin{align*}
        \frac{d}{d\alpha} \log |\xi| 
        &= \frac{1}{|\xi|}\frac{d}{d\alpha} |\xi|
        = \frac{\Re \xi \frac{d}{d\alpha} R + \Im \xi \frac{d}{d\alpha} \Im \xi }{\Re \xi^2 + \Im \xi^2 }
        = 
        \frac{
        (\Re \xi \nabla \Re \xi+\Im \xi \nabla \Im \xi) \cdot (\omega_2-\omega_1)
        }{\Re \xi^2 + \Im \xi^2 },\\
        \frac{d}{d\alpha} \arg(\xi )  
        &=
        \frac{
            (\Re \xi \nabla \Im \xi -\Im \xi \nabla \Re \xi)\cdot (\omega_2-\omega_1)
        }{\Re \xi^2 + \Im \xi^2}.
    \end{align*}  
    Therefore, 
    \begin{align*}
        \frac{d}{d\alpha} \log |\xi| 
        +i\frac{d}{d\alpha} \arg(\xi )  
        &=
        \frac{
        (\Re \xi \nabla \Re \xi 
        +
        \Im \xi \nabla \Im \xi
        + i \Re \xi \nabla \Im \xi
        -
        i \Im \xi \nabla \Re \xi )\cdot (\omega_2-\omega_1)
        }{\Re \xi^2 + \Im \xi^2 }\\
        &=
        \frac{
        ((\nabla \Re \xi + i \nabla \Im \xi)(\Re \xi - i \Im \xi))\cdot (\omega_2-\omega_1)
        }{R^2 + \Im \xi^2 }\\
        &= \frac{\nabla \Re \xi + i \nabla \Im \xi}{ \Re \xi + i \Im \xi}\cdot (\omega_2-\omega_1) = \frac{\nabla \xi}{\xi} \cdot (\omega_2-\omega_1).
    \end{align*}
    Now, integrating both sides of \eqref{eq:complexChainRule} yields
    \begin{align*}
        \mcL(1) - \mcL(0) 
        = \int_0^1 \frac{d}{d\alpha} \mcL(\alpha) d \alpha
        =  \int_0^1 \frac{\nabla \xi(\gamma(\alpha))}{\xi(\gamma(\alpha))} \cdot (\omega_2-\omega_1) d\alpha,
    \end{align*}
    and so the final result follows by exponentiating both sides of the previous display.
\end{proof}

\section{Supplementary Materials for Section~\ref{sec:theory}}\label{sec:appendixtheory}
\subsection{Background}

\subsubsection{Auxiliary bounds}
\begin{lemma}[$\fft$ bounds]\label{lem:fftnbound}
	For $f$ as in Assumption~\ref{assumption:psiestimates} and any $\omega \in \C$
	\begin{align*}
    |(\fft)^{(\ell)}(\omega)|
    \le \int_{\R^d} |t|^\ell f(t) dt
    \leq c \|f\|_\infty, \qquad \ell =0,1,2,\dots,
	\end{align*}
    where $c$ is a positive constant depending only on $D, \ell.$
\end{lemma}
\begin{proof}
    Observe that for any $\ell \in \{ 0,1,2,3\}$ and $\omega \in \C$
\begin{align*}
	|(\fft)^{(\ell)}(\omega)| &= |(\omega^{\ell}\cdot f)^{ft}(\omega)| 
    = \left|\int_{\R^d} e^{-i\omega t} t^\ell f(t)dt \right|
	\le \|f\|_\infty  \sup_{t \in D} |t^\ell|  \le c \|f\|_\infty,
\end{align*}
where we have used the fact that $f$ is supported on $[-1/2, 1/2]^d.$
\end{proof}

\subsubsection{Sub-Gaussian processes}

\begin{definition}[Sub-Gaussian process]\label{def:subgaussian}
A stochastic process $\{X(t): t \in D\}$ whose sample paths lie almost surely in the Sobolev space $H^s(D)$ is said to be sub-Gaussian if there exists a positive constant $K$ such that, for every bounded linear functional $L: H^s(D) \to \R$ with $\|L\| \le 1$, the scalar random variable $L(X)$ is $K$ sub-Gaussian. In other words, the Orlicz norm satisfies
$$
\|L(X)\|_{\psi_2} := \inf\left\{ m > 0 : \E \left[e^{L(X)^2 / m^2}\right] \le 2 \right\} \le K.
$$
\end{definition}
\begin{lemma}[$y$ is sub-Gaussian]\label{lem:processIsSubgaussian}
    Under Assumption \ref{assumption:psiestimates}, the process $y(t)$ defined in \eqref{eq:MRADataGeneratingModel} is sub-Gaussian.
\end{lemma}
\begin{proof}
For any bounded linear functional $L:H^s(D) \to \R,$ we have that 
\begin{align*}
    |L(f(\cdot - \zeta))| \le \|L\| \|f(\cdot - \zeta)\|_{H^s(D)} = \|L\| \|f\|_{H^s(D)} < \infty,
\end{align*}
where the last equality holds since the Sobolev norm is invariant under shifts, since both signal and the shifted signal are supported on $D$. It follows that, for $\|L\| \le 1,$ $|L(f(\cdot - \zeta))|$ is bounded by a constant independent of $L$ and is therefore sub-Gaussian. Further, since $\eta$ is a Gaussian process supported on $H^s(D)$, $s>d/2,$ it is  sub-Gaussian. Therefore, $y$ is a sum of two sub-Gaussian processes and so is itself a sub-Gaussian process.
\end{proof}

\subsubsection{Properties of Fourier transformed processes}
In this section, we recall some  properties of the Fourier transform of a stochastic process. Throughout, we consider a generic separable sub-Gaussian process $\{X(t): t \in D\}$ supported on $H^s(D)$ with $s > d/2$, 
and $D = [-1,1]^d,$ with mean function, autocorrelation function, and covariance functions $m_X(t)$, $r_X(t,t'),$ and $k_X(t,t')$ respectively. The Fourier transformed process is then:
\begin{align*}
    \Xft(\omega) = \int_D X(t) e^{-i \omega \cdot t}dt, \qquad \omega \in\R^d.
\end{align*}
    The following lemma shows that the mean function of the transformed process is the Fourier transform of the mean function of the original process. Similarly, the autocorrelation and covariance functions are the Fourier transforms of the autocorrelation and covariance functions of the original process up to a flip in sign of the second argument.
\begin{lemma}[Properties of Fourier transformed process]\label{lem:propertiesOfFourierProcess}
    For $\omega, \omega_1, \omega_2 \in \R^d$, the Fourier transformed process $\{\Xft(\omega): \omega \in \R^d\}$ has the following properties:
    \begin{enumerate}
        \item Mean function: $m_{\Xft}(\omega) = \mft_X(\omega);$
        \item Autocorrelation function: $r_{\Xft}(\omega_1, \omega_2) =\rft_X (\omega_1, -\omega_2);$
        \item Covariance function: $k_{\Xft}(\omega_1, \omega_2) =\kft_X (\omega_1, -\omega_2).$
    \end{enumerate}
\end{lemma}
\begin{proof}
    For the mean, we have
    \begin{align*}
    m_{\Xft}(\omega)
    := \E[\Xft(\omega)] 
    = \int_D \E[X(t)] e^{-i \omega t}dt
    = \int_D m_X(t) e^{-i \omega t}dt
    = \mft_X(\omega).
\end{align*}
The second equality is justified by Fubini's theorem. To verify this, note that the Sobolev embedding theorem guarantees that all point evaluation functionals have unit norm when $s > d/2$. Since $X$ takes values in $H^s(D)$, we can express $X(t)$ as $\mcE_t(X),$ where $\mcE_t$ is the evaluation functional at $t$. The sub-Gaussian property of the process $Z$ then immediately implies that each $X(t)$ is a sub-Gaussian random variable i.e. that $\|X(t)\|_{\psi_2} \le K.$

This implies (see e.g. \cite[Proposition 2.5.2]{vershynin:HDP2017}) that for any $p,$ $\|X(t)\|_{L_p} \lesssim c_p K$ for a constant $c_p$ depending only on $p$. In particular, for $p=1,$ we have $\E|X(t)| \le  c_1 K.$ Since $D \subset \R^d$ is compact, we have  
\begin{align*}
    \int_D \E |X(t)| dt \le (\sup_{t \in D} \E|X(t)|) \text{vol}(D) < \infty.
\end{align*}
For the autocorrelation, we have 
\begin{align*}
    r_{\Xft}(\omega_1, \omega_2)
    &=\E[\Xft(\omega_1) \overline{\Xft(\omega_2)}] 
    = \iint_{D \times D} \E [X(t)X(t')] e^{-i\omega_1t} e^{i\omega_2t'} dt' dt\\
    &= \iint_{D \times D} r_X(t,t') e^{-i\omega_1t} e^{i\omega_2t'} dt' dt
    =\rft_X(\omega_1,-\omega_2),
\end{align*}
where the second equality again follows by Fubini since
\begin{align*}
    \iint_{D \times D} \E|X(t)X(t')| dt' dt
    &\le 
    \iint_{D \times D} \|X(t)\|_{L_2} \|X(t')\|_{L_2} dt' dt\\
    &\le C_2^2 
    \sup_{t \in D} \|X(t)\|^2_{\psi_2} (\text{vol}(D))^2 
    <\infty,
\end{align*}
where the first inequality holds by Cauchy-Schwarz, and the second by properties of sub-Gaussian variables. The proof for the covariance function follows similarly and is omitted for brevity.
\end{proof}

\subsubsection{Autocorrelation estimation for sub-Gaussian processes}
Suppose now that we have access to $N$ i.i.d. copies of $X(\cdot),$ denoted $X_1(\cdot),\dots, X_N(\cdot)$, and that we wish to estimate the autocorrelation function $r_X.$ The natural estimator is the sample autocorrelation function
\begin{align*}
    \hatr_X(t,t') := \frac{1}{N} \sum_{n=1}^N X_n(t)X_n(t').
\end{align*}
We then have the following high-probability supremum-norm bound.

\begin{lemma}[Concentration of $\hatr_X$] \label{lem:SecondMomentHPBound}
    There exist positive universal constants $c_1,c_2$ such that, for any $\tau \ge 1,$ it holds with probability at least $1-c_1e^{-c_2 \tau}$ that
    \begin{align}\label{eq:SecondMomentHPBound}
        \sup_{t,t' \in D} |\hatr_X(t,t')-r_X(t,t')|
        \lesssim 
        \|r_X\|_\infty
        \left (
        \sqrt{\frac{\tau}{N}}
        \lor 
        \frac{\tau}{N}
        \lor 
     \frac{\Gamma(X)}{\sqrt{N}}
        \lor 
        \frac{\Gamma(X)^2}{N}
        \right),
    \end{align}
    where 
    \begin{align*}
        \Gamma(X) := \frac{\E[\sup_{t\in D} X(t)]}{\|r_X\|_\infty^{1/2}},
        \qquad 
        \|r_X\|_\infty:= \sup_{t \in D} r_X(t,t).
    \end{align*}
\end{lemma}
\begin{proof}
    First, we deal with the centered case. We write the quantity of interest as a product empirical process indexed by the class of linear evaluation functionals:
    \begin{align}\label{eq:prodEmpProcess}
        \sup_{s,t\in D} |\hatr_X(t,t') - r_X(t,t')|
        &= 
        \sup_{s,t\in D} \abs{\frac{1}{N} \sum_{n=1}^n X_n(t)X_n(t') - \E[X_n(t) X_n(t')]}\nonumber\\
        &= 
        \sup_{f, g \in \mcF} \abs{\frac{1}{N} \sum_{n=1}^n f(X_n) g(X_n) - \E[f(X) g(X)]} ,
    \end{align}
    where $\mcF := \{\mcE_t \}_{t \in D}$ is the family of evaluation functionals. By Theorem~\cite[Theorem 1.13]{mendelson2016upper}, with probability at least $1-c_1 e^{-c_2 \tau}$, \eqref{eq:prodEmpProcess} is bounded above by
    \begin{align*} 
         d^2_{\psi_2}(\mcF)
        \inparen{
            \sqrt{\frac{\tau}{N}}
            \lor 
            \frac{\tau}{N}
            \lor 
            \frac{\gamma_2(\mcF, \psi_2)}{\sqrt{N} d_{\psi_2}(\mcF)}
            \lor 
            \frac{\gamma^2_2(\mcF, \psi_2)}{N d^2_{\psi_2}(\mcF)}
        },
    \end{align*}
    where $\gamma_2$ is Talagrand's generic complexity \cite[Definition 2.7.3]{talagrand2022upper}, $d_{\psi_2}(\mathcal F) := \sup_{f \in \mathcal F}\|f\|_{\psi_2}$ where $\|\cdot\|_{\psi_2}$ denotes the Orlicz (sub-Gaussian) norm. By \cite[Proposition 3.1]{al2025covariance}, $d_{{{\psi}_2}}^2(\mathcal F) \lesssim \sup_{t\in D} g(t,t)$ and $\gamma_2(\mathcal F, \psi_2) \lesssim \E[\sup_{t\in D} X(t)],$ which completes the proof. 
    
    In the non-centered case, write $\barf(X_n) := f(X_n) - \E f(X_n)$ and $\barg(X_n) := g(X_n) - \E g(X_n)$. Then,
    \begin{align*}
        f(X_n) g(X_n) 
        &= 
        \barf(X_n) \barg(X_n) 
        +
        \barg(X_n) \E f(X_n) 
        +
        \barf(X_n) \E g(X_n) 
        +
        \E f(X_n) \E g(X_n).
    \end{align*}
    By the triangle inequality, we have that
    \begin{align*}
            \sup_{f,g \in \mcF} &
            \abs{
            \frac{1}{N} \sum_{n=1}^N 
                f(X_n)
                g(X_n)
                -
                \E [f(X_n) g(X_n)]} 
            \le 
            \sup_{f,g \in \mcF} 
            \abs{
            \frac{1}{N} \sum_{n=1}^N 
                \barf(X_n)
                \barg(X_n)
                -
                \E [\barf(X_n) \barg(X_n)]} \\
            &\qquad \qquad +
            2 \sup_{f \in \mcF} |\E f(X)|
            \sup_{g \in \mcF} 
            \abs{
            \frac{1}{N} \sum_{n=1}^N 
                \barg(X_n)
                -
                \E [\barg(X_n)]} 
                =: \textbf{(I)} + \textbf{(II)}. 
    \end{align*}
    $\textbf{(I)}$ can be controlled by invoking the result in the centered setting. For $\textbf{(II)}$, by the equivalence of $L_2$ and $\psi_2$ norms for linear functionals, we have $\sup_{f \in \mcF} |\E f(X)| \lesssim d_{2}(\mcF) \lesssim d_{\psi_2}(\mcF),$ where $d_2(\mcF):=\sup_{f\in\mcF}\|f\|_{L_2}.$ Combining this with \cite[Theorem 2.5]{mendelson2010empirical} yields that with probability at least $1-2e^{-ct}$, $\textbf{(II)} \lesssim d_{\psi_2}(\mcF) \gamma_2(\mcF, \psi_2)/\sqrt{N}.$
\end{proof}

Suppose instead that we were given the Fourier transformed observations, $\Xft_1(\cdot),\dots, \Xft_N(\cdot)$ and wished to estimate the autocorrelation function $r_{\Xft}$. It turns out that Lemma~\ref{lem:SecondMomentHPBound}, which concerns the original process in the spatial domain, readily leads to a high-probability bound in the Fourier domain. We first define the sample version of $\hatr_{\Xft}$, 
\begin{align*}
    \hatr_{\Xft}(\omega_1, \omega_2) 
    := \iint_{D \times D} \hatr_X(t,t')e^{-i\omega_1 \cdot t} e^{i \omega_2 \cdot t'}dt dt'.
\end{align*}

\begin{lemma}[Concentration of $\hatr_{\Xft}$]
    There exist positive universal constants $c_1,c_2$ such that, for any $\tau \ge 1,$ it holds with probability at least $1-c_1e^{-c_2 \tau}$ that
    \begin{align*}
        \sup_{\omega_1,\omega_2 \in \R^d} |\hatr_{\Xft}(\omega_1,\omega_2)-r_{\Xft}(\omega_1,\omega_2)|
        \lesssim 
        \| r_X \|_\infty 
        \left (
        \sqrt{\frac{\tau}{N}}
        \lor 
        \frac{\tau}{N}
        \lor 
     \frac{\Gamma(X)}{\sqrt{N}}
        \lor 
        \frac{\Gamma(X)^2}{N}
        \right).
    \end{align*}
\end{lemma}
\begin{proof}
We have
    \begin{align*}
    \sup_{\omega_1,\omega_2 \in \R^d} |\hatr_{\Xft}(\omega_1,\omega_2)-r_{\Xft}(\omega_1,\omega_2)|
    &\le 
    \sup_{\omega_1,\omega_2 \in \R^d}
    \iint_{D \times D} |\hatr_X(t,t') - r_X(t,t')| |e^{-i \omega_1 \cdot t}| |e^{i\omega_2 \cdot t'}| dt dt'\\
    &
    \le 
    \sup_{t,t' \in D}|\hatr_X(t,t') - r_X(t,t')| 
    \iint_{D \times D}  
    dt dt'\\
    &= \sup_{t,t' \in D}|\hatr_X(t,t') - r_X(t,t')| (\text{vol}(D))^2,
\end{align*}
where the second equality holds by the fact that $|e^{-i\omega \cdot t}| = 1$ for any $t, \omega$. By assumption $\text{vol}(D) < \infty$. The result then follows by Lemma~\ref{lem:SecondMomentHPBound}.
\end{proof}

\subsection{Proof of Lemma~\ref{lem:psiestimates}}\label{ssec:estimatePsi}

\begin{proof}[Proof of Lemma~\ref{lem:psiestimates}]
    Recall that $\Psi(\omega_1, \omega_2) = r_{\yft}(\omega_1, \omega_2) - k_{\etaft}(\omega_1, \omega_2)$ and $\hatPsi(\omega_1, \omega_2) = \hatr_{\yft}(\omega_1, \omega_2) - k_{\etaft}(\omega_1, \omega_2),$ where $\hatr_{y}$ is the empirical version of the autocorrelation function $r_{y}$. It follows that 
    \begin{align*}
        &\sup_{\omega \in [-h^{-1}, h^{-1}]^d} 
        |\hatPsi( \omega, \omega) -\Psi( \omega, \omega) |
        \le 
        \sup_{\omega \in [-h^{-1}, h^{-1}]^d}
        \int_{D \times D} |\hatr_{y}(t,t')-r_{y}(t,t')| |e^{-i\omega\cdot t}||e^{i \omega \cdot t'}| dt dt'
        \\
        &\quad = 
        \int_{D \times D} |\hatr_{y}(t,t')-r_{y}(t,t')| dt dt'
        \lesssim \sup_{t,t' \in D} |\hatr_{y}(t,t')-r_{y}(t,t')| (\text{vol}(D))^2
        \lesssim \rho_N,
    \end{align*}
    where the final inequality holds with probability at least $1-c_1e^{-c_2 \tau}$ by Lemma~\ref{lem:SecondMomentHPBound}.

    Next we consider the gradient term $\nabla_1 \Psi(\omega,\omega).$ Note first that 
    \begin{align*}
        \nabla_1 \Psi(\omega,\omega)
        = -i \int_{D\times D}  t r_{y}(t,t')  e^{-i\omega \cdot t +i\omega \cdot t'} dt dt',
    \end{align*}
    and 
    \begin{align*}
        \widehat{\nabla_1\Psi}(\omega,\omega)
        :=
        \nabla_1 \hatPsi(\omega,\omega)
        = -i \int_{D\times D}  t \hatr_{y}(t,t')  e^{-i\omega \cdot t + i \omega \cdot t'} dt dt'.
    \end{align*}
    Therefore, we have 
    \begin{align*}
        \sup_{\omega \in [-h^{-1}, h^{-1}]^d} 
        \|\nabla_1 \hatPsi( \omega, \omega) - \nabla_1 \Psi( \omega, \omega)\|_2
        &\le 
        \sup_{\omega \in [-h^{-1}, h^{-1}]^d} 
        \int_{D\times D}  
        \|t\|_2
        |\hatr_{y}(t,t')-r_{y}(t,t')|  
        dt dt'\\
        &\le 
        \sup_{t,t' \in D} 
        |\hatr_{y}(t,t')-r_{y}(t,t')|
        \int_{D\times D}   \|t\|_2  dt dt'
        \lesssim
        \rho_N,
    \end{align*}
    where the final inequality holds with probability at least $1-c_1e^{-c_2 \tau}$ by Lemma~\ref{lem:SecondMomentHPBound}.
\end{proof}

\subsection{Kernels satisfying Assumption \ref{assumption:kernel}}
Here, we give a simple example of a class of kernels satisfying Assumption~\ref{assumption:kernel}. As is standard (see e.g. \cite{kurisu2022uniform} for the univariate case), we specify the Fourier transform of the kernel. Let $\Kft_0:\R \to \R$ be the Fourier transform of a univariate kernel that is even $(\Kft_0(\omega)= \Kft_0(-\omega))$, supported on $[-1,1]$, $(p+2)$-times continuously differentiable and satisfying 
    \begin{align*}
        (\Kft_0)^{(\ell)}(0) = \begin{cases}
            1 \qquad \text{if}\qquad \ell=0,\\
            0 \qquad \text{if}\qquad \ell=1,\dots, p-1,
        \end{cases}
    \end{align*}
    where $(\Kft_0)^{(\ell)}$ denotes the $\ell$-th derivative. Define the univariate kernel by the inverse Fourier transform, $K_0(t) = \frac{1}{2\pi} \int_{\R} e^{-i\omega t} \Kft(\omega) d \omega,$ and define the multivariate kernel $K:\R^d \to \R$ as the tensor product kernel with components $K_0,$ that is, $K(t) = \prod_{j=1}^d K_0(t_j).$ As we will show, $K$ satisfies Assumption~\ref{assumption:kernel}. Standard choices for $\Kft_0$ are $\Kft_0(\omega) = (1-\omega^2)^a \mathbf{1}\{\omega \in [-1,1]\}$ for $a \ge p+3$ and
    \begin{align*}
    \Kft_0(\omega) = 
    \begin{cases}
         1 \qquad                         & \text{if } \qquad|\omega| \le c_0,\\
         \exp\inparen{
         \frac{-b \exp(-b/(|\omega|-c_0)^2)}{(|\omega|-1)^2}
         } \qquad & \text{if } \qquad c_0 < |\omega| < 1,\\
         0 \qquad                         & \text{if } \qquad |\omega| \ge 1,
    \end{cases}
\end{align*}
    for $c_0 \in (0,1), b > 0,$ which is commonly referred to as the flat-top kernel. Now, we prove that $K$ satisfies Assumption~\ref{assumption:kernel} (i)-(v). 
   \begin{enumerate}[label=(\roman*)]
    
    \item By the Fourier inversion theorem, $\int_{\R} K_0(t_j) dt_j = \Kft_0(0)=1.$ Therefore 
    $$
    \int_{\R^d} K(t)dt
     = \inparen{\int_{\R} K_0(t_1) dt_1}^d = 1.
    $$
    
    \item For any multi-index $\nu$ with $1 \le |\nu| \le p-1,$ we have 
    \begin{align*}
        \int_{\R^d} t^\nu K(t)dt = \prod_{j=1}^d \int_{\R} t_j^{\nu_j} K_0(t_j) dt_j,
    \end{align*}
    and so the assumption holds if at least one of the univariate moments is zero, i.e. if $\int_{\R} t_j^{\nu_j} K_0(t_j) dt_j = 0$ for some $j$ where $1 \le \nu_j \le p-1.$ Note that $\int_\R t_j^{\nu_j} K_0(t_j) dt_j = i^{-v_j} (\Kft_0)^{(v_j)}(0) = 0$ by assumption on $\Kft_0.$
    
    \item By \cite[remark 1]{kurisu2022uniform}, $|K_0(t_j)| = o(|t_j|^{-p-2})$ as $|t_j| \to \infty.$ Therefore, for sufficiently large $t_j,$ $|t_j|^p |K_0(t_j)| \le c|t_j|^{-2} \in L^1(\R),$ and so $|t_j|^{\nu_j} |K_0(t_j)| \le c|t_j|^{-2} \in L^1(\R)$ for any $1 \le \nu_j \le p.$ Hence, for $|\nu| = p$, 
    \begin{align*}
        \int_{\R^d} |t^\nu K(t)|dt = \prod_{j=1}^d \int_{\R} |t_j^{\nu_j}||K_0(t_j)| dt_j < \infty,
    \end{align*}
    since each $v_j \le p.$

    \item Using that $\|t\|_2^p \le d^{p/2} \|t\|_p^p,$ we have 
    \begin{align*}
        \int_{\R^d} \|t\|^p_2 |K(t)|dt
        &\le d^{p/2} \sum_{j=1}^d \int_{\R^d}|t_j|^p \prod_{l=1}^d |K_0(t_k)| dt\\
        &= d^{p/2} \sum_{j=1}^d 
        \inparen{\int_{\R} |t_j|^p |K_0(t_j)|  dt_j}
        \prod_{ k \neq j} \inparen{\int_{\R} |K_0(t_k)| dt_k }
        \le C_K < \infty.
    \end{align*}
    In the above display, the equality holds by Fubini's theorem and the fact that the integral factorizes, and the second inequality holds since $t_j^p K_0(t_j), K_0(t_k) \in L^1(\R^d)$ for any $j,k \le p.$

    \item Since $\Kft(\omega_j)$ is supported on $[-1,1]$ by assumption and $\Kft(\omega) = \prod_{j=1}^d \Kft(\omega_j),$ it follows immediately that $\Kft$ has support $\{ \omega: \|\omega\|_\infty \le 1\}.$
\end{enumerate}

\subsection{Proof of Theorem~\ref{thm:FTDeviation}}\label{app:thmFTDeviation}

\begin{lemma}\label{lem:DeltaControl}
    Define 
    \begin{align*}
    \Delta(\omega)
    =
    \int_0^1 \inparen{
    \frac{\nabla_1 \hatPsi(\alpha\omega,\alpha\omega)}{\tildePsi(\alpha\omega,\alpha\omega)} \cdot \omega
    -
    \frac{\nabla_1 \Psi(\alpha\omega,\alpha\omega)}{\Psi(\alpha\omega,\alpha\omega)} \cdot \omega
    } d\alpha.
\end{align*}
Under Assumption~\ref{assumption:smoothness} with smoothness parameter $\beta$, let $h\in(0,1)$ satisfy $\rho_N \le c_0 \|f\|_\infty^2 h^{3\beta+1}$ for a sufficiently small universal constant $c_0.$ Then, there exist positive universal constants $c_1,c_2$ such that for any $\tau\ge 1$,
with probability at least $1-c_1e^{-c_2\tau}$, it holds for every $\omega \in [-h^{-1},h^{-1}]^d$ that
\[
|\Delta(\omega)|
\lesssim
\|f\|_\infty^{-2}\rho_N \|\omega\|_2 (1+\|\omega\|_2)^{3\beta}.
\]
If in addition $\sup_{\omega \in [-h^{-1},h^{-1}]^d} \frac{\|\nabla \fft(\omega) \|_2}{|\fft(\omega)|}\lesssim 1$ and $\rho_N \le c_0 \|f\|_\infty^2 h^{2\beta+1},$ then the same conclusion holds with the improved bound 
\[
|\Delta(\omega)|
\lesssim
\|f\|_\infty^{-2}\rho_N \|\omega\|_2 (1+\|\omega\|_2)^{2\beta}.
\]

\end{lemma}
\begin{proof}
     Write
\begin{align*}
    \Delta(\omega)
    &=
    \int_0^1
    \frac{(\nabla_1\hatPsi(\alpha\omega,\alpha\omega)-\nabla_1\Psi(\alpha\omega,\alpha\omega))\cdot \omega}
    {\Psi(\alpha\omega,\alpha\omega)}
    d\alpha\\
    &\quad+
    \int_0^1
    \inparen{
    \frac{1}{\tildePsi(\alpha\omega,\alpha\omega)}
    -
    \frac{1}{\Psi(\alpha\omega,\alpha\omega)}
    }
    \nabla_1\Psi(\alpha\omega,\alpha\omega)\cdot \omega
    d\alpha\\
    &\quad+
    \int_0^1
    \inparen{
    \frac{1}{\tildePsi(\alpha\omega,\alpha\omega)}
    -
    \frac{1}{\Psi(\alpha\omega,\alpha\omega)}
    }
    \bigl(\nabla_1\hatPsi(\alpha\omega,\alpha\omega)-\nabla_1\Psi(\alpha\omega,\alpha\omega)\bigr)\cdot \omega
    d\alpha\\
    &=:\Delta_1(\omega)+\Delta_2(\omega)+\Delta_3(\omega).
\end{align*}
It suffices to control $|\Delta_\ell(\omega)|$ for $\ell=1,2,3$. Let $\mathsf E$ denote the event on which
\[
\sup_{\omega\in[-h^{-1},h^{-1}]^d}
|\hatPsi(\omega,\omega)-\Psi(\omega,\omega)|
\lesssim \rho_N,
\qquad
\sup_{\omega\in[-h^{-1},h^{-1}]^d}
\|\nabla_1\hatPsi(\omega,\omega)-\nabla_1\Psi(\omega,\omega)\|_2
\lesssim \rho_N,
\]
which by Lemma~\ref{lem:psiestimates} holds with probability at least $1-c_1e^{-c_2\tau}$. For the remainder of the proof we work on the event $\mathsf E$. We also make frequent use of the elementary fact: $h\in(0,1)$ implies $1+h^{-1}\asymp h^{-1}$. 
Since $\Psi(\omega,\omega)=|\fft(\omega)|^2$, Assumption~\ref{assumption:smoothness} yields
\begin{equation}
\label{eq:Psi-global}
|\Psi(\omega,\omega)|
\asymp
\|f\|_\infty^2 (1+\|\omega\|_2)^{-2\beta}
\qquad\text{for all }\omega\in\R^d.
\end{equation}

Fix $\omega \in [-h^{-1}, h^{-1}]^d$ throughout.

\textit{Controlling $\Delta_1$.}
\begin{align*}
|\Delta_1(\omega)|
&\le
\sup_{0\le \alpha\le 1}
\bigl|
(\nabla_1\hatPsi(\alpha\omega,\alpha\omega)-\nabla_1\Psi(\alpha\omega,\alpha\omega))\cdot \omega
\bigr|
\int_0^1 \frac{d\alpha}{|\Psi(\alpha\omega,\alpha\omega)|}\\
&\le
\|\omega\|_2
\sup_{\xi\in[-h^{-1},h^{-1}]^d}
\|\nabla_1\hatPsi(\xi,\xi)-\nabla_1\Psi(\xi,\xi)\|_2
\int_0^1 \frac{d\alpha}{|\Psi(\alpha\omega,\alpha\omega)|}\\
&\lesssim
\|\omega\|_2 \rho_N
\int_0^1 \frac{d\alpha}{|\Psi(\alpha\omega,\alpha\omega)|}.
\end{align*}
Using \eqref{eq:Psi-global},
\begin{align*}
\int_0^1
\frac{d\alpha}{|\Psi(\alpha\omega,\alpha\omega)|}
&\lesssim
\int_0^1
\|f\|_\infty^{-2}(1+\alpha\|\omega\|_2)^{2\beta}
d\alpha\\
&\le
\int_0^1
\|f\|_\infty^{-2}(1+\|\omega\|_2)^{2\beta}
d\alpha 
=
\|f\|_\infty^{-2}(1+\|\omega\|_2)^{2\beta}.
\end{align*} 
Hence
\[
|\Delta_1(\omega)|
\lesssim
\|f\|_\infty^{-2}\rho_N
\|\omega\|_2 (1+\|\omega\|_2)^{2\beta}.
\]

\textit{Controlling $\Delta_2$.}
\begin{align*}
|\Delta_2(\omega)|
&=
\abs{
\int_0^1
\inparen{
\frac{\nabla_1\Psi(\alpha\omega,\alpha\omega)}{\Psi(\alpha\omega,\alpha\omega)}\cdot\omega
}
\inparen{
\frac{\tildePsi(\alpha\omega,\alpha\omega)-\Psi(\alpha\omega,\alpha\omega)}
{\tildePsi(\alpha\omega,\alpha\omega)}
}
d\alpha
}\\
&\le
\int_0^1
\left\|
\frac{\nabla_1\Psi(\alpha\omega,\alpha\omega)}{\Psi(\alpha\omega,\alpha\omega)}
\right\|_2
\|\omega\|_2
\abs{
\frac{\tildePsi(\alpha\omega,\alpha\omega)-\Psi(\alpha\omega,\alpha\omega)}
{\tildePsi(\alpha\omega,\alpha\omega)}
}
d\alpha\\
&\le
\|\omega\|_2
\sup_{\alpha\in[0,1]}
\abs{
\frac{\tildePsi(\alpha\omega,\alpha\omega)-\Psi(\alpha\omega,\alpha\omega)}
{\tildePsi(\alpha\omega,\alpha\omega)}
}
\int_0^1
\left\|
\frac{\nabla_1\Psi(\alpha\omega,\alpha\omega)}{\Psi(\alpha\omega,\alpha\omega)}
\right\|_2
d\alpha\\
&= \|\omega\|_2 \times \mathbf{(I)}(\omega) \times \mathbf{(II)}(\omega),
\end{align*}
where 
\begin{align*}
    \mathbf{(I)}(\omega) &:= \sup_{\alpha\in[0,1]}
\abs{
\frac{\tildePsi(\alpha\omega,\alpha\omega)-\Psi(\alpha\omega,\alpha\omega)}
{\tildePsi(\alpha\omega,\alpha\omega)}
}, \quad 
    \mathbf{(II)}(\omega) := \int_0^1
\left\|
\frac{\nabla_1\Psi(\alpha\omega,\alpha\omega)}{\Psi(\alpha\omega,\alpha\omega)}
\right\|_2
d\alpha.
\end{align*}
For the numerator of $\mathbf{(I)}(\omega)$, we have
\[
|\tildePsi(\alpha\omega,\alpha\omega)-\Psi(\alpha\omega,\alpha\omega)|
\le 
\frac{2}{\sqrt{N}} + \sup_{\xi\in[-h^{-1},h^{-1}]^d}
|\hatPsi(\xi,\xi)-\Psi(\xi,\xi)|
\lesssim \rho_N,
\]
uniformly over $\alpha \in [0,1]$, where the first inequality holds by \cite[Lemma 3]{kurisu2022uniform}. For the denominator, note that $\Psi(\alpha\omega,\alpha\omega)
=
|\fft(\alpha\omega)|^2$,
and hence, by Assumption~\ref{assumption:smoothness} and since $\alpha \in (0,1)$
\[
|\Psi(\alpha\omega,\alpha\omega)|
\gtrsim
\|f\|_\infty^2(1+\|\alpha\omega\|_2)^{-2\beta}
\gtrsim
\|f\|_\infty^2(1+\|\omega\|_2)^{-2\beta}.
\]

Now, by definition of $\tildePsi$,
\begin{align*}
|\tildePsi(\alpha\omega,\alpha\omega)|
&\ge
|\hatPsi(\alpha\omega,\alpha\omega)|\\
&\ge
|\Psi(\alpha\omega,\alpha\omega)|
-
|\hatPsi(\alpha\omega,\alpha\omega)-\Psi(\alpha\omega,\alpha\omega)|\\
&\gtrsim
\|f\|_\infty^2(1+\|\omega\|_2)^{-2\beta}
-\rho_N.
\end{align*}

Since \(\rho_N \le c_0\|f\|_\infty^2 h^{3\beta+1}\) with $c_0>0$ sufficiently small and \(h\in(0,1)\), we have
\[
\rho_N
\le
c_0
\|f\|_\infty^2 h^{3\beta+1}
\le
c_0
\|f\|_\infty^2 h^{2\beta}.
\]
Moreover, 
since $\| \omega\|_2 \lesssim h^{-1},$
\[
(1+\|\omega\|_2)^{-2\beta}
\gtrsim
(1+h^{-1})^{-2\beta}
\asymp
h^{2\beta}.
\]
Therefore, choosing $c_0>0$ sufficiently small yields 
$|\tildePsi(\alpha \omega,\alpha \omega)| \gtrsim \|f\|^2_\infty (1+\|\omega\|_2)^{-2\beta}$ and consequently that
\[
\mathbf{(I)}(\omega)
=
\sup_{\alpha\in[0,1]}
\abs{
\frac{\tildePsi(\alpha\omega,\alpha\omega)-\Psi(\alpha\omega,\alpha\omega)}
{\tildePsi(\alpha\omega,\alpha\omega)}
}
\lesssim
\frac{\rho_N}{\|f\|_\infty^2}(1+\|\omega\|_2)^{2\beta}.
\]

For $\mathbf{(II)}(\omega)$, direct calculation shows that 
\begin{align*}
    \Psi(\omega_1,\omega_2) 
    &= r_{\yft}(\omega_1,\omega_2)  - k_{\etaft}(\omega_1,\omega_2) 
    = \E[e^{-i(\omega_1-\omega_2)\zeta}] \fft(\omega_1) \overline{\fft(\omega_2)},
\end{align*}
and so 
\begin{align*}
    \nabla_1 \Psi(\omega_1,\omega_2) 
    &=
    -i \E[\zeta e^{-i(\omega_1-\omega_2)\zeta}]  \fft(\omega_1) \overline{\fft(\omega_2)}
    +
    \E[e^{-i(\omega_1-\omega_2)\zeta}]  \overline{\fft(\omega_2)}\nabla \fft(\omega_1),
\end{align*}
which, since $\E[\zeta]=0,$ gives 
$\nabla_1 \Psi(\omega, \omega)= 
     \overline{\fft(\omega)}\nabla \fft(\omega)$. Therefore
\[
\left\|
\frac{\nabla_1\Psi(\omega,\omega)}{\Psi(\omega,\omega)}
\right\|_2
\lesssim
\left\|
\frac{\nabla \fft(\omega)}{\fft(\omega)}
\right\|_2.
\]
Since $f$ is compactly supported,
\[
\|\nabla \fft(\omega)\|_2
\le
\int_{\R^d}\|t\|_2 |f(t)|dt
\lesssim
\|f\|_\infty,
\]
and by Assumption~\ref{assumption:smoothness},
$|\fft(\omega)|
\gtrsim
\|f\|_\infty (1+\|\omega\|_2)^{-\beta}$. Thus
\[
\left\|
\frac{\nabla \fft(\omega)}{\fft(\omega)}
\right\|_2
\lesssim
(1+\|\omega\|_2)^\beta.
\]
Hence, 
\begin{align*}
\mathbf{(II)}(\omega)
&=
\int_0^1
\left\|
\frac{\nabla_1\Psi(\alpha\omega,\alpha\omega)}{\Psi(\alpha\omega,\alpha\omega)}
\right\|_2
d\alpha 
\lesssim
\int_0^1
(1+\|\alpha\omega\|_2)^\beta d\alpha \\
&\le
\int_0^1
(1+\|\omega\|_2)^\beta d\alpha 
\lesssim
(1+\|\omega\|_2)^\beta.
\end{align*}
Combining the bounds for $\mathbf{(I)}(\omega)$ and $\mathbf{(II)}(\omega)$, we obtain for every $\omega \in [-h^{-1}, h^{-1}]^d$
\begin{align*}
|\Delta_2(\omega)|
&\le
\|\omega\|_2\mathbf{(I)}(\omega)\mathbf{(II)}(\omega)
\lesssim
\|\omega\|_2
\cdot
\frac{\rho_N}{\|f\|_\infty^2}(1+\|\omega\|_2)^{2\beta}
\cdot
(1+\|\omega\|_2)^\beta\\
&\lesssim
\frac{\rho_N}{\|f\|_\infty^2}
\|\omega\|_2(1+\|\omega\|_2)^{3\beta}.
\end{align*}

Under the additional assumption, $\sup_{\omega \in [-h^{-1},h^{-1}]^d} \frac{\|\nabla \fft(\omega) \|_2}{|\fft(\omega)|}\lesssim 1$, it follows immediately that $\mathbf{(II)}(\omega) \lesssim 1,$ and so we instead obtain 
\begin{align*}
|\Delta_2(\omega)|
&\lesssim
\frac{\rho_N}{\|f\|_\infty^2}
\|\omega\|_2(1+\|\omega\|_2)^{2\beta}.
\end{align*}

\textit{Controlling $\Delta_3$.}
Using the same estimates as above,
\begin{align*}
|\Delta_3(\omega)|
&\le
\sup_{\alpha\in[0,1]}
\abs{
\frac{\tildePsi(\alpha\omega,\alpha\omega)-\Psi(\alpha\omega,\alpha\omega)}{\tildePsi(\alpha\omega,\alpha\omega)\Psi(\alpha\omega,\alpha\omega)}
}
\int_0^1
\bigl|
(\nabla_1\hatPsi(\alpha\omega,\alpha\omega)-\nabla_1\Psi(\alpha\omega,\alpha\omega))\cdot \omega
\bigr|d\alpha.
\end{align*}
For the first term, arguing as in the control of $\Delta_2$, we obtain
\begin{align*}
    \sup_{\alpha\in[0,1]}
\abs{
\frac{\tildePsi(\alpha\omega,\alpha\omega)-\Psi(\alpha\omega,\alpha\omega)}{\tildePsi(\alpha\omega,\alpha\omega)\Psi(\alpha\omega,\alpha\omega)}
}
\lesssim 
\frac{\rho_N}{\|f\|^4_\infty} (1+\|\omega\|_2)^{4\beta}.
\end{align*}
For the second term, 
\begin{align*}
    \int_0^1
\bigl|
(\nabla_1\hatPsi(\alpha\omega,\alpha\omega)-\nabla_1\Psi(\alpha\omega,\alpha\omega))\cdot \omega
\bigr|d\alpha
&\le \|\omega\|_2 \sup_{\omega \in [-h^{-1}, h^{-1}]^d} \| \nabla_1 \hatPsi(\omega,\omega)-\nabla_1 \Psi(\omega,\omega)\|\\
&\lesssim \|\omega\| \rho_N.
\end{align*}
Combining the two bounds yields
\begin{align*}
    |\Delta_3(\omega)| \lesssim
    \rho_N^2 \|f\|^{-4}_\infty \|\omega\|_2 (1+\|\omega\|_2)^{4\beta}.
\end{align*}

Since $\rho_N \le c_0 \|f\|_\infty^2 h^{3\beta+1}$ with $c_0>0$ sufficiently small and $h\in(0,1)$, it follows that
\begin{align*}
    |\Delta_3(\omega)| \le
    c
    \rho_N \|f\|^{-2}_\infty \|\omega\|_2 (1+\|\omega\|_2)^{3\beta}.
\end{align*}
for sufficiently small $c>0$. Under the additional assumption, $\sup_{\omega \in [-h^{-1},h^{-1}]^d} \frac{\|\nabla \fft(\omega) \|_2}{|\fft(\omega)|}\lesssim 1$, and with $\rho_N \le c_0 \|f\|_\infty^2 h^{2\beta+1}$, we similarly obtain
\begin{align*}
    |\Delta_3(\omega)| \le
    c
    \rho_N \|f\|^{-2}_\infty \|\omega\|_2 (1+\|\omega\|_2)^{2\beta}.
\end{align*}
Therefore, in both cases, the contribution of $\Delta_3$ is of strictly smaller order and can be absorbed into the bounds for $\Delta_2$.

Combining the bounds for $\Delta_1,\Delta_2,\Delta_3$, we conclude that
\[
|\Delta(\omega)|
\lesssim
\|f\|_\infty^{-2}\rho_N \|\omega\|_2 (1 + \|\omega\|_2)^{3\beta}.
\]
Under the additional assumption, $\sup_{\omega \in [-h^{-1},h^{-1}]^d} \frac{\|\nabla \fft(\omega) \|_2}{|\fft(\omega)|}\lesssim 1$, it follows instead that
\[
|\Delta(\omega)|
\lesssim
\|f\|_\infty^{-2}\rho_N \|\omega\|_2 (1 + \|\omega\|_2)^{2\beta}.
\]
\end{proof}

\begin{proof}[Proof of Theorem~\ref{thm:FTDeviation}]
Let $\mathsf E$ denote the event from the proof of Lemma~\ref{lem:DeltaControl}, which by Lemma~\ref{lem:psiestimates}, holds with probability at least
$1-c_1e^{-c_2\tau}$. We work on $\mathsf E$ throughout the remainder of the proof.

Let $\Delta(\omega)$ be the quantity defined in Lemma~\ref{lem:DeltaControl}. By construction, $\hatfft(\omega)=\fft(\omega)e^{\Delta(\omega)}$, and so
\[
\hatfft(\omega)-\fft(\omega)
=
\fft(\omega)\bigl(e^{\Delta(\omega)}-1\bigr).
\]
Moreover, by Lemma~\ref{lem:DeltaControl}, 
\[
\sup_{\omega\in[-h^{-1},h^{-1}]^d} |\Delta(\omega)|
\lesssim
\|f\|_\infty^{-2}\rho_N h^{-1} (1+h^{-1})^{3\beta}
\lesssim
\|f\|_\infty^{-2}\rho_N h^{-3\beta-1}.
\]
Since $\rho_N\le c_0\|f\|_\infty^2 h^{3\beta+1}$ with $c_0>0$ sufficiently small, we have
$\sup_{\omega\in[-h^{-1},h^{-1}]^d} |\Delta(\omega)|\le 1$. Furthermore, by 
Assumption~\ref{assumption:smoothness}, $|\fft(\omega)|\asymp \|f\|_\infty(1+\|\omega\|_2)^{-\beta}.$ 
Therefore
\begin{align*}
  |\hatfft(\omega)-\fft(\omega)|
&\le 2|\fft(\omega)\Delta(\omega)| \\
&\lesssim
\|f\|_\infty(1+\|\omega\|_2)^{-\beta}
\|f\|_\infty^{-2}\rho_N
\|\omega\|_2(1+\|\omega\|_2)^{3\beta} \\
&\lesssim
\|f\|_\infty^{-1}\rho_N
\|\omega\|_2(1+\|\omega\|_2)^{2\beta},
\end{align*}

where the first inequality follows using the fact $|1-e^z| \le 2 |z|$ for $z \in \C$ with $|z| \le 1.$ It follows that the supremum is bounded by $\rho_N \|f\|_\infty^{-1} h^{-2\beta-1}.$ Under the additional assumption, the proof follows using an identical argument, invoking instead the second part of Lemma~\ref{lem:DeltaControl}.
\end{proof}

\subsection{Supplementary Materials for Subsection~\ref{subsec:snr}}\label{app:highNoiseRegime}
\begin{proof}[Proof of Corollary~\ref{corr:highNoiseRegime}]
        Under the additional assumptions of this subsection, we have 
        \begin{align*}
            \|r_y\|_\infty \frac{\Gamma(y)}{\sqrt{N}} 
            &= \|r_y\|_\infty^{1/2} \frac{\E[\sup_{t\in D} y(t)]}{\sqrt{N}} 
            \le \frac{1}{\sqrt{N}} (\|f\|_\infty + \|k_\eta \|^{1/2}_\infty) ( \|f\|_\infty +  \E [\sup_{t \in D} \eta(t)])\\
            &\lesssim  \frac{1}{\sqrt{N}} \|k_\eta \|^{1/2}_\infty \E [\sup_{t \in D} \eta(t)]
            = \frac{\mathsf{snr}^{-1} \|f\|_\infty^2}{\sqrt{N}}.
        \end{align*}
        The last inequality follows since by assumption $\|k_\eta\|^{1/2}_\infty \ge \|f\|_\infty$, and the fact that for any centered Gaussian process, we have, for any $t \in D,$
    $$
    \E[\sup_{t\in D} \eta(t)]
    \asymp 
    \E[\sup_{t\in D} |\eta(t)|]
    \ge 
    \E|\eta(t)|
    \ge 
    \sqrt{\E|\eta(t)|^2}
    = \sqrt{k_\eta(t,t)},
    $$
    where the last inequality holds by Jensen's inequality. Theorem~\ref{thm:ErrorBoundDeconvolutionEstimator} then implies that with probability at least $1-c_1e^{-c_2 \tau },$ 
    \begin{align}\label{eq:corr1-prelim}
        \frac{\|\hatf - f\|_\infty}{\|f\|_\infty}
        \lesssim
        \inparen{\|r_y\|_{\infty}
        \sqrt{\frac{\tau}{N}}
        \lor 
        \frac{\mathsf{snr}^{-1}}{\sqrt{N}}} h^{-2\beta - d - 1} + h^{\beta-d}.
    \end{align}
    The optimal bandwidth is chosen to minimize the right-hand side of the above display, and a calculation yields 
    \begin{align*}
        h^* \asymp \inparen{
        \|r_y\|_{\infty}
        \sqrt{\frac{\tau}{N}}
        \lor 
        \frac{\mathsf{snr}^{-1}}{\sqrt{N}}
        }^{\frac{1}{3\beta  + 1}},
    \end{align*}
    which can be plugged into \eqref{eq:corr1-prelim} to give the desired relative error bound, after noting that  $\| r_y\|_\infty \le \|f\|_\infty^2 + \|k_\eta\|_\infty \lesssim \|k_\eta\|_\infty$.
    \end{proof}

\begin{proof}[Proof of Corollary~\ref{corr:largesigma}]
         Clearly, $k_\eta^{\text{SE}}(t,t) = \sigma^2$ for any $t \in D$. Further, we may write $\eta = \sigma z$ where $z$ is a Gaussian process with covariance function $\exp(-\|t-t'\|^2_2/2\lambda^2),$ and so $\E[\sup_{t \in D} \eta(t)]=\sigma \E[\sup_{t \in D} z(t)],$ with the second term having no dependence on $\sigma.$ Therefore, 
         $\E[\sup_{t \in D} \eta(t)] \asymp \sigma.$ 
         The result follows from Corollary~\ref{corr:highNoiseRegime} after noting that $\Gamma(y) \asymp 1$ and $\|k_\eta\|_\infty = \sigma^2.$
    \end{proof}

\begin{proof}[Proof of Corollary~\ref{corr:smallLengthscale}]
        Clearly, $k_\eta^{\text{SE}}(t,t) = \sigma^2$ for any $t \in D$, and by \cite[Lemma 4.3]{al2025covariance}, there exists a universal constant $\lambda_0>0$ such that, for any $\lambda < \lambda_0,$ $\E [\sup_{t \in D} \eta(t)] \asymp \sigma \sqrt{\log\frac{1}{\lambda^d}}.$ 
       The result follows  from Corollary~\ref{corr:highNoiseRegime} by taking $\tau = \Gamma(y)^2$ and noting that 
        \begin{align*}
            \Gamma(y) 
            = \frac{\E[\sup_{t \in D} y(t)]}{\|r_y\|_\infty^{1/2}}
            \le 
            \frac{\|f\|_\infty + \E[\sup_{t \in D} \eta(t)]}{\|k_\eta\|_\infty^{1/2}}
            \lesssim
             \frac{\E[\sup_{t \in D} \eta(t)]}{\|k_\eta\|_\infty^{1/2}} \asymp \sqrt{\log \lambda^{-d}}.
        \end{align*}
    \end{proof}

\section{Supplementary Materials for Section~\ref{sec:vanishing}}
\label{app:vanishing}

\begin{lemma}[Power spectrum bounds]\label{PSderivbounds}
Suppose $f$ satisfies Assumption \ref{assumption:psiestimates}. Let $P(\omega) = \Psi(\omega, \omega)=\fft(\omega)\fft(-\omega)$. Then, $\|P^{(\ell)}\|_\infty \lesssim \|f\|^2_\infty$ for $\ell=0,1,2,3.$
\end{lemma}

\begin{proof}
Direct calculation shows:
\begin{align*}
	P(\omega) &= \Psi(\omega, \omega) = \fft(\omega) \fft(-\omega),\\
	P'(\omega) &= \partial_1 \Psi(\omega, \omega) + \partial_2 \Psi(\omega, \omega)=  (\fft)'(\omega) \fft(-\omega)-\fft(\omega) (\fft)'(-\omega),\\
	P''(\omega)&= 
	\partial_1^2 \Psi(\omega, \omega) + \partial_2^2 \Psi(\omega, \omega)
	+2\partial_2\partial_1 \Psi(\omega, \omega)\\
	&= 
	(\fft)''(\omega)\fft(-\omega) + \fft(\omega)(\fft)''(-\omega) - 2 (\fft)'(\omega) (\fft)'(-\omega),\\
	P'''(\omega) &= \partial_1^3 \Psi(\omega,\omega)+ \partial_2^3\Psi(\omega,\omega) + 3\partial_2 \partial_1^2\Psi(\omega,\omega) + 3\partial_1 \partial_2^2\Psi(\omega,\omega)\\
	&= (\fft)'''(\omega)\fft(-\omega)-(\fft)'''(-\omega)\fft(\omega)\\
	&\quad+(\fft)'(-\omega)(\fft)''(\omega)-(\fft)'(\omega)(\fft)''(-\omega).
\end{align*}
The result follows immediately by Lemma~\ref{lem:fftnbound}.
\end{proof}

\begin{proof}[Proof of Lemma \ref{zeros_location_control}]
	Suppose that $h\in(0,1)$. We will make repeated use of the fact that for such an $h$, $1+\hinv \asymp \hinv.$  Defining $ P^{(\ell)}(\omega),  \ell = 0,1,2,3$, to be the power spectrum and its derivatives, as in Lemma \ref{PSderivbounds}, note that $P^{(\ell)}$ is real for $\ell = 0,1,2,3$. Likewise, since by assumption $\hatP$ is real, it follows that $\hatP^{(\ell)}$ is real for $\ell = 0,1,2$. By Lemma \ref{lem:psiestimates} and a slight modification of its proof to handle estimation of $\nabla_2\Psi$, it follows that there exist universal constants $ c_1,c_2 $ such that, for any $ \tau \geq 1 $, it holds with probability at least $ 1-c_1 e^{-c_2 \tau} $ that
	\[
	\sup_{\omega\in[0,\hinv]} | P^{(\ell)}(\omega)-\hatP^{(\ell)}(\omega) | 
	\lesssim \rho_N, \qquad \ell=0,1,2.
	\]
    Define $\checkep := h^{-\beta}\|f\|_\infty^{-1}\sqrt{\thresh}$. Then, by Assumption \ref{assumption:thresholding}, $\xi_{k(\hinv)}+h^{-\beta}\|f\|_\infty^{-1}\sqrt{\thresh} \leq \hinv$, so we can assume going forwards that $[\xi_i-\checkep, \xi_i+\checkep] \subset [0,\hinv]$ for all $\xi_i \in Z_P$. Additionally, by the assumptions,
    \begin{align}
        \checkep
        &\lesssim h^{2\beta} \leq \epsilon_{\max},\label{eq:checkep}\\
        \rho_N &\lesssim \checkep^2h^{4\beta}\|f\|_\infty^2.\label{eq:rhocheckep}
    \end{align}
    We will reference these bounds later on. Note that $ \checkep \leq \epsilon_{\max}$ in \eqref{eq:checkep} implies that there exist $L,B>0$ satisfying the inequalities in Assumption \ref{assumption:nice_oscillations}.\\ \textbf{(I)} We first show that
\begin{equation}\label{eq:P''lowerbound}
    	\min_{\xi_i\in Z_{P}} |\omega-\xi_i| \leq \checkep \implies |P''(\omega)| \gtrsim h^{2\beta}\|f\|_\infty^2 .
\end{equation}
	Let $ \xi_i \in Z_P $ and $ \checkep \in [0, \epsilon_{\max}] $. Note that $ Z_P $ and $ Z_{\fft} $ are identical. Then, $(\fft)'(\xi_i) = o_f'(\xi_i)(1+\xi_i)^{-\beta} - \beta o_f(\xi_i)(1+\xi_i)^{-\beta-1} = o_f'(\xi_i)(1+\xi_i)^{-\beta}$, and so by Assumption \ref{assumption:nice_oscillations} and the fact that $ \xi_i\leq \hinv $, $ |P''(\xi_i)| = 2|(\fft)'(\xi_i)|^2 \geq 2B^2\| f\|_\infty^2 (1+\hinv)^{-2\beta} \gtrsim 2\| f\|_\infty^2 h^{2\beta}$. By the mean value theorem, there exists a $ \omega_1 \in [\xi_i,\xi_i+\checkep] $ such that 
	\[
	P''(\xi_i+\checkep) =  P''(\xi_i) + P'''(\omega_1)\cdot \checkep.
	\]
	By Lemma \ref{PSderivbounds}, $ |P'''(\omega)| \lesssim \| f\|_\infty^2$, so there exists some constant $C_3 > 0 $ such that
	\[
	|P''(\xi_i+\checkep)| \geq \left|  |P''(\xi_i)| - C_3\|f\|_\infty^2\checkep \right|.
	\] 
	   By \eqref{eq:checkep}, $ \checkep \lesssim h^{2\beta}$, so $|P''(\xi_i+\checkep)| \gtrsim h^{2\beta}\|f\|_\infty^2$. Extending to $ \xi_i-\checkep $ without loss of generality yields
	\[
	\min_{\xi_i\in Z_{P}} |\omega-\xi_i| \leq \checkep \implies |P''(\omega)| \gtrsim h^{2\beta}\|f\|_\infty^2 .
	\]
	\textbf{(II)} Next, we show that     \begin{equation}\label{eq:exists0}
        \forall \xi_i \in Z_P,  \exists \hatxi_i \in Z_{\hatP'} :  \hatxi_i \in  [\xi_i-\checkep, \xi_i+\checkep].
    \end{equation} 
    Suppose this were not true, and fix a $\xi_i\in Z_P$. It follows that $\hatP'$ could not have a zero on $(\xi_i-\checkep, \xi_i+\checkep)$, and so without loss of generality, we say that $\hatP'(\omega) > 0$ for all $\omega \in [\xi_i-\checkep, \xi_i+\checkep]$. Specifically, we have (i) $\hatP'(\xi_i+\checkep) > 0$. As $Z_P \subset Z_{P'}$, by (\ref{eq:P''lowerbound}), $|P''(\omega)| \gtrsim h^{2\beta}\|f\|_\infty^2 >0$ on $(\xi_i-\checkep, \xi_i+\checkep)$. Thus, $P'$ must be monotonically increasing or decreasing on $(\xi_i-\checkep, \xi_i+\checkep)$. Note that $P'(\xi_i) = 0$. Since $[\xi_i-\checkep, \xi_i+\checkep] \subset [0,\hinv]$ for all $\xi_i \in Z_P$, without loss of generality we can say that $P'(\xi_i-\checkep) > 0$ and (ii) $P'(\xi_i+\checkep)<0$. Additionally, we observe that (iii) $|P'(\xi_i+\checkep)| \geq B^2h^{2\beta}\|f\|_\infty^2 \checkep$, as otherwise by the mean value theorem there would exist an $\omega_2 \in (\xi_i, \xi_i+\checkep)$ such that 
	\begin{align*}
		B^2h^{2\beta}\|f\|_\infty^2 \leq |P''(w_2)| &= \frac{|P'(\xi_i+\checkep)-P'(\xi_i))|}{|\xi_i+\checkep-\xi_i|} = \frac{|P'(\xi_i+\checkep)|}{\checkep}\\
		&<B^2h^{2\beta}\|f\|_\infty^2 \checkep/\checkep = B^2 h^{2\beta}\|f\|_\infty^2 ,
	\end{align*}
	yielding $B^2h^{2\beta}\|f\|_\infty^2 < B^2h^{2\beta}\|f\|_\infty^2 $, a contradiction. Collecting all the pieces, we have that, since 
    \begin{enumerate}
        \item $P'(\xi_i+\checkep)<0;$ 
        \item $|P'(\xi_i+\checkep)| \gtrsim  h^{2\beta}\| f\|_\infty^2;$
        \item $\hatP'(\xi_i+\checkep) > 0$,
    \end{enumerate}
    it follows that
	\[
	h^{2\beta}\|f\|_\infty^2 \checkep \lesssim | P'(\xi_i+\checkep)| < | P'(\xi_i+\checkep) - \hatP'(\xi_i+\checkep) | < \rho_N,
	\]
	a contradiction by \eqref{eq:rhocheckep}, and so there must exist a $\hatxi_i\in (\xi_i-\checkep, \xi_i+\checkep)$ such that $\hatP'(\hatxi_i) = 0.$ \\
	\textbf{(III)} Next, we show that 
\begin{equation}\label{eq:unique0}
        \forall \xi_i \in Z_P,  \exists! \hatxi_i \in Z_{\hatP'} :  \hatxi_i \in  [\xi_i-\checkep, \xi_i+\checkep].
    \end{equation} 
	We have already shown in \textbf{(II)} that there exists such a $\hatxi_i$. Now we show that such a $\hatxi_i$ is unique. We proceed by contradiction. If such a $\hatxi_i$ was not unique in $(\xi_i-\checkep, \xi_i+\checkep)$, i.e. there existed a $\hatxi_i^1 < \hatxi_i^2$ such that $\hatP'(\hatxi_i^1) = \hatP'(\hatxi_i^2) = 0$, then there would exist a $\omega_3 \in (\hatxi_i^1, \hatxi_i^2)$ such that $\hatP''(\omega_3) = 0$. But this is a contradiction, as if $\hatP''(\omega_3) = 0$,
	\[
	h^{2\beta}\|f\|_\infty^2  \lesssim   \left|  |P''(\omega_3)| - | \hatP''(\omega_3) |   \right| \leq |P''(\omega_3) - \hatP''(\omega_3)| < \rho_N,
	\]
	which is impossible by \eqref{eq:rhocheckep}.\\
    \textbf{(IV)} Define $ \delta = \rho_N/(h^{2\beta}\|f\|_\infty^2 )$. Finally, we show that
    \begin{equation}\label{eq:0bound}
        \forall  \xi_i\in Z_P, |\xi_i-\hatxi_i| \lesssim \delta,
    \end{equation}
    where $\hatxi_i$ is the unique element of $Z_{\hatP'}$ in $[\xi_i-\checkep, \xi_i+\checkep]$ that we attain in equations (\ref{eq:exists0}) and (\ref{eq:unique0}). Indeed, by the mean value theorem and equation (\ref{eq:P''lowerbound}), there exists a $\omega_4 \in (\xi_i-\checkep, \xi_i+\checkep)$ such that
	\[
	h^{2\beta}\|f\|_\infty^2 \lesssim |P''(\omega_4)| = \frac{|P'(\xi_i)-P'(\hatxi_i)|}{|\xi_i-\hatxi_i|}.
	\] 
    Note that $P'(\xi_i)=0$ and $\hatP'(\hatxi_i)=0$. Thus by our assumptions, $|P'(\hatxi_i)-\hatP'(\hatxi_i)| \leq \rho_N$, so $|P'(\hatxi_i)| \leq \rho_N$.
    Consequently,
    \[
    h^{2\beta}\|f\|_\infty^2 \lesssim \frac{|P'(\xi_i)-P'(\hatxi_i)|}{|\xi_i-\hatxi_i|} = \frac{|P'(\hatxi_i)|}{|\xi_i-\hatxi_i|}
    \leq \frac{\rho_N}{|\xi_i-\hatxi_i|},
    \]
    which ultimately implies that
    \[ |\xi_i-\hatxi_i|  \lesssim \rho_N/(h^{2\beta}\|f\|_\infty^2), 
    \]
    as desired. We have now shown that for each $ \xi_i \in Z_P $, there exists a unique $ \hatxi_i \in Z_{\hatP'} $ such that $ |\xi_i - \hatxi_i| \lesssim\delta $.\\	
	\textbf{(V)} Recall that $\checkep := h^{-\beta}\|f\|_\infty^{-1}\sqrt{2\thresh}$ and let $\rho_N h^{-2\beta} \lesssim \thresh \lesssim \|f\|_\infty^2h^{6\beta}$, as assumed. We now show that for 
    \begin{align}
        \hatxi_i\in Z_{\hatP'} \text{ and }\hatP(\hatxi_i) > \thresh &\implies\hatxi_i \not\in Z_P^{\checkep},\\
        \hatxi_i\in Z_{\hatP'} \text{ and } \hatP(\hatxi_i) < \thresh &\implies \hatxi_i \in Z_P^{\checkep}.
    \end{align}

	We proceed by proving the contrapositive of the two statements. Suppose that $\hatxi_i \not\in \cup_{\xi_i\in Z_P} (\xi_i-\checkep, \xi_i+\checkep)$. Then $|\hatxi_i-Z_P| \geq \checkep$ and so $|\hatxi_i-Z_{\fft} | \geq \checkep$. Thus by Assumption \ref{assumption:nice_oscillations}, $|P(\hatxi_i)| \geq L^2\|f\|_\infty^2\checkep^2  (1+\hinv)^{-2\beta} \gtrsim \|f\|_\infty^2\checkep^2h^{2\beta}$. Additionally, as $| \hatP(\hatxi_i)-P(\hatxi_i) | < \rho_N$, $\hatP(\hatxi_i) \geq P(\hatxi_i)-\rho_N \gtrsim \|f\|_\infty^2\checkep^2h^{2\beta}- \rho_N$, so since $\rho_B \lesssim h^{2\beta}\thresh$ and by \eqref{eq:rhocheckep} $\rho_N \lesssim \checkep^2h^{4\beta}\|f\|_\infty^2$, it must be that $\hatP(\hatxi_i) > \thresh$, completing the first part of the proof.
	
	Now suppose that $\hatxi_i \in\cup_{\xi_i\in Z_P} (\xi_i-\checkep, \xi_i+\checkep)$. Then, since the zeros of $\fft$ and therefore $P$ are well isolated, there exists a unique $\xi_i\in Z_P$ such that $|\xi_i - \hatxi_i| < \checkep$. But we have already shown that for such a $\xi_i\in Z_P$, there can be at most one $\hatxi_i\in Z_{\hatP'}$ such that $|\xi_i - \hatxi_i| < \checkep$ and in fact for such a $\hatxi_i$, it must be that $|\xi_i - \hatxi_i| \lesssim\delta$.
	
	By the mean value theorem, assuming without loss of generality that $\xi_i < \hatxi_i$, there exists an $\omega_5 \in [\xi_i, \hatxi_i]$ such that $|P'(\omega_5)| = | P(\hatxi_i)-P(\xi_i)|/ |\hatxi_i - \xi_i|$. But $P(\xi_i) = 0$, so by Lemma \ref{PSderivbounds} there exists a constant $C >0$ such that
	\[
	|P(\hatxi_i) |= | P'(\omega^*)|\cdot |\hatxi_i - \xi_i| \leq \frac{C}{h^{2\beta}\|f\|_\infty^2}\|f\|_\infty^2\rho_N = {C\cdot h^{-2\beta}}\rho_N.
	\]
	Finally, $\hatP(\hatxi_i) < P(\hatxi_i) + \rho_N < ({h^{-2\beta }C}+1)\rho_N$. Since by \eqref{eq:rhocheckep}, $\rho_N \lesssim \|f\|_\infty^2\checkep^2 h^{4\beta}$, it must be that $\hatP(\hatxi_i) < \thresh$.\\	
	\textbf{(VI)} We conclude by defining $ \mathcal{Z}_\thresh := \{ \hatxi \in Z_{\hatP'}  :  \hatP(\xi) < \thresh  \} $.  We began by showing that for any element $\xi \in Z_P$ there is a unique element $\hatxi \in Z_{\hatP'}$ such that $|\xi - \hatxi| \lesssim \delta$, and now we have shown that every element of $\mathcal{Z}_\thresh$ is in $\cup_{\xi_i\in Z_P} (\xi_i-\checkep, \xi_i+\checkep)$ and thus is $\delta $ away from a unique element of $Z_P$. Thus $|\mathcal{Z}_\thresh| = |Z_P|$ and the lemma is proved.
\end{proof}

\begin{proof}[Proof of Lemma \ref{lem:error_from_zeros}]
	We first show that \eqref{equ:error_intermediate} holds when $\omega \in \hatomega_\epsilon$, i.e. $|\omega - \hatxi_i|\geq \epsilon$ for all the estimated zeros $\hatxi_i$. Note that in this case, $\hatk(\omega)=k(\omega)$, and for brevity we denote the number of zeros less than $\omega$ simply by $k$ in the following argument. Note that $\frac{\partial_{1} \Psi(\xi,\xi)}{\Psi(\xi,\xi)}=\frac{(\fft)'(\xi)}{\fft(\xi)}$ and so $\frac{\fft(\xi_2)}{\fft(\xi_1)}=\exp\left(\int_{\xi_1}^{\xi_2} \frac{\partial_{1} \Psi(\xi,\xi)}{\Psi(\xi,\xi)} d\xi\right)$ by Lemma \ref{lem:technicalKotlarski}, as long as $\fft \ne 0$ on the straight-line path connecting $\xi_1,\xi_2$.  We thus obtain: 
	\begin{align}
&\exp\left(\int_{\hatomega_\epsilon} \frac{\partial_{1} \Psi(\xi,\xi)}{\Psi(\xi,\xi)} d\xi\right) \nonumber \\&  = \exp\left(\int_0^{\hatxi_1^-} \frac{\partial_{1} \Psi(\xi,\xi)}{\Psi(\xi,\xi)} d\xi\right)\exp\left(\int_{\hatxi_1^+}^{\hatxi_2^-}\frac{\partial_{1} \Psi(\xi,\xi)}{\Psi(\xi,\xi)} d\xi\right) \cdots \nonumber
		\exp\left(\int_{\hatxi_{k}^+}^{\omega}\frac{\partial_{1} \Psi(\xi,\xi)}{\Psi(\xi,\xi)} d\xi\right) \\
		\hspace{1cm}& = \left(\frac{\fft(\hatxi_1^-)}{\fft(\hatxi_1^+)}\right)\left(\frac{\fft(\hatxi_2^-)}{\fft(\hatxi_2^+)}\right) \cdots \left(\frac{\fft(\hatxi_k^-)}{\fft(\hatxi_k^+)}\right)\fft(\omega) \label{equ:initial_formula} ,
	\end{align}
	where $\hatxi^{-}_i = \hatxi_i-\epsilon$, $\hatxi^{+}_i=\hatxi_i+\epsilon$.
	Taylor expanding $\fft$ about $\xi_i$, we have 
	\begin{align*}
		\fft(\xi) &= (\fft)'(\xi_i)(\xi-\xi_i) +  \frac{(\fft)''(\xi_i)}{2}(\xi-\xi_i)^2 +\frac{(\fft)^{(3)}(\xi^*)}{6}(\xi-\xi_i)^3 \\
		&:= c_1(\xi-\xi_i)+c_2(\xi-\xi_i)^2+c_3(\xi-\xi_i)^3
	\end{align*}
	for some $\xi^*$ on the line segment containing $\xi_i$ and $\xi$. Plugging in $\hatxi_i^+,\hatxi_i^-$ and letting $\Delta \xi_i = \hatxi_i - \xi_i$ gives:
	\begin{align*}
		\fft(\hatxi_i^+) &= c_1(\hatxi_i +\epsilon - \xi_i)+ c_2(\hatxi_i +\epsilon - \xi_i)^2 + c^+_3(\hatxi_i +\epsilon-\xi_i)^3 \\
		&= c_1(\epsilon +\Delta \xi_i)+ c_2(\epsilon +\Delta \xi_i)^2 + O(\epsilon^3), \\
		\fft(\hatxi_i^-) &= c_1(\hatxi_i - \epsilon - \xi_i)+ c_2(\hatxi_i -\epsilon - \xi_i)^2 + c^-_3(\hatxi_i - \epsilon-\xi_i)^3 \\
		&= c_1(-\epsilon +\Delta \xi_i)+ c_2(- \epsilon +\Delta \xi_i)^2 + O(\epsilon^3)  .
	\end{align*}
	We thus obtain:
	\begin{align*}
		\frac{\fft(\hatxi_i^+)}{\fft(\hatxi_i^-)} &= \frac{ c_1(\epsilon +\Delta \xi_i)+ c_2(\epsilon +\Delta \xi_i)^2 + O(\epsilon^3) }{c_1(-\epsilon +\Delta \xi_i)+ c_2(- \epsilon +\Delta \xi_i)^2 + O(\epsilon^3) } \\
		&= \frac{ \frac{\epsilon +\Delta \xi_i}{-\epsilon +\Delta \xi_i} + \frac{c_2}{c_1}\frac{(\epsilon +\Delta \xi_i)^2}{(-\epsilon +\Delta \xi_i)} + O(\epsilon^2) }{ 1 + \frac{c_2}{c_1}(-\epsilon +\Delta \xi_i) + O(\epsilon^2)} \\
		&= -1 \frac{ \left(\frac{\epsilon +\Delta \xi_i}{\epsilon -\Delta \xi_i} + \frac{c_2}{c_1}\frac{(\epsilon +\Delta \xi_i)^2}{(\epsilon -\Delta \xi_i)} + O(\epsilon^2) \right)}{ \left(1 + \frac{c_2}{c_1}(-\epsilon +\Delta \xi_i) + O(\epsilon^2)\right)}
		:= -1\cdot \frac{\textbf{(I)}}{\textbf{(II)}} ,
	\end{align*}
	where $\textbf{(I)}$ denotes the numerator and $\textbf{(II)}$ the denominator. We will show that both factors are close to 1, starting with $\textbf{(I)}$. Note:
	\begin{align*}
		\textbf{(I)} &=  \textbf{(A)} + \textbf{(B)} + O(\epsilon^2) , \quad\quad \textbf{(A)} = \frac{1+\frac{\Delta \xi_i}{\epsilon}}{1-\frac{\Delta \xi_i}{\epsilon}}  , \quad\quad \textbf{(B)} = \frac{c_2}{c_1}\textbf{(A)}(\epsilon+\Delta\xi_i) .
	\end{align*}
	We thus obtain:
	\begin{align}
		\label{equ:bound_for_A}
		\textbf{(A)} = 1 + O\left(\delta \epsilon^{-1}\right)  , \quad\quad \textbf{(B)} = \frac{c_2}{c_1}\left(1 + O\left(\delta \epsilon^{-1}\right)\right)(\epsilon+\Delta\xi_i) = \frac{c_2}{c_1}O(\epsilon) .
	\end{align}
    Note that by Assumption \ref{assumption:FT_decay}, we have at the zero $\xi_i$:
	\begin{align*}
		(\fft)'(\xi_i) &= \frac{o'(\xi_i)}{1+\xi_i^\beta}  , \quad\quad {(\fft)''(\xi_i) =-2\beta\frac{o'(\xi_i)}{1+\xi_i^{\beta+1}} + \frac{o''(\xi_i)}1+\xi_i^\beta}  ,
	\end{align*}
	so that utilizing $\xi_i\geq 1$ we obtain
	\begin{align*}
		\left|\frac{c_2}{c_1}\right| &= \left|\frac{(\fft)''(\xi_i)}{2(\fft)'(\xi_i)}\right| \leq \frac{\beta}{ 1+\xi_i} + \left|\frac{o_f''(\xi_i)}{2o_f'(\xi_i)}\right| \leq \beta + \frac{C}{2B},
	\end{align*}
    which ensures
    \begin{align*}
    \textbf{(I)} &= 1+O(\delta\epsilon^{-1}+\epsilon).
    \end{align*}
	Next we bound $\textbf{(II)}$. Note:
	\begin{align*}
		\textbf{(II)} &= 1+ O(\epsilon) \quad\implies\quad \textbf{(II)}^{-1} = 1+ O(\epsilon)  .
	\end{align*}
	Combining the bounds for $\textbf{(I)}$, $\textbf{(II)}^{-1}$, we obtain:
	\begin{align*}
		\frac{\textbf{(I)}}{\textbf{(II)}} &= \left(1+O(\delta\epsilon^{-1}+\epsilon)\right)\bigl(1+ O(\epsilon)\bigr) = 1 + O(\delta\epsilon^{-1}+\epsilon).
	\end{align*}
	Since the above bounds hold for all $i$, i.e. for all the zeros, and since $\hatk(\omega)=k(\omega)$ since $|\omega-Z_{\fft}|>\delta$, we have from \eqref{equ:initial_formula} that
	\begin{align*}
		\fft_\epsilon(\omega)=(-1)^{\hatk(\omega)}\exp\left(\int_{\hatomega_\epsilon} \frac{\partial_{1} \Psi(\xi,\xi)}{\Psi(\xi,\xi)} d\xi\right) &= \left(1+ O(\delta\epsilon^{-1}+\epsilon)\right)^{k(\omega)}\fft(\omega) .
	\end{align*}
	Utilizing $k(\omega)\leq \omega \leq \hinv$ gives $\left(1+ O(\delta\epsilon^{-1}+\epsilon)\right)^{k(\omega)}=1+ O \bigl(\hinv(\delta\epsilon^{-1}+\epsilon) \bigr)$, i.e. $|\fft_\epsilon(\omega)-f(\omega)| \lesssim \hinv\|f\|_\infty(\delta\epsilon^{-1}+\epsilon)$ for $\delta\epsilon^{-1}+\epsilon \lesssim h$, proving \eqref{equ:error_intermediate}.
    
    Finally, we check that \eqref{equ:error_intermediate} holds when $\omega \notin \hatomega_\epsilon$, i.e. when there exists an estimated zero $\hatxi_i$ such that $|\omega - \hatxi_i|<\epsilon$. Note by definition of the estimator in \eqref{equ:GeneralizedKotlarskiIntermediate}, we have
    \[ \fft_\epsilon(\omega) = \fft_\epsilon(\hatxi_i^-)  ,\]
    but since $\hatxi_i^-$ is distance $\epsilon$ from all estimated zeros, we have by the previous analysis that
    \[ |\fft_\epsilon(\hatxi_i^-) - \fft(\hatxi_i^-)| \lesssim \hinv\|f\|_\infty(\delta\epsilon^{-1}+\epsilon) . \]
    Furthermore, Taylor expanding $\fft$ about $\omega$ gives
    \[ \fft(\hatxi_i^-)  - \fft(\omega)=(\fft)'(\omega)O(\hatxi_i^- -\omega)=\|f\|_\infty O(\epsilon)  ,\]
    (see proof of Lemma \ref{PSderivbounds} for last equality). Thus:
    \begin{align*}
        |\fft_\epsilon(\omega)-\fft(\omega)| &= |\fft_\epsilon(\hatxi_i^-) - \fft(\omega)| \\
        &\leq |\fft_\epsilon(\hatxi_i^-) - \fft(\hatxi_i^-)|+|\fft(\hatxi_i^-) - \fft(\omega)| \\
        &\lesssim \hinv\|f\|_\infty(\delta\epsilon^{-1}+\epsilon) ,
    \end{align*}
    giving \eqref{equ:error_intermediate}.
\end{proof}

\begin{proof}[Proof of Lemma \ref{lem:error_from_integrand}]
Recall the estimator
\begin{align*}
\hatfft_\epsilon(\omega) &= (-1)^{\hatk(\omega)}\exp\left(\int_{\hatomega_\epsilon}\frac{\partial_{1} \hatPsi(\xi,\xi)}{\tildePsi(\xi,\xi)} d\xi \right).
\end{align*}
A straightforward modification of Theorem~\ref{thm:FTDeviation} which invokes Assumptions~\ref{assumption:FT_decay} and \ref{assumption:nice_oscillations} to bound terms involving $|\Psi|$ implies that there exist positive universal constants $c_1,c_2$ such that, for any $\tau \ge 1$, it holds with probability at least $1-c_1 e^{-c_2 \tau}$ that
\begin{align}
\label{equ:bound_on_integrand}
	\sup_{\xi \in \hatomega_{\epsilon}}\left|\frac{\partial_{1} \Psi(\xi,\xi)}{\Psi(\xi,\xi)} - \frac{\partial_{1} \hatPsi(\xi,\xi)}{\tildePsi(\xi,\xi)} \right| &\lesssim 
	\frac{\rho_N h^{-3\beta} }{ \epsilon^3 \|f\|_\infty^2}  ,
\end{align}
as long as $ \rho_N \lesssim \epsilon^2 \|f\|_\infty^2 h^{2\beta}$, which holds by the assumption on $\rho_N$.
We omit the details for brevity, as the calculation is nearly identical. For notational convenience, define 
    \[E(\xi) := \frac{\partial_{1} \hatPsi(\xi,\xi)}{\tildePsi(\xi,\xi)} -\frac{\partial_{1} \Psi(\xi,\xi)}{\Psi(\xi,\xi)} 
    \] 
    so that
    \begin{align*}
       \hatfft(\omega)  &=  \fft_\epsilon(\omega)\exp\left(\int_{\hatomega_\epsilon}E(\xi) d\xi\right).
       \end{align*}
 It follows that
 \begin{align*}
       |\hatfft(\omega) -\fft_\epsilon(\omega)| &\leq |\fft_\epsilon(\omega)|\cdot\left|\exp\left(\int_{\hatomega_\epsilon}E(\xi) d\xi\right)-1\right|\\
       &\leq 2\|f\|_\infty  \cdot\left|1+\omega\right|^{-\beta}\cdot\left|\int_{\hatomega_\epsilon}E(\xi) d\xi\right| \lesssim \frac{\rho_N h^{-2\beta-1} }{ \epsilon^3 \|f\|_\infty}
    \end{align*}
    with probability at least $1-c_1 e^{-c_2 \tau}$,
    where we have utilized $|e^z-1|\leq 2|z|$ for $|z|\leq 1$ for the next to last inequality and \eqref{equ:bound_on_integrand} for the last one. Note $\rho_N \lesssim \epsilon^3 \|f\|^2_\infty h^{3\beta+1}$ guarantees that $|z| = \left| \int_{\hatomega_\epsilon}E(\xi) d\xi\right| \leq 1$.
\section{Supplementary Materials for Section \ref{sec:numerics}}\label{appendix:numerics}

\begin{figure}
    \centering
    \includegraphics[width=1.0\linewidth]{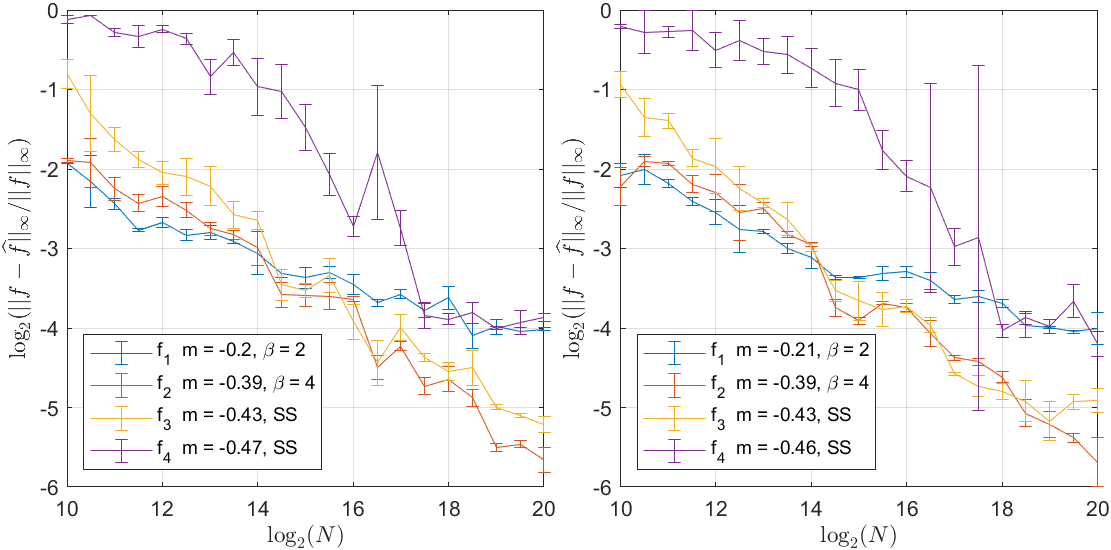}
    \caption{Error decay with varying sample size for fixed $\sigma =0.5$, $\lambda = 0.1$; slopes $m$ for least-squares fit shown in legend. Left: shifts from $\zeta_1$ (uniform). Right: shifts from $\zeta_2$.}
    \label{fig:nerr_sigma.5}
\end{figure}
\begin{remark}\label{rmk:numerics_padding}
    After shifting and noise corruption, our samples are then zero padded by a factor of $ \mathcal{B}=10 $, so that the samples in space live on $ [-20,20] $ with a sample rate of $ 2^{-5} $. This is done so that in frequency, our samples, which are supported on $ [-2^5\pi, 2^5\pi] $, are sampled at a rate of $ \pi/2 \mathcal{B} $ instead of $ \pi/2 $. This increases the fidelity of the discrete gradient.
\end{remark}

\begin{remark}\label{rmk:numerics_regconst}
    There is a minor yet important difference between theory and numerics in the regularization of the estimate of $\Psi$ used in Algorithms \ref{alg:mainalgorithm} and \ref{alg:zerosalgorithm}. In the numerics, we define 
    \[
    \tildePsi(u,v) := \frac{\hatPsi(u,v)}{1 \wedge r\cdot \sqrt{M} |\hatPsi(u,v) |} ,
    \]
    where $r$ is a regularization constant that is set to 1 in the theory. In practice, values of $r < 1$ yield better recovery for the critical lower frequencies in the high-noise regime, but it comes at the cost of bias in higher frequencies, as the recovered signal's frequency component does not decay to 0 (see bottom row of Figure \ref{fig:signals}). However, since the deconvolution step of Algorithms \ref{alg:mainalgorithm} and \ref{alg:zerosalgorithm} truncates frequency components at $h^{-1}$, this bias is minimally represented in the final recovery if $h$ is chosen appropriately. Tuning $r$ can be difficult and signal dependent however. In Figures \ref{fig:nerr_sigma.5} and \ref{fig:nerr_sigma1}, $r = 0.01$ for $f_1,f_2,f_3$ and $r=0.0001$ for $f_4$, while for Figures \ref{fig:lambda}, \ref{fig:sigma}, \ref{fig:sigma_nchange}, $r= 0.001$ for $f_1,f_2,f_3$ and $r=0.0001$ for $f_4$. Increased regularization is especially necessary for Algorithm \ref{alg:zerosalgorithm}, or for any signal $f$ where $\fft(0)$ is near 0. 
\end{remark}

\begin{remark}\label{rmk:numerics_sig.5}
Figure \ref{fig:nerr_sigma.5} is identical in methodology to Figure \ref{fig:nerr_sigma1}, except that the experiment was performed at a noise intensity of $\sigma = 0.5$ instead of $\sigma = 1$. This experiment is thus not in the high noise regime, but better accents the effect of $\beta$ on the average slope of error decay for the four signals. Indeed, the average slope increases from $-.2$ to $-.39$ to $-.43$ as one moves from $f_1$ to $f_2$ to $f_3$. Of interest is that the average slope for $f_4$ is actually steepest in this experiment, contrary to Figure \ref{fig:nerr_sigma1}, where it is the least steep. However, if the plot was continued for samples larger than $2^{20}$, it is likely that the slope would notably decrease, as beyond sample sizes of roughly $\log_2(N) = 18$, the error is likely dominated by the bias introduced by the small regularization constant of $r=0.0001$ used for $f_4$ in this experiment, see Remark \ref{rmk:numerics_regconst}. Additionally, observe the large error bars for $f_4$; if the threshold set in Algorithm \ref{alg:zerosalgorithm} is too sensitive and the zeros are not found, then the recovery relative error is near 1, even if the absolute value of the signal is estimated well. Thus, a single failure to estimate zeros can cause a large sample standard deviation for a set sample size.
\end{remark}

\begin{figure}
    \centering
    \includegraphics[width=0.5\linewidth]{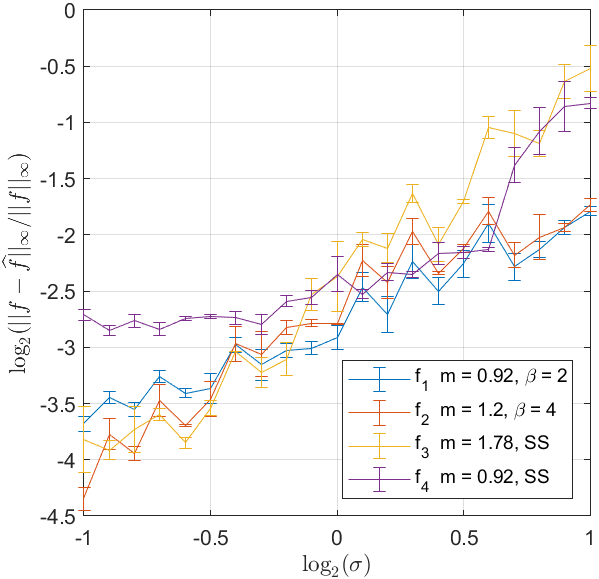}
    \caption{Recovery error as a function of noise intensity $\sigma$, with $N = 100,000$ fixed, $\lambda = 0.1$ fixed. Slopes $m$ for least-squares fit shown in legend.}
    \label{fig:sigma}
\end{figure}
\begin{remark}\label{rmk:numerics_sigma_nfix}
In Figure \ref{fig:sigma}, we kept $N=100,000, \lambda = 0.1$ fixed but varied $\sigma$ evenly in $\log_2$ space from $2^{-1}  =0.5$ to $2^1 = 2$. Taking the log of the error bound in Corollary \ref{corr:largesigma} implies that the constant preceding $\log\sigma$ of the bounding term should increase as $\beta$ increases. Indeed, this is exactly what we observe in Figure \ref{fig:sigma}; as $\beta$ increases from 2 to 4 to the super smooth case, the average slope increases from 0.92 for $f_1$ to 1.2 for $f_2$ and finally 1.78 for $f_3$, numerically verifying Corollary \ref{corr:largesigma} and confirming the conclusions drawn from Figure \ref{fig:sigma_nchange}.
\end{remark}

\begin{remark}\label{rmk:numerics_signalrec}
The first row of Figure \ref{fig:signals} is identical to Figure \ref{fig:signalsintro}, just with the noisy shifted samples omitted. The second row of Figure \ref{fig:signals} shows the Fourier transforms of the signals $f_1$ through $f_4$, with only the real parts shown for $f_3$ and $f_4$. All recoveries were performed with $\sigma = 1, \lambda = 0.1, N = 100,000$. Figure \ref{fig:signals} demonstrates especially well how as $\beta$ increases up to the super smooth case, recovery difficulty decreases. This is at least partially due to the fact that as $\beta$ increases, there are generally fewer high frequency components to recover. This is key, as after the regularization process kicks in, subsequent small high frequency components are lost. Thus, one has to decide whether to preserve high fidelity recovery of lower frequency components at the expense of higher frequency components or accept a larger error from noise in the whole recovery. This is especially obvious in the bottom right subplot of Figure \ref{fig:signals}, where a small regularization constant of $0.0001$ was used for $f_4$, introducing a large bias in frequency that is nonetheless truncated away in space. 

\end{remark}

\begin{remark}\label{rmk:numerics_hfix}
The left panel of Figure \ref{fig:hfixsig1} is identical to the left panel of Figure \ref{fig:nerr_sigma1}, with $h$ optimized dynamically at each sample size, while the right panel of Figure \ref{fig:hfixsig1} is identical in methodology to the left panel, but with fixed deconvolution bandwidths $h= 0.035$ for $f_1,f_2,f_3$, and $h= 0.05$ for $f_4$. Both plots were recovered from samples shifted by $\zeta_1$, the uniform distribution. Different fixed bandwidths were chosen for the signals recovered from Algorithm \ref{alg:mainalgorithm} versus Algorithm \ref{alg:zerosalgorithm} in the right panel of Figure \ref{fig:hfixsig1} due to regularization; $r = 0.01$ for $f_1,f_2,f_3$ and $r= 0.0001$ for $f_4$. Thus, due to the increased error from bias for $f_4$ recovery, a larger $h$ and thus smaller bandwidth $h^{-1}$ is desired. Figure \ref{fig:hfixsig1} demonstrates that although we performed all experiments in Section \ref{sec:numerics} with $h$ optimized dynamically, in general this is not necessary to guarantee good recovery: the error is generally comparable between $h$ fixed and $h$ optimized across all model parameters, here for example $N$. We chose to dynamically optimize $h$ in Section \ref{sec:numerics} to minimize the impact of $h$ as a model parameter and focus on the dependence on the other parameters $N,\beta, \sigma, \lambda$.
\end{remark}

\begin{figure}
    \includegraphics[width=\textwidth]{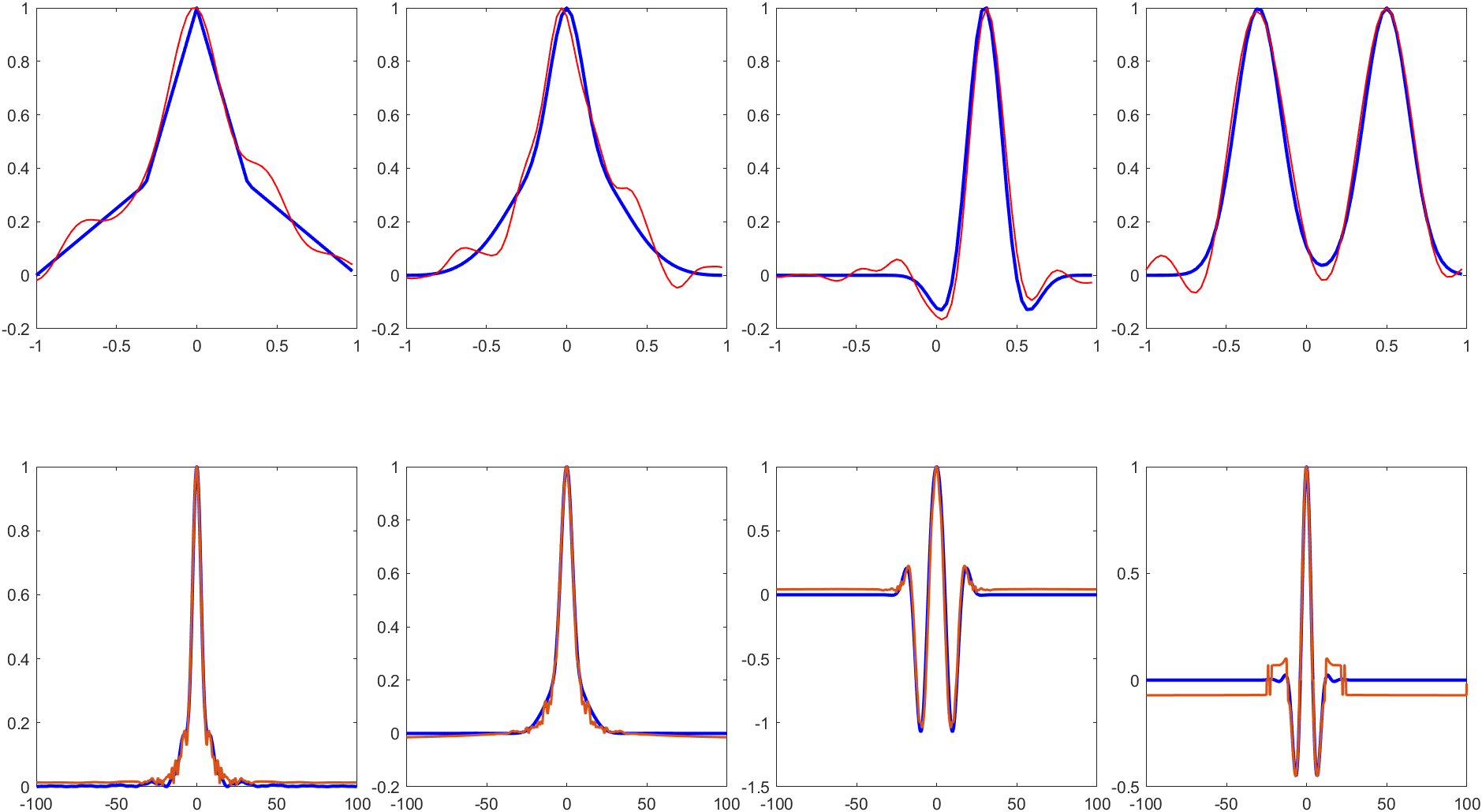}
    \caption{In thick blue: true signals in space and frequency. In thin red, their recoveries following Algorithms \ref{alg:mainalgorithm} and \ref{alg:zerosalgorithm}. Top row: $f_1, f_2,f_3,f_4$. Bottom row $\fft_1, \fft_2,\Re\fft_3,\Re\fft_4$. $\sigma = 1, \lambda = 0.1, N = 100,000$ for all plots.}\label{fig:signals}
\end{figure}
\begin{figure}
    \centering
    \includegraphics[width=1.0\linewidth]{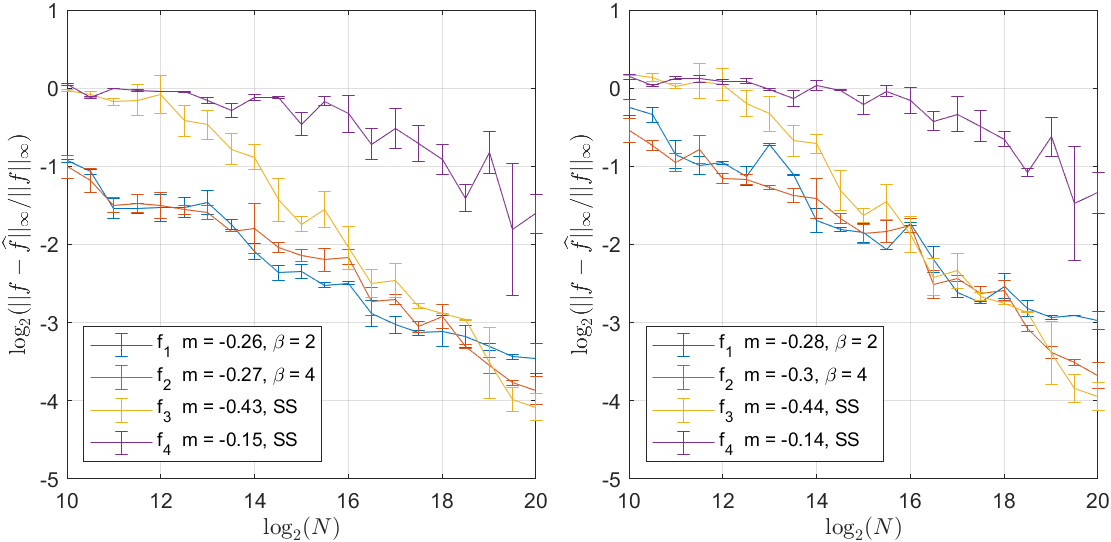}
    \caption{Error decay with varying sample size for fixed $\sigma =1$, $\lambda = 0.1$; slopes $m$ for least-squares fit shown in legend; all samples shifted by $\zeta_1$, the uniform distribution. Left: $h$ optimized dynamically at each sample size. Right: $h$ fixed at $h=0.035$ for $f_1,f_2,f_3$ and $h=0.05$ for $f_4$}
    \label{fig:hfixsig1}
\end{figure}

\subsection{Comparison to spectral algorithm}\label{ssec:suppspectral}
\begin{algorithm}[htp]
\caption{\label{alg:samplingdiscrete}Generating Discrete MRA Samples}
\begin{algorithmic}[1]
     \STATE {\bf Input:} True signal $f$; $N$ number of samples; $J$ integer; $dx=1/J$ spatial sampling rate; noise variance $\sigma^2$.
    \STATE {\bf Sample shifts:} Uniformly sample $N$ integers with replacement from $[-J, J]$ to get $N$ shift samples $\zeta_n, n \in \{1,\ldots,N\}$. 
    \STATE {\bf Initialize signal $f$:} Evaluate $f$ on $[-2,2)$ at a sampling rate of $dx$ to obtain a vector of length $4J$, $f_{vec}\in\R^{4J}.$
    \STATE {\bf Shift signal:} For each $n \in \{1,\ldots,N\}$, circularly shift $f_{vec}$ by $\zeta_n$ to obtain a vector $y_n$.
    \STATE {\bf Add noise:} For each $n \in \{1,\dots,N\}$, sample $\varepsilon_n \sim \mathcal{N}(0,\sigma^2 I)$ and set $y_n \leftarrow y_n + \varepsilon_n$.
\end{algorithmic}
\end{algorithm}

\begin{algorithm}[htp]
\caption{\label{alg:samplingfunctional}Generating Functional MRA Samples}
\begin{algorithmic}[1]
     \STATE {\bf Input:} True signal $f$; $N$ number of samples; $J$ integer; $dx=1/J$ spatial sampling rate; $k_\eta$ covariance function.
    \STATE {\bf Sample shifts:} Uniformly sample $[-1,1]$ to get $N$ shift samples $\zeta_n, n \in \{1,\ldots,N\}$.
    \STATE {\bf Initialize covariance matrix:} Let $\{x_i\}_{i=1}^{4J}$ be the uniform grid on $[-2,2)$ with spacing $dx$. Form the covariance matrix $\Sigma \in \mathbb{R}^{4J\times 4J}$ with entries $\Sigma_{ij} = k(x_i,x_j)$.
    \STATE {\bf Shift signal:} For each $n \in \{1,\ldots,N\}$, 
    obtain vector $y_n$ by evaluating $f$ on $[-2+\zeta_n, 2+\zeta_n)$ at sampling rate $dx$.
    \STATE {\bf Add noise:} For each $n \in \{1,\dots,N\}$, sample $\varepsilon_n \sim \mathcal{N}(0,\Sigma)$ and set $y_n \leftarrow y_n + \varepsilon_n$.
\end{algorithmic}
\end{algorithm}

\begin{figure}[tb]\label{fig:spec_comp_f4}
    \centering
    \begin{subfigure}[b]{0.47\textwidth}
            \includegraphics[width=\linewidth]{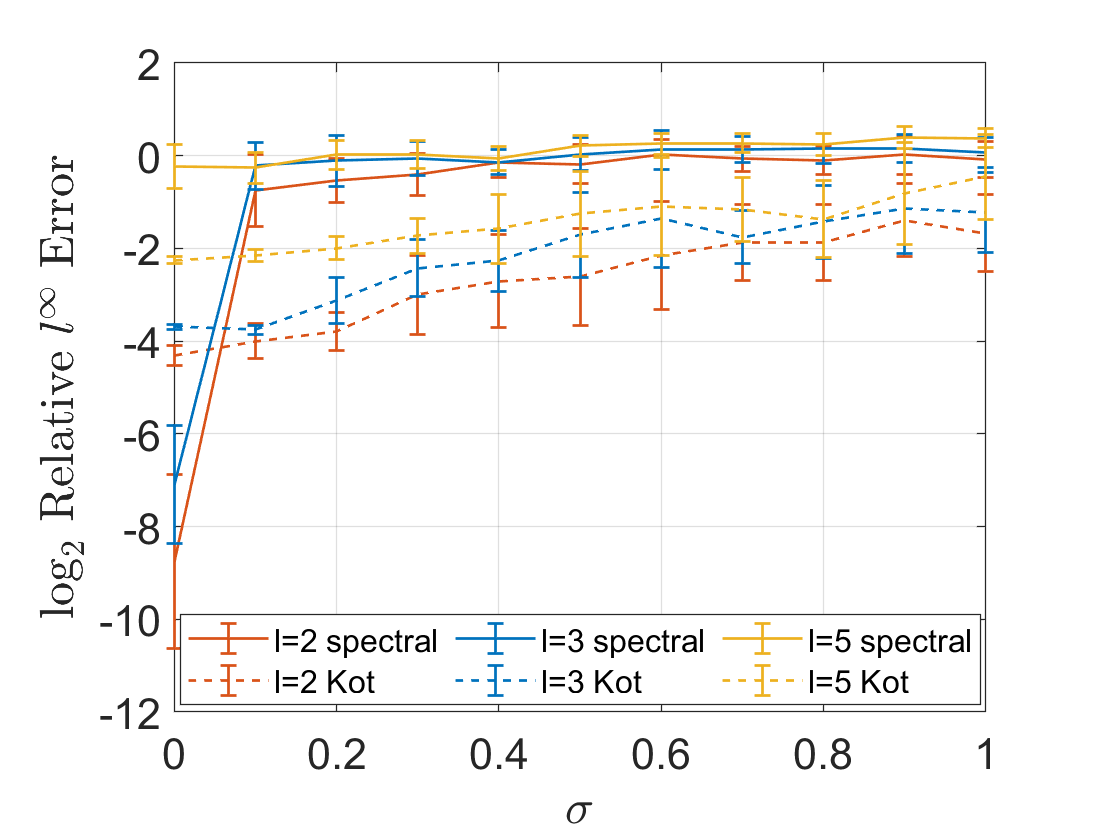}
        \caption{Isotropic noise, shifts on grid.}
        \label{fig:theirs_spec_comp_f4}
    \end{subfigure}    
    \hfill
    \begin{subfigure}[b]{0.47\textwidth}
        \includegraphics[width=\linewidth]{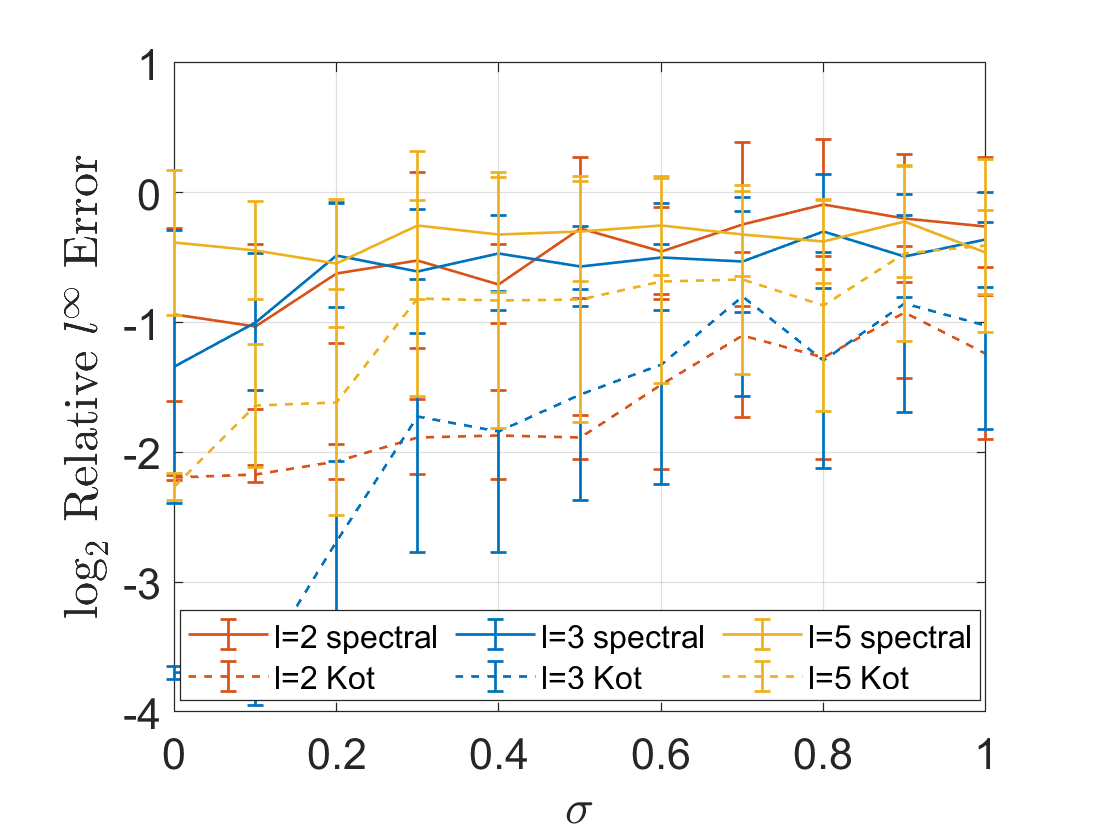}
        \caption{Squared exponential noise, continuous shifts.}
        \label{fig:ours_spec_comp_f4}
    \end{subfigure}
    \caption{Comparison of recovery for $f_4$ using the spectral algorithm and Algorithm \ref{alg:zerosalgorithm} across sampling rates $2^{-l}$ and noise levels with $N=10,000$ fixed. Note the vertical axes differ across panels. }
\end{figure}

\begin{remark}\label{rmk:numerics_spec_f4}
    In Figures \ref{fig:theirs_spec_comp_f4} and \ref{fig:ours_spec_comp_f4} we repeated the experiments in Figures \ref{fig:theirs_spec_comp} and \ref{fig:ours_spec_comp} for the signal $f_4$. All parameters were held exactly the same as in Subsection \ref{sec:comp_to_spectral}. Although the spectral algorithm cannot theoretically handle a vanishing Fourier transform, with a coarse enough sampling rate in space and low noise, the spectral algorithm can recover $f_4$.
\end{remark}

\end{proof}
\end{appendix}

\end{document}